\documentclass[12pt]{article}
\usepackage{graphicx}
\usepackage{color}
\usepackage{amsmath}
\usepackage{amssymb}
\usepackage{amscd}
\usepackage{amsthm}
\usepackage{amsopn}
\usepackage{xspace}
\usepackage{verbatim}
\usepackage{amsmath}
\usepackage{amsfonts}
\usepackage{euscript}
\usepackage{epic}
\usepackage{eepic}
\usepackage{stmaryrd}
\usepackage{empheq}
\usepackage[active]{srcltx}
\usepackage[a4paper,margin=3cm]{geometry}
\definecolor{red}{rgb}{1,0,0}
\definecolor{green}{rgb}{0,1,0}
\definecolor{SeaGreen}{RGB}{46,139,87}
\definecolor{Maroon}{RGB}{128,0,0}

\newcommand{\Supp}{\textrm{Supp~}}
\newcommand{\N}{\mathbb{N}}

\newcommand{\C}{{\mathbb{C}}}
\newcommand{\R}{{\mathbb{R}}}

\newcommand{\A}{\mathcal A}
\newcommand{\B}{\mathcal B}
\newcommand{\CC}{\mathcal C}

\def\Dg {{\mathcal D}}

\newcommand{\FF}{\EuScript F}
\def\Eg {{\mathcal E}}
\def\Mg {{\mathcal M}}
\def\Gg {{\mathcal G}}

\def\Jg {{\mathcal J}}

\def\Kg {{\mathcal K}}
\newcommand{\LL}{\mathcal L}

\newcommand{\OO}{\mathcal O}

\def\PP{\mathcal P}
\newcommand{\QQ}{\mathcal Q}
\newcommand{\RR}{\mathcal R}
\def\Rg {{\mathcal R}}
\def\Sg {{\mathcal S}}
\def\Se {\EuScript{S}}
\def\Ug {{\mathcal U}}

\def\Vg {{\mathcal V}}
\def\Tg {{\mathcal T}}
\def\Ai{\text{\rm Ai\,}}

\def\det{\text{\rm det\,}}

\def\sign{\text{\rm sign\,}}

\newcommand{\supp}{\rm supp\,}

\renewcommand {\Re}{{\rm Re\,}}
\renewcommand{\Im}{{\rm Im\,}}
\newcommand {\pa}{\partial}

\newcommand {\ar}{\to}

\def \jg{\mathfrak j}
\def\0{\mathbf  0}

\newcommand{\eq}{\begin{equation}}
\newcommand{\eeq}{\end{equation}}

\def\XXint#1#2#3{{\setbox0=\hbox{$#1{#2#3}{\int}$ }
\vcenter{\hbox{$#2#3$ }}\kern-.6\wd0}}

\newcommand{\Union}{\mathop{\bigcup}\limits}
\newcommand \ego {\mathfrak e}

 

\errorcontextlines=0 \numberwithin{equation}{section}
\theoremstyle{plain}
\newtheorem{theorem}{Theorem}[section]
\newtheorem{lemma}[theorem]{Lemma}
\newtheorem{proposition}[theorem]{Proposition}

\newtheorem{remark}[theorem]{Remark}
\newtheorem{corollary}[theorem]{Corollary}

\title{The spectrum of a Schr\"odinger operator in a wire-like domain  with a purely imaginary
  degenerate potential in the semiclassical limit} 
  \author{ Y. Almog, Department of
  Mathematics, \\ Louisiana State University,\\ 
    Baton Rouge, LA 70803, USA,\\~\\
  and \\~\\
\noindent   B. Helffer, Laboratoire de Math\'ematiques Jean Leray, \\CNRS and Universit\'e de Nantes, \\
  2 rue de la Houssini\`ere, 44322 Nantes Cedex France\\
   and Laboratoire de Math\'ematiques d'Orsay, Universit\'e Paris-Sud.}

\date{}

\date{}

\begin{document}
\bibliographystyle{siam}

\maketitle

\begin{abstract}
  Consider a two-dimensional domain shaped like a wire, not
  necessarily of uniform cross section. Let $V$ denote an electric
  potential driven by a voltage drop between the conducting surfaces
  of the wire. We consider the operator $\A_h=-h^2\Delta+iV$ in the
  semi-classical limit $h\to0$. We obtain both the asymptotic behaviour
  of the left margin of the spectrum, as well as resolvent estimates
  on the left side of this margin. We extend here previous results
  obtained for potentials for which the set where the current (or $\nabla
  V$) is normal to the boundary is discrete, in contrast with the
  present case where $V$ is constant along the conducting surfaces.
\end{abstract}
\section{Introduction}
\label{sec:1}

\subsection{Main assumptions}

We consider the operator
\begin{subequations}
  \label{eq:1}
  \begin{equation}\label{eq:1a}
\A_h = -h^2\Delta + i\, V \,,
\end{equation}
defined on
\begin{equation}
  D(\A_h)=\{ \,u\in H^2(\Omega,\C) \,| \, u|_{\partial\Omega_D}=0 \,;\,
  \partial u/\partial\nu|_{\partial\Omega_N}=0 \,\}\,.
\end{equation}
\end{subequations}
In the above, $\Omega\subset\R^2$ denotes a bounded, simply connected domain which has the same
characteristics as in \cite{alhe14,aletal17}. In particular its
boundary $\pa \Omega$ contains two disjoint open subsets $\partial\Omega_D\,$ and $\pa \Omega_N$ such that
$$
\overline{\partial\Omega_D} \cup \overline{\partial\Omega_N}=\pa \Omega\,,
$$
where $\partial \Omega_D$ is a union of two disjoint smooth interfaces on which
we prescribe a Dirichlet boundary condition, and $\partial\Omega_N$ is a union
of two disjoint smooth interfaces on which we prescribe a Neumann
boundary condition. Hence $\overline{\partial\Omega_D} \cap \overline{\partial\Omega_N}$
consists of four points which will be called corners.  The analysis
can be extended to domains, where $\partial\Omega_D$ (and $\partial\Omega_N$) consists of
a greater number of disjoint components. In the interest of simplicity
we shall confine ourselves to the simplest possible case.

 In the context of
superconductivity we may say that $\partial\Omega_D\,$ and $\partial\Omega_N\,$, are
respectively adjacent either to a normal metal or to an insulator. We
denote each connected component of $\partial\Omega_\#$ ($\#\in\{D,N\}$) by a
superscript $i\in\{1,2\}$, i.e.,
\begin{displaymath}
  \partial\Omega_\#=\partial\Omega_\#^1\cup\partial\Omega_\#^2\,, \quad \#\in\{D,N\}\,,\; i\in\{1,2\}\,.
\end{displaymath}
We say that $\pa \Omega$ is of class $C^{n,+}$ for some $n\in \mathbb N$,
if there exists $\check \beta >0$ such that $\pa \Omega$ is of class
$C^{n,\check \beta}$. As in \cite{al08,alhe14,aletal17} we make the
following assumptions on $\partial\Omega$
\begin{equation}
\label{AssR1}
(R1)\,\left\{
\begin{array}{l}
(a) \;  \overline{\pa \Omega_\#} \mbox{ is of class } C^{n,+} 
\mbox{ for } \#\in\{D,N\} \,;\\
(b) \mbox{  near each corner, } \overline{ \pa
\Omega_D }\mbox{  and  } \overline{ \pa \Omega_N}  \mbox{ meet with  an angle of
} \frac \pi 2\,. 
\end{array}\right.
\end{equation}
We define $n$ for each result separately (but always have $n\geq2$).  We
occasionally use the notation $(R1(n))$ to  specify $n$ in the assumption.

 Near the corners, we  assume in addition that there exists a smooth
tranformation, mapping the vicinity of the corner onto a vicinity of
rectangular corner. More precisely
\begin{equation}\label{AssR2}
(R2) \,\left\{
\begin{array}{l}
\mbox{ For each corner } \mathfrak c \,,
 \mbox{   there
exist } R>0 \mbox{ and  an invertible holomorphic function }\\
 \Phi \mbox{ in } B(\mathfrak c , R)\cap \Omega, \mbox{ which
is in addition in } C^{n,+}(\bar \Omega\cap B(\mathfrak c, R)),  \\ \mbox{ such that } \Phi (\mathfrak c) =0\,, 
\, \Phi (B(\mathfrak c, R)\cap \Omega) \subset Q:= \mathbb R_+\times \mathbb R_+\,, \\
\mbox{ and } \Phi (\pa \Omega \cap B(\mathfrak c, R)) \subset  (\mathbb R_+\times \{0\}) \cup ( \{0\}\times \mathbb R_+) \cup \{0\} \,.
\end{array} \right.
\end{equation}
Again we use $(R2(n))$ to specify $n$ in the assumption.

We consider potentials $V\in H^2(\Omega)$ satisfying 
\begin{equation}
  \label{eq:2}
  \begin{cases}
    \Delta V = 0 & \text{in } \Omega\,, \\
    V=C_i &  \text{on } \partial\Omega_{D} ^i\mbox{ for }  i=1,2\,, \\
    \frac{\partial V}{\partial\nu} =0 & \text{ on } \partial\Omega_N \,,
  \end{cases}
\end{equation}
describing a potential drop along a wire.\\
Assumptions $(R1(n))$ and $(R2(n))$ imply that $V\in C^{n,+}(\bar \Omega)$.
Away from the corners, we may rely on Schauder estimates to establish
the desired regularity. In the neighborhood of a corner, we may use
the conformal map given by Assumption $(R2(n))$ to obtain a
problem for $V$ in a right-angled sector. Then we can use a
  reflection argument to establish the announced regularity of $V$
  (cf. \cite{al08,alhe14} for instance). 
 
We assume further, as in \cite{AGH}, that $V$ satisfies
\begin{equation}\label{ass1}
|\nabla V(x)| \neq 0\,,\, \forall x \in \overline{\Omega}\,.
\end{equation}
This implies  that
$$
C_1 \neq C_2\,.
$$
We can indeed follow one component of $\partial\Omega_N$ between two
  corners and observes that the tangential derivative of $V$ never
  vanishes (cf. \cite{al08}).

  The mathematical analysis of Equation \eqref{eq:2} has a very long
  record in the literature. We refer to \cite{PS}, where explicitly
  known solutions, for many simple domains including the square, are
  listed.
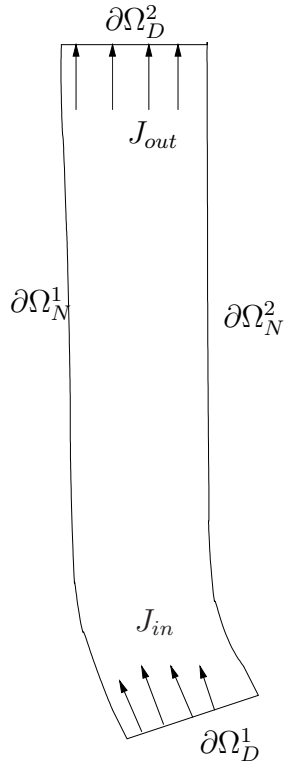
\begin{figure}
  \begin{center}
\setlength{\unitlength}{0.0005in}
\begingroup\makeatletter\ifx\SetFigFont\undefined%
\gdef\SetFigFont#1#2#3#4#5{%
  \reset@font\fontsize{#1}{#2pt}%
  \fontfamily{#3}\fontseries{#4}\fontshape{#5}%
  \selectfont}%
\fi\endgroup%
{\renewcommand{\dashlinestretch}{30}
\begin{picture}(2649,7740)(0,-10)
\put(1140,1290){\makebox(0,0)[lb]{\smash{{\SetFigFont{12}{14.4}{\rmdefault}{\mddefault}{\updefault}$J_{in}$}}}}
\path(915,6690)(915,7365)
\blacken\path(945.000,7245.000)(915.000,7365.000)(885.000,7245.000)(945.000,7245.000)
\path(1590,6690)(1590,7365)
\blacken\path(1620.000,7245.000)(1590.000,7365.000)(1560.000,7245.000)(1620.000,7245.000)
\path(390,7365)(1890,7365)
\path(915,465)(1065,165)(2415,615)(2265,915)
\path(1215,240)(990,765)
\blacken\path(1064.845,666.520)(990.000,765.000)(1009.696,642.885)(1064.845,666.520)
\path(1440,315)(1215,915)
\blacken\path(1285.225,813.174)(1215.000,915.000)(1229.045,792.107)(1285.225,813.174)
\drawline(1515,915)(1515,915)
\path(1740,390)(1515,915)
\blacken\path(1589.845,816.520)(1515.000,915.000)(1534.696,792.885)(1589.845,816.520)
\path(1965,465)(1815,915)
\blacken\path(1881.408,810.645)(1815.000,915.000)(1824.487,791.671)(1881.408,810.645)
\path(1290,6690)(1290,7365)
\blacken\path(1320.000,7245.000)(1290.000,7365.000)(1260.000,7245.000)(1320.000,7245.000)
\path(2265,915)(2264,916)(2263,918)
	(2260,922)(2255,929)(2249,937)
	(2241,948)(2232,962)(2222,977)
	(2210,994)(2198,1013)(2186,1034)
	(2173,1056)(2159,1080)(2146,1105)
	(2132,1133)(2118,1164)(2103,1197)
	(2087,1234)(2072,1274)(2056,1318)
	(2040,1365)(2026,1409)(2014,1450)
	(2003,1488)(1994,1519)(1986,1545)
	(1981,1564)(1976,1578)(1973,1586)
	(1970,1591)(1968,1592)(1967,1591)
	(1966,1590)(1965,1589)(1965,1590)
	(1963,1594)(1962,1602)(1960,1616)
	(1957,1638)(1953,1667)(1949,1707)
	(1944,1757)(1938,1817)(1932,1887)
	(1926,1965)(1922,2020)(1919,2075)
	(1916,2128)(1913,2179)(1911,2225)
	(1909,2266)(1908,2303)(1906,2333)
	(1905,2359)(1905,2379)(1904,2395)
	(1904,2407)(1904,2414)(1904,2419)
	(1904,2421)(1905,2422)(1906,2421)
	(1906,2423)(1907,2427)(1907,2434)
	(1908,2445)(1908,2461)(1908,2482)
	(1908,2511)(1908,2547)(1908,2592)
	(1908,2647)(1907,2711)(1907,2786)
	(1906,2872)(1905,2968)(1904,3075)
	(1903,3191)(1902,3315)(1901,3398)
	(1901,3482)(1900,3567)(1900,3651)
	(1899,3734)(1899,3816)(1898,3896)
	(1898,3974)(1897,4049)(1897,4121)
	(1897,4190)(1896,4255)(1896,4318)
	(1896,4378)(1895,4434)(1895,4488)
	(1895,4539)(1895,4587)(1895,4633)
	(1895,4676)(1894,4717)(1894,4757)
	(1894,4795)(1894,4832)(1894,4868)
	(1894,4903)(1894,4937)(1894,4971)
	(1894,5005)(1894,5039)(1894,5074)
	(1894,5110)(1894,5146)(1894,5183)
	(1894,5222)(1894,5263)(1893,5305)
	(1893,5350)(1893,5397)(1893,5446)
	(1893,5498)(1893,5552)(1893,5609)
	(1893,5669)(1892,5732)(1892,5797)
	(1892,5865)(1892,5936)(1892,6009)
	(1891,6083)(1891,6159)(1891,6236)
	(1891,6314)(1890,6391)(1890,6466)
	(1890,6540)(1890,6600)(1890,6658)
	(1890,6713)(1889,6766)(1889,6816)
	(1889,6863)(1889,6907)(1889,6949)
	(1889,6988)(1889,7025)(1889,7059)
	(1889,7091)(1889,7121)(1889,7148)
	(1889,7174)(1888,7197)(1888,7219)
	(1888,7239)(1888,7257)(1888,7274)
	(1888,7289)(1888,7303)(1888,7315)
	(1889,7327)(1889,7337)(1889,7346)
	(1889,7354)(1889,7361)(1889,7367)
	(1889,7372)(1889,7377)(1889,7381)
	(1889,7384)(1889,7386)(1889,7388)
	(1889,7389)(1889,7390)(1889,7391)
	(1889,7390)(1889,7389)(1889,7388)
	(1889,7386)(1890,7385)(1890,7383)
	(1890,7382)(1890,7380)(1890,7378)
	(1890,7377)(1890,7375)(1890,7374)
	(1890,7372)(1890,7371)(1890,7370)
	(1890,7369)(1890,7368)(1890,7367)
	(1890,7366)(1890,7365)
\path(915,465)(914,466)(913,470)
	(911,475)(907,484)(902,496)
	(896,512)(889,530)(880,552)
	(871,575)(860,601)(849,629)
	(838,658)(826,688)(814,720)
	(802,752)(789,785)(776,820)
	(763,857)(749,895)(735,935)
	(720,978)(705,1021)(690,1065)
	(673,1116)(658,1162)(646,1200)
	(636,1229)(628,1250)(622,1264)
	(619,1271)(616,1274)(615,1273)
	(613,1271)(613,1269)(611,1269)
	(610,1273)(607,1282)(603,1298)
	(598,1322)(591,1356)(583,1400)
	(574,1454)(564,1515)(557,1561)
	(551,1606)(546,1648)(541,1685)
	(537,1716)(534,1741)(531,1761)
	(529,1774)(528,1783)(527,1787)
	(526,1787)(525,1785)(525,1782)
	(525,1777)(525,1774)(525,1773)
	(524,1774)(524,1780)(523,1792)
	(522,1811)(521,1838)(519,1875)
	(517,1923)(514,1982)(511,2054)
	(508,2139)(505,2234)(501,2340)
	(499,2408)(497,2477)(495,2547)
	(493,2616)(491,2683)(490,2748)
	(488,2811)(487,2870)(486,2927)
	(485,2979)(484,3029)(483,3075)
	(483,3118)(482,3157)(482,3194)
	(482,3228)(481,3260)(481,3290)
	(481,3318)(481,3345)(481,3371)
	(481,3396)(481,3421)(481,3446)
	(481,3472)(481,3498)(481,3525)
	(481,3554)(481,3585)(481,3618)
	(480,3654)(480,3692)(480,3734)
	(479,3779)(478,3828)(478,3881)
	(477,3937)(476,3998)(475,4062)
	(474,4130)(473,4202)(471,4276)
	(470,4353)(468,4432)(467,4511)
	(465,4590)(463,4688)(461,4782)
	(459,4869)(457,4950)(455,5022)
	(454,5087)(453,5143)(451,5192)
	(450,5233)(449,5268)(449,5296)
	(448,5319)(447,5337)(447,5352)
	(446,5363)(446,5371)(446,5378)
	(445,5384)(445,5389)(445,5395)
	(445,5402)(444,5410)(444,5421)
	(443,5435)(442,5452)(442,5474)
	(441,5500)(440,5530)(438,5566)
	(437,5608)(436,5654)(434,5705)
	(432,5760)(430,5819)(428,5879)
	(426,5940)(423,6016)(421,6087)
	(418,6151)(416,6207)(414,6255)
	(413,6295)(411,6328)(410,6354)
	(409,6374)(408,6388)(407,6399)
	(406,6406)(405,6411)(405,6415)
	(404,6419)(404,6423)(403,6429)
	(402,6437)(401,6448)(400,6464)
	(399,6484)(398,6509)(397,6540)
	(396,6576)(394,6618)(393,6664)
	(391,6714)(390,6765)(389,6823)
	(388,6878)(387,6929)(386,6976)
	(386,7019)(385,7057)(385,7093)
	(385,7126)(386,7156)(386,7185)
	(386,7211)(387,7236)(387,7259)
	(388,7280)(388,7299)(388,7316)
	(389,7330)(389,7342)(390,7351)
	(390,7357)(390,7362)(390,7364)(390,7365)
\put(840,7515){\makebox(0,0)[lb]{\smash{{\SetFigFont{12}{14.4}{\rmdefault}{\mddefault}{\updefault}$\partial\Omega_D^2$}}}}
\put(2055,4440){\makebox(0,0)[lb]{\smash{{\SetFigFont{12}{14.4}{\rmdefault}{\mddefault}{\updefault}$\partial\Omega_N^2$}}}}
\put(-135,4590){\makebox(0,0)[lb]{\smash{{\SetFigFont{12}{14.4}{\rmdefault}{\mddefault}{\updefault}$\partial\Omega_N^1$}}}}
\put(1065,6340){\makebox(0,0)[lb]{\smash{{\SetFigFont{12}{14.4}{\rmdefault}{\mddefault}{\updefault}$J_{out}$}}}}
\put(1815,15){\makebox(0,0)[lb]{\smash{{\SetFigFont{12}{14.4}{\rmdefault}{\mddefault}{\updefault}$\partial\Omega_D^1$}}}}
\path(540,6690)(540,7365)
\blacken\path(570.000,7245.000)(540.000,7365.000)(510.000,7245.000)(570.000,7245.000)
\end{picture}
}
  \end{center}
\caption{A typical wire-like domain. The arrows denote the
  direction of the potential gradient (or the current flow: $J_{in}$  for the inlet, and
  $J_{out}$ for the outlet).}
  \label{fig:1}
\end{figure}
Figure 1 presents a typical sample with properties (R1) and (R2),
where the current flows into the sample from one connected component
of $\partial\Omega_D$, and exits from another part, disconnected from the first
one. Most wires would fall into the above class of domains.\\
Note that, $V$ being constant on each connected component of $\partial
\Omega_D$, we have
$$
|\nabla V|= |\partial V/\partial\nu|\,\mbox{ on } \pa \Omega_D\,.
$$

We distinguish  in the sequel  between two types of potentials satisfying
\eqref{eq:2}.
\begin{description}
\item[V1] Potentials for which all points where
  $\inf_{x\in\overline{\partial\Omega_D}}|\partial V/\partial\nu|$ is attained, lie  in
   $\partial\Omega_D$.
 \item[V2] Potentials for which all points where $\inf_{x\in
     \overline{\partial\Omega_D}}|\partial V/\partial\nu|$ is attained are corners.
\end{description}
In appendix \ref{app:example-potent} we present examples corresponding
to both cases. While other cases could be treated by the same
techniques, we limit ourselves to these two cases in the interest of
simplicity.

The spectral analysis of a Schr\"odinger operator with a purely
imaginary potential has several applications in mathematical physics,
among them are the Orr-Sommerfeld equations in fluid dynamics
\cite{sh03}, the Ginzburg-Landau equation in the presence of electric
current (when magnetic field effects are neglected)
\cite{al08,ivko84}, the null controllability of Kolmogorov type
equations \cite{BHHR}, and the diffusion nuclear magnetic resonance
\cite{Stoller91,deSwiet94,Grebenkov07}. In the present contribution we
focus on the Ginzburg-Landau model, in the absence of magnetic field,
and choose an electric potential satisfying \eqref{eq:2}. Such a
potential was considered in \cite{al08} where the asymptotics of a
lower bound of $\inf\Re\sigma(\A_h)$ have been obtained as $h\ar 0$.
Assuming a smooth domain, a similar result has been established in
\cite{Hen}, using a more constructive technique, which is employed in
the present contribution as well, providing resolvent estimates in
addition to the above lower bound.

In \cite{AGH}, improving previous results from \cite{Ahen}, we
obtained in collaboration with D.~Grebenkov, the asymptotic behaviour
of an upper bound for $\inf\Re\sigma(-h^2\Delta+iV)$ on smooth bounded domains
in $\R^d$. To characterize the potentials addressed in \cite{AGH} we
first define (for $d=2$, which is the case considered in this work)
\begin{displaymath}
  \partial\Omega_\perp= \{x\in\partial\Omega\,| \, {\det}( \nabla V(x), \vec{\nu} (x)) =0 \}\,, 
\end{displaymath}
where $\vec{\nu}(x) $ denotes the outward normal at $x$. Then it is required in
\cite{AGH} that
\begin{displaymath}
 \inf_{x\in\partial\Omega_\perp}\big|\text{\rm det } D^2V_\partial(x)\big|> 0\,,
\end{displaymath}
where $V_\partial$ denotes the restriction of $V$ to $\partial\Omega$, and $D^2V_\partial$
denotes its Hessian matrix.  Note that $\partial\Omega_\perp$ must be a discrete
set in that case. Clearly, such potentials do not belong to the class
of potentials considered in this contribution, as is evident from
\eqref{eq:2}. It will become clear in Sections \ref{sec:2} and
\ref{sec:quazimode-V2} that the techniques employed in \cite{AGH} are not
applicable for potentials satisfying \eqref{eq:2}. The reason is that
the ensuing approximate operators near the boundaries are not
separable. Thus, while electric potentials satisfying \eqref{eq:2}
appear very naturally in applications, their spectral analysis poses a
significant challenge beyond the potentials addressed in \cite{AGH}.

\subsection{Main results}
We seek an approximation for $\inf \Re\sigma(\A_h)$ in the limit $h\to0\,$.
 Let 
\begin{equation}
\label{eq:3}
  J_m= \min_{x\in\partial\Omega_D}|\nabla V(x)| \,. 
\end{equation}
Denote by $\Se$ the set 
\begin{equation}
\label{eq:4}
  \Se=\{x\in\partial\Omega_D\,: \,|\nabla V(x)| = J_m \}\,.
\end{equation}
\subsubsection{ Type V1 potentials}
In this case, any $x\in\Se$ is a minimum point of $|\partial V/\partial \nu|$ on
$\partial\Omega_D$. Thus,
\begin{equation}
  \label{eq:5}
 \partial_\|\partial_\nu V(x) =0\,,\quad \forall x\in\Se\,,
\end{equation}
where $\partial_\|$ represents the derivative with respect to the arclength
along the boundary in the positive trigonometric direction. We next
introduce
  \begin{equation}
\label{eq:6}
  \alpha (x) = \partial^2_\| \partial_\nu V(x)\,, \quad \forall x\in\Se \,.
\end{equation}
Let 
\begin{equation}\label{eq:4a}
  \alpha_m = \min_{x\in\Se}|\alpha(x)| \,.
\end{equation}
We then define a new set
\begin{equation}
\label{eq:7}
\Se^m= \{x\in\Se \,| \,|\alpha(x)|=\alpha_m\,\}\,,
\end{equation}
We assume in the following that
  \begin{equation}
\label{eq:8}
 \alpha_m>0\,.
 \end{equation}
 Consequently any $x\in\Se$ is a non-degenerate minimum point of
 $|\partial V/\partial\nu|$\,. Furthermore we may conclude from \eqref{eq:8} that
 $\Se$ is discrete.

Our main result in this case is the following theorem.
\begin{theorem}
  \label{thm:interior}
  Let $\A_h$ be given by \eqref{eq:1}, in which $\partial\Omega$ satisfies
  (\ref{AssR1}) and (\ref{AssR2}) for $n=4$, and let $V,$ the solution
  of \eqref{eq:2}, be of type V1 and satisfy \eqref{ass1}. Suppose
  further that \eqref{eq:8} is satisfied. Then
\begin{equation}
  \label{eq:9}
 \lim_{h\to0}\frac{1}{h^{2/3}}\inf \bigl\{ \Re\, \sigma(\A_h) \bigr\} =   J_m^{2/3}\frac{|\nu_1|}{2}\,,  
\end{equation}
where $\nu_1<0$ is the rightmost zero of Airy's function.
\end{theorem}
\begin{remark}
  As will become evident in the sequel, some of the conditions set
  above are unnecessary in order to obtain the lower bound on the left
  hand side of \eqref{eq:9}.
\end{remark}
\subsubsection{Type V2 potentials}
In this case we similarly define
\begin{equation}\label{defhatalpha}
  \hat \alpha(x) = \partial_\| \partial_\nu V(x) \,,
\end{equation}
where 
\begin{equation}\label{eq:4av2}
 \hat \alpha_m = \min_{x\in\Se}|\hat \alpha(x)| \,.
\end{equation}
We then define in $\Se$ a new subset
\begin{equation}
\label{eq:8v2}
\hat \Se^m= \{x\in\Se \,| \,|\hat \alpha(x)|=\hat \alpha_m\,\}\,.
\end{equation}
\begin{theorem}
  \label{thm:nonsmooth}
  Let $\A_h$ be given by \eqref{eq:1}, in which $\partial\Omega$ satisfies
  \eqref{AssR1}  and \eqref{AssR2} for $n=4$\,,
   let $V$, the solution of \eqref{eq:2}, be of type V2 and  satisfy
  \eqref{ass1}.  Suppose further that $\hat{\alpha}_m>0\,$. 
 Then
\begin{equation}
  \label{eq:9v2}
 \lim_{h\to0}\frac{1}{h^{2/3}}\inf \bigl\{ \Re\, \sigma(\A_h) \bigr\} =   J_m^{2/3}\frac{|\nu_1|}{2}\,.
\end{equation}
\end{theorem}

The rest of the contribution is arranged as follows. In the next
section we obtain the leading order asymptotic behaviour of a lower
bound of $\inf \Re\sigma(\A_h)$ in the limit $h\to0$. In Sections
\ref{sec:2} and \ref{sec:quazimode-V2} we obtain a quasimode for
$\A_h$ for potentials of type $V1$ and $V2$ respectively. In
Section~\ref{sec:v1-potentials:-1d} we obtain some auxiliary resolvent
estimates in one dimension, that are employed in Section
\ref{sec:simplified}. In Section \ref{sec:simplified} we obtain
resolvent estimates for the approximate operator appearing in
Section \ref{sec:2} for type $V1$ potentials. A similar task is
carried in Section \ref{sec:simplified-V2} for type $V2$ potentials.
In the last section we complete the proof of Theorems
\ref{thm:interior} and \ref{thm:nonsmooth}. Finally, in the appendix
we bring examples of potentials of both types.

\section{Lower bound}
\label{s2}
\subsection{Main statement}

We now state and prove
\begin{proposition}
  \label{prop:lower}
 Let $\Omega$ satisfy \eqref{AssR1}  and \eqref{AssR2} with $n=3$ and $V$ satisfy
  \eqref{ass1}. Then, we have
  \begin{equation}
    \label{eq:10}
\liminf_{h\to0}\frac{1}{h^{2/3}}\inf \bigl\{ \Re\, \sigma(\A_h) \bigr\} \geq  J_m^{2/3}\frac{|\nu_1|}{2}\,.
  \end{equation}
\end{proposition}
The proof differs from the proof of the lower bound in \cite{AGH} only
by the need to estimate the resolvent in the vicinity of the
Dirichlet-Neumann corners. We thus begin by recalling various
  lemmas from \cite {AGH}, and then continue by treating the corner
  case.  

Note that \eqref{eq:10} has already been proved in
\cite{al08}. Nevertheless, the proof brought in this section, is more
constructive and provides resolvent estimates for $\A_h$ in addition
for the lower bound on the spectrum. 
\subsection{Preliminary  lemmas}\label{ss2.2}
The following lemmas all involve an affine approximation of $V$.
\subsubsection{Complex Airy operator in $\mathbb R^2$}
  \begin{lemma} 
\label{lem:preliminary-lemmas-entire}
    Let 
\begin{displaymath}
  \mathcal A_0 = -\Delta+i  x_1\,,
\end{displaymath}
be defined on 
\begin{displaymath}
  D ( \mathcal A_0) = \{ u\in  H^2(\mathbb R^2)\,|\, x_1u \in L^2(\mathbb R^2)\}\,.
\end{displaymath}
Then, for any $\omega >0\,$,  there exists $C_\omega$ such that
\begin{displaymath}  
\sup_{\Re z\leq\omega}\|(\mathcal A_0-z)^{-1}\|\leq C_\omega\,.
\end{displaymath}
\end{lemma}
\begin{remark}
\label{rem:entire}
By dilation, we obtain the same result for $ -\Delta+i \,\jg\, x_1$ for
any $\jg \in \mathbb R\setminus \{0\}$. Hence, we can obtain a uniform bound,  with respect to
$\jg$,  of $\sup_{\Re z\leq\omega}\|(-\Delta+i \,\jg\,x_1-z)^{-1}\|$ on any compact interval excluding $0$.
\end{remark}
\subsubsection{Complex Airy operator in $\mathbb R^2_+$}
The next lemma considers the Neumann problem in $\R^2_+=\R\times\R_+$ which arises while localizing $\A_h$ near
$\partial\Omega_N$. It follows immediately from \cite[Proposition 4.9]{AGH}.
  \begin{lemma}
\label{lem:preliminary-lemmas-half-Neumann}
    Let 
\begin{displaymath}
  \mathcal A_N= -\Delta+ix_1\,,
\end{displaymath}
be defined on 
\begin{displaymath}
  D ( \mathcal A_N) = \{ u\in  H^2(\mathbb R^2_+)\,|\, x_1u \in
  L^2(\mathbb R^2_+) \,,\;\partial u/\partial x_2|_{\partial\R^2_+}=0\}\,.
\end{displaymath}
Then, for any $\omega >0$, there exists $C_\omega$ such that
\begin{displaymath}  
\sup_{\Re z\leq\omega}\|(\mathcal A_N-z)^{-1}\|\leq C_\omega.
\end{displaymath}
\end{lemma}
\begin{remark}
\label{rem:Neumann}
By the same argument of Remark \ref{rem:entire} the resolvent of the
Neumann realization of $ -\Delta+i \,\jg\,x_1$ is uniformly bounded with
respect to $\jg$ on any compact interval excluding $0$.
\end{remark}

We also restate another conclusion of \cite[Proposition 4.9]{AGH} and
\cite[Proposition 4.5]{AGH}, which is related to the localization of
$\A_h$ near $\partial\Omega_D$.
  \begin{lemma}
\label{lem:preliminary-lemmas-half-Dirichlet}
    Let, for $\jg \neq 0$,  
\begin{displaymath}
  \mathcal A_D= -\Delta+i \,\jg \, x_2\,,
\end{displaymath}
be defined on 
\begin{displaymath}
  D ( \mathcal A_D) = \{ u\in  H^2(\mathbb R^2_+)\,|\, x_2u \in
  L^2(\mathbb R^2_+) \,, \;u|_{\partial\R^2_+}=0\}\,.
\end{displaymath}
Then,  there exists $C>0$ such that, for all $0<\epsilon\leq 1$, 
\begin{displaymath}  
\sup_{\Re z\leq |\jg|^\frac 23 |\nu_1|/2-\epsilon}\|(\mathcal
A_D-z)^{-1}\|+\|\nabla(\mathcal A_D-z)^{-1}\|+\|\Delta(\mathcal A_D-z)^{-1}\|
\leq \frac{C}{\epsilon}\,.
\end{displaymath}
Moreover, $C$ may be chosen  independently of $\jg$ if we confine $\jg$ to a
closed bounded interval excluding $0$.
\end{lemma}

We continue  with the following estimate (cf. \cite[Lemma 4.12]{AGH})
which will become useful in Section~\ref{sec:upper}.
\begin{lemma} 
\label{lemma4.12}
With the notation of  Lemma
\ref{lem:preliminary-lemmas-half-Dirichlet}, for any compact interval
$I= [\mu_1,\mu_2] $, there exists a positive $C(I)$ such that, 
for any $z=\mu+i\nu$ with  $|\nu|>\mu_2+4\,$ and $\mu\in I$,
\begin{equation}
  \label{eq:30aa}
\|(\A^D -z )^{-1}\| \leq  C(I) \,.
\end{equation}
\end{lemma}
\subsubsection{Complex Airy operator $\mathbb R_+\times \mathbb R_+$} 
We now present a new result which is useful while using blow-up
analysis to obtain the contribution of the corners to the resolvent of
$\A_h$:
  \begin{lemma}
    Let $\A_c$ denote the operator
    \begin{subequations}
\label{eq:11}
          \begin{equation}
      \A_c= -\Delta +i\, \jg \, x_1
    \end{equation}
   defined on
\begin{equation}
  D(\A_c)= \{ u\in H^2(Q) \, | \, u_{\partial Q_\|}=0 \,;\,
  \partial_\nu u_{\partial Q_\perp}=0\, ;\,x_1\, u\in L^2(Q)\}\,,
\end{equation}
    \end{subequations}
where $\jg \neq 0$ and 
\begin{displaymath}
  Q=\R_+\times\R_+ \quad ; \quad \partial Q_\perp=\R_+\times\partial\R_+ \quad ; \quad
  \partial Q_\|=\partial\R_+\times\R_+ \,.
\end{displaymath}
Then,  there exists $C>0$ such that, for any $\epsilon>0$, 
\begin{equation}
\label{eq:12} 
  \sup_{\Re\lambda \leq \, |\jg|^{\frac 23}\, |\nu_1|/2-\epsilon}\| (\A_c-\lambda)^{-1}\|\leq \frac C \epsilon \,.
\end{equation}
Moreover, $C$ may be chosen  independently of $\jg$ if we confine its value to a
closed bounded interval excluding $0$.
  \end{lemma}
  \begin{proof}
    It can be easily verified that $\A_c:D(\A_c)\to L^2(Q)$ is
    surjective, injective, and maximally accretive. This can be done
    either by the separation of variable technique as  in \cite{AGH} or by
    using generalized Lax-Milgram lemma from \cite{AH}. Note that by the
    arguments presented in \cite[Proposition A.3]{alhe14} 
    there exists $C>0$ such that for every $u\in D(\A_c)$ and $0<r_1<r_2\,$,
    \begin{displaymath}
      \|u\|_{H^2(D_{r_1})} \leq C\,  \left(\|\Delta u\|_{L^2(D_{r_2})}+ \|u\|_{L^2(D_{r_2})} \right)\,,
    \end{displaymath}
where $D_r=B(0,r)\cap\QQ$. 
Hence, the presence of a corner does not
pose a significant obstacle on the way to obtain global regularity
estimates.\\

To prove \eqref{eq:12} we write
\begin{displaymath}
  \A_c = \LL_+ -\partial^2_{x_2} \,,
\end{displaymath}
where $\LL_+$ is the Dirichlet realization in $\R_+$ of 
\begin{displaymath} 
  \LL_+=-\frac{d^2}{dx_1^2}+ i\,\jg\, x_1 \,,
\end{displaymath}
and $-\partial^2_{x_2} $ denotes by abuse of notation the Neumann realization of $-\frac{d^2}{dx_2^2}$ in $\R_+$.\\
Since $\LL_+$ and $-\partial^2_{x_2}$ commute, we have
\begin{displaymath}
  e^{-t\A_c}= e^{t\partial^2_{x_2} }\otimes e^{-t\LL_+} \,,
\end{displaymath}
and hence
\begin{equation}
\label{eq:13}
  \|e^{-t\A_c}\|\leq  \|e^{-t\LL_+}\|\,.
\end{equation}
From \cite[Proposition 2.4]{AGH} we learn that
\begin{displaymath}
  \|e^{-t\LL_+}\| \leq C\,e^{-t \, |\jg|^\frac 23\, |\nu_1|/2}\,,
\end{displaymath}
and hence by \eqref{eq:13} we have 
\begin{displaymath}
    \|e^{-t\A_c}\|\leq  C\, e^{-t  \, |\jg|^\frac 23\,  |\nu_1|/2}\,.
\end{displaymath}
Hence, whenever $ \Re\lambda<|\jg|^{\frac 23} \, |\nu_1|/2$ we have 
\begin{displaymath}
  \| (\A_c-\lambda)^{-1}\| \leq \int_0^\infty\|e^{-t(\A_c-\Re\lambda)}\|\,dt \leq
  \frac{C}{|\jg|^\frac 23\, |\nu_1|/2-\Re\lambda} \,.
\end{displaymath}
  \end{proof}
Finally,  we shall need, in the last section,  the following lemma,
which is analogous to Lemma \ref{lemma4.12},
\begin{lemma} 
\label{lemma4.12-corner}
For any compact interval $I= [\mu_1,\mu_2] $, there exists a positive $C(I)$ such that,
for any $z=\mu+i\nu$ with  $|\nu|>\mu_2+4\,$ and $\mu\in I$,
\begin{equation}
\label{eq:14}
\|(\A_c-z )^{-1}\| \leq  C(I) \,.
\end{equation}
\end{lemma}
\begin{proof}
  To obtain a resolvent estimate, we use an even extension in $s$, i.e.,
we define the operator $\A_{c}^e$ which is associated with
the same differential operator as $\A_{c}$ but whose domain is
defined by
\begin{displaymath}
  \Dg(\A_{c}^e) = \{ u\in H^2(\R^2_+) \, | \,  u_{\partial\R^2_+}=0 \,;\,
\Rg u =u  \,;\, x_1 u\in L^2(\R^2_+)  \}\,,
\end{displaymath}
where $\Rg$ denotes the reflection $x_2\to-x_2$. The lemma then follows
immediately from Lemma~\ref{lemma4.12}.
\end{proof}

 \subsection{Proof of Proposition \ref{prop:lower}}
 The proof is similar to the derivation of the lower bound in
 \cite[Section 6]{AGH} or \cite[Section 4]{Hen}. We thus, recall the
 main steps only, focusing primarly on the resolvent estimates near
 the corners (that are absent from \cite{AGH,Hen}).
    
\subsubsection{Partition of unity}\label{sss2.3.1}
For some $1/3<\varrho<2/3$, $h_0>0$, and for every $h\in(0,h_0]$, we choose two
sets of indices $\Jg_{i}(h)\,$, $ \Jg_{\partial}(h) \,$, and a set of
points 
\begin{subequations}   
\label{eq:15}
\begin{equation}
\big\{a_j(h)\in \Omega : j\in \mathcal J_i(h)\big\}\cup\big\{b_k(h)\in\pa {\Omega} : k\in \mathcal J_\partial (h)\big\}\,,
\end{equation}
such that $B(a_j(h),h^\varrho)\subset\,{\Omega} $\,,
\begin{equation}
    \overline{\Omega} \subset\bigcup_{j\in \Jg_{i}(h)}B(a_j(h),h^{\varrho})~\cup\bigcup_{k\in \Jg_{\partial }(h)}B(b_k(h),h^{\varrho})\,,
\end{equation}
\end{subequations}
and such that the closed balls $\bar B(a_j(h),h^{\varrho}/2)\,$, $\bar
B(b_k(h),h^{\varrho}/2)$ are all disjoint. \\
We further split $\Jg_\partial (h)$  into three disjoint subsets
\begin{equation}\label{decJdelta}
  \Jg_\partial (h) = \Jg_\pa^D\cup\Jg_\pa^N\cup\Jg_\pa^c \,,
\end{equation}
such that $b_k(h)\in\partial\Omega_D$ whenever $k\in \Jg_\pa^D$, $b_k(h) \in\partial\Omega_N$
whenever $k\in\Jg_\pa^N$, and for every $k\in\Jg_{\pa}^c$  $b_{k}$ denotes
a corner. \\
We note that $\overline{\partial\Omega_D}\cap\overline{\partial\Omega_N}$ is a
finite set consisting of the four corners, and that by the above
construction 
\begin{equation}
\label{eq:16}
  \Union_{k\in \Jg_\partial^c}\{b_k\} = \overline{\partial\Omega_D}\cap\overline{\partial\Omega_N}  \,. 
\end{equation}
We now construct in $\mathbb R^2$ two families of $C^\infty$ functions 
\begin{subequations}\label{eq:14a}
\begin{equation}
 (\chi_{j,h})_{j\in \mathcal J_i(h)} \mbox{ and } (\zeta_{k,h})_{k\in \mathcal J_\partial (h)}\,,
\end{equation}
such that, for every $x\in\overline{\Omega}\,$,
\begin{equation}
\sum_{j\in \mathcal J_i(h)}\chi_{j,h}(x)^2+\sum_{k\in \mathcal J_\partial (h)}\zeta_{k,h}(x)^2=1\,,
\end{equation}
and such that $\supp \chi_{j,h}\subset B(a_j(h),h^{\varrho})$ for
$j\in \mathcal J_i(h)$, $\supp\zeta_{k,h}\subset
B(b_k(h),h^{\varrho})$ for $k\in \Jg_\partial (h)\,$, and
$\chi_{j,h}\equiv 1$ (respectively $\zeta_{k,h}\equiv1$) on $\bar
B(a_j(h),h^{\varrho}/2)$ (respectively $\bar
B(b_k(h),h^{\varrho}/2)$)\,.\\

In addition, we also assume that, for all
$\alpha\in\mathbb{N}^2\,$,  there exist positive
$h_0$ and $C_\alpha$, such that, $\forall h \in (0,h_0]$, $ \forall
x\in \overline{\Omega}$,
\begin{equation}
\sum_j |\pa^\alpha\chi_{j,h}(x)|^2\leq C_\alpha \, h^{- 2 |\alpha|{\varrho}} ~~~\textrm{ and }~~~
\sum_k |\pa^\alpha\zeta_{k,h}(x) |^2 \leq C_\alpha  \, h^{-2 |\alpha|{\varrho}}\,.
\end{equation}
\end{subequations}
To satisfy the Neumann boundary condition on $\partial\Omega_N$, and for later
  reference, we introduce an additional condition
  \begin{equation}
\label{eq:224}
 \frac{\pa \check \xi_{h}}{\pa \nu} \big |_{\partial\Omega}=0 \,.
  \end{equation}
We further set $\eta_{k,h} = 1_{\Omega}\, \zeta_{k,h}\,. $

\subsubsection{Definition of the approximate resolvent}
Following \cite{AGH,Hen} we construct, for any $\epsilon >0$,  an approximate resolvent, which
should be close in operator norm to $(\A_h-\lambda)^{-1}$ as $h\to0$ for
\begin{equation} \label{condlambda}
\Re \lambda \leq   \left(  |J_m|^{2/3}\frac{|\nu_1|}{2} -\epsilon  \right) h^\frac
23\,. 
\end{equation} 

The construction is based on localized resolvents
defined on the disks $B(a_j(h), h^\varrho)$ or $B(b_k(h), h^\varrho)$.\\ 
For $j\in\Jg_i$ we set
\begin{equation}  
\label{eq:17}
\begin{cases}
  \A_{j,h} = -h^2\Delta+i\, \bigl( V(a_j(h))+\nabla V(a_j(h))\cdot(x-a_j(h))\bigr)\,,\\
\Dg(\A_{j,h}) = H^2(\mathbb{R}^2)\cap L^2(\mathbb{R}^2 ; |\nabla V(a_j(h)) \cdot x|^2dx)\,.
\end{cases}
\end{equation}
By Remark \ref{rem:entire} (see also \cite[Lemma 2.1]{Hen}) \\ 
\begin{equation}    
\label{eq:18}
\sup_{\Re \lambda \leq \, \omega \, h^{2/3}}\|(\A_{j,h}-\lambda)^{-1}\|\leq \frac{C_\omega}{h^{2/3}}\,.
\end{equation}
Define in a vicinity of $b_k$ a curvilinear coordinate system $(s,\rho)$
such that $\rho=d(x,\partial\Omega)$ and $s(x)$ denotes the  signed arclength along $\partial\Omega$
connecting $b_k$ and the projection of $x$ on $\partial\Omega$. The boundary
transformation is denoted by $\mathcal{F}_{b_k}$ and
its associated operator by  $T_{\mathcal{F}_{b_k}}$.\\
For $k \in\Jg_\partial^N$ we have $b_k\in\partial\Omega_N$. Hence, we may use the
approximate operator
\begin{equation}  
\label{eq:19}
\begin{cases}
 \tilde\A_{k,h} :=  -h^2\Delta_{s,\rho}+i\big(V(b_k)\pm
  \jg_ks\big) \\
\Dg(\tilde \A_{k,h}) = \{ u\in H^2(\R^2_+) \, | \, \partial_\nu u_{\partial\R^2_+}=0
\,;\; s\, u\in L^2(\R^2_+)  \} \,,
\end{cases}
\end{equation}
where
 \begin{subequations}
\begin{equation}\label{defjk}
\jg_k =|\nabla V (b_k)|= |\pa_\nu V(b_k)|
\end{equation}
and the $\pm$ sign is determined by the condition
\begin{equation} 
\pm \pa_\nu V(b_k) <0\,.
\end{equation}
\end{subequations}
Since $\nabla V$ is parallel to the boundary, it follows from Remark~\ref{rem:Neumann}
that 
\begin{equation}    
\label{eq:20}
\sup_{\Re \lambda \leq \, \omega \, h^{2/3}}\|(\tilde{\A}_{k,h}-\lambda )^{-1}\|\leq
\frac{C_\omega}{h^{2/3}}\,. 
\end{equation}
For  $k\in \Jg_{\pa }^D$,  we use the
approximate operator
\begin{equation}  
\label{eq:183a}
\begin{cases}
 \tilde\A_{k,h} = -h^2\Delta_{s,\rho}+i\big(V(b_k)\pm
  \jg_k\rho\big) \\
\Dg(\tilde \A_{k h}) = \{ u\in H^2(\R^2_+)\cap H^1_0(\R^2_+) \, | \,\rho u\in L^2(\R^2_+)  \} \,.
\end{cases}
\end{equation}

By Lemma \ref{lem:preliminary-lemmas-half-Dirichlet} we have, for any
$\epsilon >0$, the existence of $C_\epsilon$ and $h_\epsilon >0$ such that   
\begin{equation}    
\label{eq:21}
\sup_{\Re \lambda \leq (\jg_k^\frac 23\, |\nu_1|/2-\epsilon)\, \, h^{2/3}}\|(\A_{k,h}-\lambda )^{-1}\|\leq
\frac{C_\epsilon}{h^{2/3}}\,,\,\forall h\in (0,h_\epsilon]\,.
\end{equation}
Hence, by \eqref{eq:20} and \eqref{eq:18} all localized resolvents
satisfy \eqref{eq:21}.

As in \cite{AGH} we construct the following  approximate resolvent 
\begin{equation}
\label{eq:22}
 \mathcal{R}(h,\lambda) = \sum_{j\in \mathcal \Jg_i(h)}\chi_{j,h}(\A_{j,h}-\lambda)^{-1}\chi_{j,h} 
+  \sum_{k\in \mathcal \Jg_\partial (h)}\eta_{k,h}  R_{k,h}(\lambda) \eta_{k,h}\,,
\end{equation}
where 
\begin{equation} 
\label{eq:23}
R_{k,h}(\lambda) =
T_{\mathcal{F}_{b_k}}^{-1}(\tilde\A_{k,h}-\lambda)^{-1}T_{\mathcal{F}_{b_k}}\,.
\end{equation}
Hence it remains to define $\mathcal{F}_{b_k}$ and $\tilde\A_{k,h}$
when $b_k$ is a corner and to establish the corresponding localized
resolvent estimates. This is the object of the next few paragraphs.  
 
\subsubsection{Conjugate harmonic maps}

Let $U$ denote the conjugate harmonic map of $V$ in $\Omega$, which exists
under the assumption that $\Omega$ is simply connected. Using the
Cauchy-Riemann equation, we immediately get that $U\in H^2(\Omega)$ and  has the
same regularity of $V$. As a matter of fact, $U$ is a solution of
$$
\Delta  U= 0\,,\, U= D_\ell   \text{ on } \partial\Omega_N^\ell \mbox{ for }\ell=1,2 \mbox{ and }  \pa_\nu U =0 \mbox{ on } \pa \Omega_D\,,
$$
where $D_1$ and $D_2$ are constants.\\
Due to Assumption \eqref{AssR2}, $U$ and $V$ are also in $C^{n,+}$ at
the corners.
\subsubsection{Curvilinear coordinates in the neighborhood of a corner}\label{sss2.3.4}
Let $\mathfrak c$ be a corner  (after a translation we can
set  $\mathfrak c =(0,0)$) and since $\nabla V (\mathfrak c)\neq0$ and $|\nabla U|=|\nabla V|$
in $\overline{\Omega}$, we define a new local coordinate system $(s,\rho)$ by
\begin{subequations}
\label{eq:24}
  \begin{equation}
  (s,\rho)= \mathcal G_{\mathfrak c}(x)=\Big(\pm\frac{U-D_{\ell(\mathfrak
      c)}}{|\nabla V(\mathfrak c)|} ,\pm\frac{V-C_{i(\mathfrak c)}}{|\nabla
    V(\mathfrak c)|}\Big) \,, 
\end{equation}
where $D_{\ell(\mathfrak c)} = U(\mathfrak c) $, $C_{i(\mathfrak c)}=V(\mathfrak c)$, and the signs
are chosen so that both $s$ and $\rho$ are positive.  The Jacobian $g_\mathfrak c $ of $\mathcal G_{\mathfrak c}(x)$ 
equals $1$ at $\mathfrak c$ and $\mathcal G_{\mathfrak c}$ admits locally an inverse
$\FF_{\mathfrak c}$ of class $C^{n,+}$. Hence we may write
\begin{equation}\label{chc}
  x =\FF_{\mathfrak c}(s,\rho)\,.
\end{equation}
\end{subequations}
Note that
$(s,\rho)$ is an orthogonal coordinate system, and  it can be easily verified that
\begin{subequations}
  \label{eq:25}
  \begin{equation}
\Delta = g_\mathfrak c   \, \Delta_{s,\rho} \,,
\end{equation}
where 
\begin{equation}
  g_\mathfrak c (x) = \frac{|\nabla V (x)|^2}{|\nabla V(\mathfrak c)|^2} = \tilde g_{\mathfrak c} (s,\rho)\,.
\end{equation}
\end{subequations}

A simple computation yields 
\begin{subequations}
\label{eq:26}
  \begin{equation}
  \frac{\partial \tilde g_{\mathfrak c}}{\partial s}=\pm\frac{1}{g_{\mathfrak c}|\nabla V(\mathfrak c)|}\Big[-\frac{\partial g_{\mathfrak c}}{\partial
    x_1}\frac{\partial V}{\partial x_2}+\frac{\partial g_{\mathfrak c}}{\partial x_2}\frac{\partial V}{\partial x_1}\Big] \,, 
\end{equation}
and
\begin{equation}
  \frac{\partial \tilde  g_{\mathfrak c}}{\partial\rho}=\pm\frac{1}{g_{\mathfrak c}|\nabla V(\mathfrak c)|}\Big[-\frac{\partial
    g_{\mathfrak c}}{\partial x_1}\frac{\partial V}{\partial x_1}-\frac{\partial g_{\mathfrak c}}{\partial x_2}\frac{\partial V}{\partial
    x_2}\Big] \,. 
\end{equation}
\end{subequations}
For $n\geq 3$, we may write
 \begin{equation}
  \label{eq:27}
\tilde g_{\mathfrak c} (s,\rho)=1+\tilde \alpha_{\mathfrak c} s+ \tilde \beta _{\mathfrak c} \rho+\OO(s^2+\rho^2) \,.
\end{equation}
One can obtain $\tilde \alpha_{\mathfrak c}$ and $\tilde \beta_{\mathfrak c}$ by setting $(s,\rho)=(0,0)$ in
\eqref{eq:26}. 
\subsubsection{Estimates of  the localized resolvent near the corners}
\label{sec:conclusion}
Let  
\begin{displaymath}
  T_{\mathcal{F}}^x : L^2(\Omega\cap B(x,\delta))  \longrightarrow L^2(\mathcal{U})
\quad \text{s.t.}\quad      T_{\mathcal{F}}^x (u) = u\circ\mathcal{F}_x\,.
\end{displaymath}
Let $b_k$ denote a corner point and $\tilde{\eta}_{k,h}
=T_{\mathcal{F}_{b_k}}(\eta_{k,h})$.   We also introduce
  \begin{equation}\label{defwidehatA}
  \widehat {\A}_{k,h} =T_{\mathcal{F}_{b_k}} \A_h T_{\mathcal{F}_{b_k}}^{-1}\,.
\end{equation}
Let further, with the notation of
\eqref{eq:11},
\begin{equation}\label{deftildeAkcorner}
  \begin{cases}
{\tilde\A_{k,h}} = -h^2\Delta_{s,\rho}+i{\big(V(b_k)\pm \jg_k \rho\big)}
\,,\\
\Dg(\tilde \A_{k,h}) = \{ u\in H^2(Q) \, | \,u|_{\partial Q_\|}=0
\, ; \,\partial u/\partial\nu|_{\partial Q_\perp}=0 \,;\,
\rho\, u\in L^2 (Q)\} \,.
\end{cases}
\end{equation}
Let $\varepsilon>0$ and $\lambda\in\C$ satisfy $\Re\lambda\leq(|\nu_1\|\jg_k|^\frac 23 /2-\varepsilon)
h^{2/3}$. By \eqref{eq:25} and \eqref{eq:27} we have
\begin{multline}
\label{eq:28} 
\|\tilde\eta_{k,h}(\widehat \A_{k,h}-\tilde\A_{k,h})(\tilde\A_{k,h}-\lambda)^{-1}\|
\leq h^2\|\tilde\eta_{k,h}{(\hat g_{b_k}-1)}\Delta_{(s,\rho)}(\tilde\A_{k,h}-\lambda)^{-1}\|  \\
\leq  C\, h^{2+\varrho}\, \|\Delta_{(s,\rho)}(\tilde\A_{k,h}-\lambda)^{-1}\| \,.
\end{multline}
By \eqref{eq:12} there exists $C_\varepsilon>0$ such that
\begin{equation}
\label{eq:29}
  \|(\tilde\A_{k,h}-\lambda)^{-1}\|\leq \frac{C_\varepsilon}{h^{2/3}} \,.
\end{equation}
Furthermore, an integration by parts readily yields 
\begin{displaymath} 
  \|\nabla_{(s,\rho)}(\tilde\A_{k,h}-\lambda)^{-1}\|^2\leq\frac{1}{h^2}
  \|(\tilde\A_{k,h}-\lambda)^{-1}\| + \frac{\Re\lambda}{h^2}\|(\tilde\A_{k,h}-\lambda)^{-1}\|^2\,,
\end{displaymath}
from which, with the aid of \eqref{eq:29}, it  follows that
\begin{equation}
\label{eq:30}
   \|\nabla_{(s,\rho)}(\tilde\A_{k,h}-\lambda)^{-1}\| \leq \frac{C_\varepsilon}{h^{4/3}}\,.
\end{equation}
Finally, as in \cite[Eq. (4.26)]{Hen} we obtain 
\begin{displaymath}
  \|\Delta_{(s,\rho)}(\tilde\A_{k,h}-\lambda)^{-1}\|^2 \leq \frac{C}{h^2}\,
  \|\nabla_{(s,\rho)}(\tilde\A_{k,h}-\lambda)^{-1}\|
  \|(\tilde\A_{k,h}-\lambda)^{-1}\|  \,,
\end{displaymath}
which, with the aid of \eqref{eq:30} and \eqref{eq:29}, yields
\begin{displaymath}
  \|\Delta_{(s,\rho)}(\tilde\A_{k,h}-\lambda)^{-1}\| \leq \frac{C}{h^2}\,.
\end{displaymath}
Substituting the above together with \eqref{eq:29} and \eqref{eq:30}
into \eqref{eq:28} yields
\begin{equation}
\label{eq:31}
  \|\tilde\eta_{k,h}(\widehat \A_{k,h}-\tilde\A_{k,h})(\tilde\A_{k,h}-\lambda)^{-1}\| \leq C\, h^\varrho\,.
\end{equation}
We also need the estimate
\begin{displaymath}
  \|1_{\Omega}  [\A_h,\eta_{k,h}]\,
  R_{k,h}(\lambda)\|\leq C(h^{2-\varrho}\|\nabla_{(s,\rho)}(\tilde\A_{k,h}-\lambda)^{-1}\|+ h^{2-2\varrho}\|(\tilde\A_{k,h}-\lambda)^{-1}\| ) \,,
\end{displaymath}
which follows from (\ref{eq:14a}c) and the fact that
\begin{equation}\label{eq:formcomm}
[\A_h,\eta_{k,h}]= - h^2 (\Delta \eta_{k,h}) - 2 h^2 \nabla \eta_{k,h} \cdot \nabla\,. 
\end{equation}
By \eqref{eq:29} and \eqref{eq:30} we then have
\begin{displaymath}
  \|1_{\Omega}  [\A_h,\eta_{k,h}]\,  R_{k,h}(\lambda)\|\leq Ch^{2/3-\varrho}\,.
\end{displaymath}
Combining the above with \eqref{eq:31} yields 
\begin{equation}
  \label{eq:243}
 \|\tilde\eta_{k,h}(\widehat
 \A_{k,h}-\tilde\A_{k,h})(\tilde\A_{k,h}-\lambda)^{-1}+1_{\Omega}
 [\A_h,\eta_{k,h}]\,  R_{k,h}(\lambda)\|\leq C\, (h^\varrho+h^{2/3-\varrho})\,. 
\end{equation}

\subsubsection{Global error estimate}\label{sss2.3.6}
We may now continue as in \cite[Section 6]{AGH}. We recall that
\begin{equation}   
\label{eq:32} 
 (\A_h -\lambda)\, \mathcal{R}(h,\lambda) = I  +  \Eg(h,\lambda)\,,
\end{equation}
where 
\begin{equation}
\label{eq:202}
 \Eg(h,\lambda)= \sum_{j\in \mathcal J_i(h)}\B_j (h,\lambda)\, \chi_{j,h}  
+  \sum_{k\in \mathcal J_\partial (h)}\B_k(h,\lambda) \, \eta_{k,h} \,.
\end{equation}
In the above, for  $ j\in \mathcal J_i(h) $\,,
\begin{subequations}  
\label{eq:203}
\begin{equation}
  \mathcal B_j:=\B_j(h,\lambda) =\chi_{j,h}(\A_h-\A_{j,h})(\A_{j,h}-\lambda)^{-1}\widehat \chi_{j,h} \\
  + [\A_h,\chi_{j,h}](\A_{j,h}-\lambda)^{-1}\widehat \chi_{j,h} \,, 
\end{equation}
and, for $k \in \mathcal J_\partial(h)$,  
\begin{equation}
      \mathcal B_k:= \B_k(h,\lambda) =  \eta_{k,h} T_{\mathcal{F}_{b_k}}^{-1}
 (\widehat {\A}_{k,h}-\tilde{\A}_{k,h}) (\tilde{\A}_{k,h}
  -\lambda)^{-1} T_{\mathcal{F}_{b_k}} \widehat \eta_{k,h} \\ +
  1_\Omega  [\A_h,\eta_{k,h}]\, R_{k,h}\,\widehat \eta_{k,h}\,,  
\end{equation}
\end{subequations}
where $\widehat \chi_{j,h}$ and $\widehat \eta_{k,h}$ are such that
\begin{itemize}
\item
$\Supp \widehat \chi_{j,h}\subset B(a_j(h),2 h^{\varrho})$ for $j\in \mathcal J_i(h)\,$,
\item
$\Supp \widehat \eta_{k,h}\subset B(b_k (h),2h^{\varrho})$ for $k\in \Jg_\partial\,$, 
\item  
$\widehat \chi_{j,h} \chi_{j,h}  =\chi_{j,h}$ and $\widehat \eta_{k,h}
  \eta_{k,h}=\eta_{k,h}\,$,
\end{itemize}
and
\begin{displaymath}
  \widehat {\A}_{k,h}=T_{\mathcal{F}_{b_k}} \A_h T_{\mathcal{F}_{b_k}}^{-1}\,.
\end{displaymath}

By \eqref{eq:18}, \eqref{eq:20},  \eqref{eq:21}, and \eqref{eq:29}, it
follows, as in \cite{AGH},  that $\mathcal R(h,\lambda)$ is well defined, for $\lambda$
satisfying \eqref{condlambda}. Furthermore, we have
\begin{equation}\label{estR}
\| \mathcal R(h,\lambda)\| \leq C_\epsilon h^{-\frac 23}\,,\, \forall h\in (0, h_\epsilon]\,.
\end{equation}

We now estimate the remainder $ \Eg(h,\lambda)$.  It has been established in
\cite[Section 6]{AGH} that there exists $h_0$ and $C >0$ such that,
for $h\in (0,h_0]$, $j \in \Jg_i(h)$, $k\in \Jg_\pa^D(h)\cup \Jg_\pa^N(h)$ and
$\lambda$ satisfying \eqref{condlambda},
 \begin{equation}
 \| \B_j(h,\lambda)\| + \| \B_k(h,\lambda)\|  \leq C\,  h^{\min{ (\varrho, \frac 23 -\varrho  ) }}\,.
 \end{equation}
By \eqref{eq:31} we have also, for $ k\in \Jg_\partial^c $,
\begin{displaymath}
  \| \B_k(h,\lambda)\,\| \leq  C\,  h^{\min{ (\varrho, \frac 23 -\varrho  ) }} \,.
\end{displaymath}
We now observe, using the finite covering property of the partition,
(\ref{eq:14a}b) and \eqref{eq:203}, that  
\begin{equation}
  \label{eq:205}
  \begin{array}{ll}
\|\Eg(h,\lambda)f\|_2^2 & \leq C_0 \, \left(  \sum_{j\in \mathcal J_i(h)}  \|\mathcal B_j(h,\lambda)\|^2 \| \chi_{j,h} f\|^2_2\,
   + \sum_{k \in \mathcal J_\partial(h)}  \|\mathcal B_k(h,\lambda)\|^2 \| \eta_{k,h}  f\|^2_2\right)\\
   & \leq C h^{2{\min{ (\varrho, \frac 23 -\varrho  ) }}} \,\|f\|^2\,.
   \end{array}
\end{equation}
Consequently, 
\begin{equation}\label{minilemme}
 \sup_{\{\Re \lambda \leq (|J_m|^\frac 23 |\nu_1|/2-\epsilon)\, \, h ^{2/3}\} }
 \|\Eg(h,\lambda)\|\xrightarrow[h\to0]{}0\,, 
\end{equation}
and hence, for sufficiently small $h$, $I+\Eg(h,\lambda)$ is invertible. 
With the aid of \eqref{estR}, we then obtain that for each $\lambda$
  satisfying \eqref{condlambda} we must have $\lambda \in \rho(\A_h)$,  
and by \eqref{eq:32} and \eqref{estR}
\begin{displaymath}
   \|(\A_h -\lambda)^{-1}\|\leq  \|\mathcal{R}(h,\lambda)\|\, \|(I  +
   \Eg(h,\lambda))^{-1}\|\leq C_\epsilon h^{-\frac 23}\,. 
\end{displaymath}
We may now conclude that
for each $\epsilon>0$, there exists $h_0(\epsilon)$ such that whenever
$0<h\leq h_0(\epsilon)$ and  
\begin{displaymath}
  \{\lambda\in\C\,|\, \Re \lambda \leq (|J_m|^\frac 23 |\nu_1|/2-\epsilon)\, \, h ^{2/3}\}\subset \rho(\A_h) \,.
\end{displaymath}
Proposition \ref{prop:lower} is proved.

\section{Quasimode construction - Type V1}
\label{sec:2}
In this section we construct a three terms expansion of a
quasimode. Had $\A_h$ been self-adjoint, we could have used from here
the spectral theorem to obtain the existence of an eigenvalue.
Alternatively, we can use in the self-adjoint case the Min-max Theorem
to obtain an upper bound for the left margin of the spectrum.  This is,
of course, not possible in the non selfadjoint case which is
considered in this work.

\subsection{Local coordinates and approximate operator}
\label{sec:local-coordinates}
Let $x_0\in\Se^m$. Recall that for type V1 potentials, $x_0$ lies in the
interior of $\partial\Omega_D$.  In the curvilinear
coordinate system $(s,\rho)$ centered at  $x_0$ we have
\begin{equation}
\label{eq:33}
  \Delta = \Big(\frac{1}{g} \frac{\partial}{\partial s}\Big)^2 + \frac{1}{g}
  \frac{\partial}{\partial \rho}\Big(g\frac{\partial}{\partial \rho}\Big) =\frac{1}{\tilde g^2} \frac{\pa^2}{\pa s^2} + \frac{\pa^2}{\pa \rho^2}  
  - \frac{ \rho \kappa'(s) }{\tilde g^3} \frac{\pa}{\pa s} - \frac{\kappa(s)}{\tilde g} \frac{\pa}{\pa \rho}  \,, 
\end{equation}
where
\begin{equation}\label{gkappa}
 g(x):=\widetilde g (s,\rho)= 1- \rho \, \kappa(s)\,,
\end{equation}
and $\kappa(s)$ is the curvature at $s$ on $\partial\Omega\,$.\\
Note for later reference that (\ref{gkappa}) implies
\begin{equation}\label{estg}
|\widetilde g (s,\rho)- 1| \leq  C \,\rho \,  \mbox{ for } (s,\rho) \in (-s_0,s_0) \times [0,\rho_0)\,.
\end{equation}
We next expand $V$ in the curvilinear coordinates $(s,\rho)$,
\begin{equation}\label{deftildeV}
V(x)-V(x_0)=  \widetilde V (s,\rho)- V(x_0) = c\, \rho +\frac{1}{2}\hat{\beta}\rho^2+\frac{1}{2}\alpha s^2\rho \,
  + \delta \widetilde V (s,\rho) \,,
\end{equation}
where
\begin{equation}\label{defalphac}
  c = \widetilde V_\rho(0) \quad , \quad \alpha= \widetilde V_{ss\rho }(0)  \quad , \quad
  \hat{\beta} = \widetilde V_{\rho\rho}(0)  \,,
\end{equation}
and
\begin{equation}\label{esttildeV}
  |\delta \widetilde V (s,\rho)|\leq C\, (|s|\rho^2 + |\rho|^3 )\,,\, \mbox{ for } (s,\rho) \in (-s_0,s_0) \times [0,\rho_0)\,.
\end{equation}
\begin{remark}
\label{rem:sign-c}
Using the notation of the introduction, we observe that $$|c| = |\nabla V
(x_0)| = J_m$$ and note that the sign of $c$ is determined by the
values of $C_1$ and $C_2$ in \eqref{eq:2}. Thus, for $x_0\in \Omega^1_D$
and $C_2 >C_1$ or $x_0 \in \Omega^2_D$ and $C_2 < C_1$, we have $c>0$
whereas for $x_0\in \Omega^2_D$ and $C_2 >C_1$ or $x_0 \in \Omega^1_D$ and $C_2
< C_1$, we have $c<0\,$.
\end{remark}
Note that we may deduce from \eqref{ass1},  \eqref{eq:7} and \eqref{eq:8} that:
\begin{equation}
\alpha c  >0\,.
\end{equation}
In the rest of this section we assume $c >0$, without any loss of
generality, since otherwise we consider $\bar{\A}_h$ instead of
$\A_h$ and use the relation  
$$
\sigma(\bar{\A}_h)=\overline{\sigma(\A_h)}\,.
$$
{\bf Blowup}\\
 Applying the transformation
\begin{equation}
\label{eq:34}
  \tau = \Big(\frac{J_m}{h^2}\Big)^{1/3}\rho \,,\, \sigma=
  \Big(\frac{\alpha^3_m}{8J_mh^4}\Big)^{1/12}s \,,
\end{equation}
to \eqref{eq:33}  with 
$$ 
u(x) = \widetilde u (s,\rho)=\check u(\sigma, \tau) \,,
$$
yields the identity
\begin{equation}
\label{eq:36} 
 h^2\Delta u  =  (hJ_m)^{2/3}\big(\check u_{\tau\tau}+ \ego(h) \, \check
 u_{\sigma\sigma}-\, \ego (h) \check \kappa (\sigma)   [2 J_m/ \alpha_m]^{1/2} \, \check u_\tau+ \delta \,   u \big) \,,
\end{equation}
where
\begin{equation}
\label{eq:35}
   \ego (h)=  \,\frac{\alpha_m^{1/2}h^{2/3}}{2^\frac 12 J_m^{5/6}}\,.
\end{equation}
Here $\check  \kappa(\sigma) =\kappa(s(\sigma))$  and $\delta$ is the operator $ u \mapsto \delta
u$ given by
\begin{multline}\label{eq:21a}
  \delta  u = \ego (h) \Big(\frac{1}{\tilde g^2} - 1\Big) \, \check u_{\sigma\sigma} +
 \ego(h) ^{5/2} \frac{2J_m}{\alpha_m} \frac{ \tau\kappa^\prime(s(\sigma))
 }{\tilde g^3}\,\check u_\sigma\\ - \ego(h) \check \kappa (\sigma)
   [2 J_m/\alpha_m]^{1/2} \Big(\frac{1}{\tilde g} 
  -1\Big)\, \check u_\tau\,.
\end{multline}
It can be easily verified, using \eqref{estg}, that, there exists $C$,  $h_0 >0$ and $\rho_0 >0$ such that, for $h \in (0,h_0]$ and   $u$ s.t.  $\supp \widetilde u \subset  (-s_0,s_0)
\times [0,\rho_0)$, 
\begin{equation}  
  \|\delta  u\|_2 \leq C \ego(h)^2 \|\check u\|_{B^3(\mathbb R^2_+)} \,,
\end{equation}
where for $\ell \in \mathbb N$, 
$$B^\ell(\mathbb R^2_+) =\{ \check u\in L^2(\mathbb R^2_+)\,,\, \sigma^p\tau^q \pa_\sigma^m\pa_\tau^n \check u \in L^2 \,,\, \forall\, p,q,m,n \geq 0 \mbox{ s.t. } p+q+m+n \leq \ell\}\,.$$

Converting (\ref{deftildeV}) to the coordinates $(\sigma,\tau)$ via \eqref{eq:34}
yields
\begin{displaymath}
   \check V(\sigma,\tau) - V(x_0) = (hJ_m)^{2/3}\Big(\tau+ \ego(h) \Big[\sigma^2\tau+
    \frac{\hat{\beta}}{2^{1/2}[\alpha_mJ_m]^{1/2}}\tau^2\Big]+\delta \check {V}\Big)\,.
\end{displaymath}
Using \eqref{esttildeV}, we may conclude that there exists $C$, $h_0
>0$ and $\rho_0 >0$ such that, for $h \in (0,h_0]$ and $u$ s.t.  $\supp
\widetilde u \subset (-s_0,s_0) \times [0,\rho_0)$,
\begin{equation}\label{errest}
\| \delta {V} u \| \leq C\, \ego(h) ^\frac 32\,\|\check u\|_{B^3(\mathbb R^2_+)} \,.
\end{equation}

We thus obtain the approximate problem  (for $c>0$)
\begin{equation}
\label{eq:37}
  \begin{cases}
    - \check u_{\tau\tau} + i\tau \check  u + \ego(- \check u_{\sigma\sigma}+i\sigma^2\tau  \check
    u+i\beta\tau^2  \check u + 2 \omega\check u_\tau)
    +  \OO(\ego^{3/2}) = \lambda  \check u  & \text{in } \R^2_+ \\
\check u(0,\sigma)=0  \mbox{ for }  \sigma\in\R &  \,,
  \end{cases}
\end{equation}
where the $\OO(\ego^\frac 32)= \OO (h)$ term is bounded by the right
hand side of \eqref{errest}, for $\supp \widetilde u \subset (-s_0,s_0) \times
[0,\rho_0)$, and
\begin{equation}\label{defomega}
  \omega= \kappa (0) \Big[\frac{J_m}{2 \alpha_m}\Big]^{1/2}\, \quad ; \quad \beta= \frac{\hat{\beta}}{[2\alpha_mJ_m]^{1/2}}\,.
\end{equation}
We recall that for $c<0$ we obtain \eqref{eq:37} once again by taking the
complex conjugate of the approximate equation, together with the
change of parameters $(\beta,\omega)\to(-\beta,-\omega)$. 
\begin{remark}\label{remext}
  Although not needed in this section for the formal construction of
  the quasimode, it will become necessary in Section \ref{sec:upper},
  to define the curvilinear coordinates $(s,\rho)$ and their
  corresponding blowup \eqref{eq:34} centered at a point $y$ in $\Se$ (instead previously at the point $x_0$).
  All the quantities appearing above $\kappa, c,\alpha, \hat \beta, \beta, \ego$ are
  then computed at the point $y$ chosen as the origin (for clarity
  we denote them in Section~\ref{sec:upper} by $\kappa (y), c(y), \alpha (y), \dots, \ego(h,y)$).
\end{remark}
\subsection{The formal construction}
We  look, in the $(\sigma,\tau)$ variables, for an approximate
  spectral pair in the form (modulo a multiplication by a cut-off
  function) 
\begin{displaymath}
  u = u_0+\ego u_1  \,, \,
  \lambda=\lambda_0+\ego\lambda_1  \,,
\end{displaymath}
with $u_0$, $u_1$ in $\mathcal S (\overline{\R^2_+})$.\\

{\bf The leading order balance} reads
\begin{equation}\label{eq:24c}
 (\LL^+-\lambda_0) u_0= 0\,,
 \end{equation}
 where
 \begin{subequations}
  \label{eq:38}
  \begin{equation}
\LL^+ = -\frac{\partial^2}{\partial\tau^2} + i\tau \,,
\end{equation}
As an unbounded operator on $L^2(\R_+)$ its domain is
\begin{equation} 
  D(\LL^+) =\{ u\in H^2(\R_+)\cap H^1_0(\R_+) \, | \, \tau u\in L^2(\R_+)\}\,,
\end{equation}
\end{subequations}
We use the same notation for its natural extension to
$L^2(\R_+^2)$ by a tensor product.\\ 
For $u_0$ in the form
\begin{equation}
u_0(\sigma,\tau) = v(\tau) \, w_0(\sigma)\,,
\end{equation}
\eqref{eq:24c} leads to
\begin{equation}
(\LL^+-\lambda_0) \, v =  0
\end{equation}
in $L^2(\R_+)$  and hence 
\begin{displaymath}
  v(\tau)= v_1(\tau), \quad \lambda_0 = - e^{-i2\pi/3}\nu_1 \,.
\end{displaymath}
Here $\nu_k$ denotes, for $k\geq 1$  the $k$th zero of Airy's function, and
\begin{subequations}
\label{eq:39}
  \begin{equation}
  v_k (\tau)= C_k\, A_i(e^{i\pi/6}\tau+\nu_k)\,,
\end{equation} 
where
\begin{equation}
  C_k = \Big[\int_0^\infty |A_i(\tau+\nu_k)|^2\,d\tau\Big]^{-1/2}\,,
\end{equation}
\end{subequations}
which  follows from the normalization
\begin{displaymath}
  \langle \bar{v}_k , v_k \rangle=1 \,. 
\end{displaymath}
We thus conclude that $u_0$ must have the form
\begin{equation}
u_0 (\sigma,\tau) = v_1(\tau) \, w_0 (\sigma)\,,
\end{equation}
where $w_0\in \mathcal S(\R)$. It follows that $u_0 \in \mathcal
S(\overline{\R_+^2})$ and satisfies the Dirichlet boundary condition
at $\tau=0$. We will determine $w_0$ from the next order balance.

{\bf The next order} balance assumes the form
\begin{equation}
  \label{eq:40}
  (\LL^+-\lambda_0) u_1  =- \Big(-\frac{\partial^2}{\partial\sigma^2} +2\omega\frac{\partial}{\partial\tau}+ i(\sigma^2\tau +\beta\tau^2)-\lambda_1\Big)u_0 \quad;\quad u_1(0,\sigma)=0 \,.
\end{equation}

Taking the inner product of \eqref{eq:40} with $\bar{v}_0$ in
$L^2(\R_+,\C)$ we obtain that the pair $(\lambda_1,w_0)$ should satisfy
\begin{displaymath}
(\PP -\lambda_1)w_0=0\,,
\end{displaymath}
where $\PP$ is defined on
\begin{displaymath}
  D(\PP)=\{ u\in H^2(\R)\, | \, \sigma^2 u\in L^2(\R)\}
\end{displaymath}
by
 \begin{equation}\label{eq:27a}
 \PP:=  -\frac{\partial^2}{\partial\sigma^2} + e^{i\pi/6}\tau_m\sigma^2+\beta\tau_{m,2}  \,,
 \end{equation}
 with
\begin{equation}
\label{eq:41} 
  \tau_m = e^{i\pi/3}\langle\bar{v}_1,\tau v_1\rangle \,, \quad  \tau_{m,2} = i\langle\bar{v}_1,\tau^2v_1\rangle\,.
\end{equation}
Note that by Cauchy's Theorem and deformation of contour, we obtain that  $\tau_m$ and $\tau_{m,2}$ are  real and satisfy
\begin{equation}\label{deftaum}
  \tau_m=\int_{\R_+} \tau\Ai^2 (\tau+\nu_1)\,d\tau>0 \mbox{ and }  \tau_{m,2}=\int_{\R_+} \tau^2 \Ai^2 (\tau+\nu_1)\,d\tau>0  \,. 
\end{equation}
We now choose  $\lambda_1$ as the eigenvalue with smallest real part  of
$\PP$ which is a complex harmonic operator  
 and take $w_0$ as the corresponding eigenfunction
\begin{equation}
\label{eq:42}  
  w_0 (\sigma) = C_0\exp\Big\{-\Big[\frac{\tau_m}{2}\Big]^{1/2} e^{i\frac{\pi}{12}}  
  \sigma^2\Big\} \quad ,  \quad \lambda_1 = \sqrt{2\tau_m}e^{i\tau_m\frac{\pi}{12}}
  +\beta\tau_{m,2} \,,
\end{equation}
where $C_0$ is  chosen so that 
$$
 \int_\R  w_0(\sigma)^2\,d\sigma=1\,.
 $$ 
With this choice of $\lambda_1$, the function $u_1 \in  \mathcal
S(\overline{\R_+^2})$  must satisfy 
\begin{equation}
\label{eq:43} 
  (\LL^+-\lambda_0) u_1  = -i[\sigma^2(\tau-e^{-i\pi/3}\tau_m)+\beta(\tau^2-i\tau_{m,2})+2 i\omega\partial_\tau]u_0 \quad;\quad u_1(0,\sigma)=0 \,.
\end{equation}
Let $\Pi_k$ denote the spectral projection of $L^2(\R_+,\C)$
on ${\rm
  span}\,v_k$, defined by:
\begin{equation}
  \label{eq:44}
\Pi_k u = \langle u, \bar{v}_k\rangle_\tau \, v_k\,,
\end{equation}
where $\langle\cdot,\cdot\rangle_\tau$ denotes the inner product in $L^2(\R_+,\C)$ with
respect to the $\tau$ variable.  We use the same notation (instead of
${\it Id}\, \widehat \otimes\,  \Pi_k $) for its natural extension to \break
$ L^2(\mathbb R)  \widehat \otimes L^2(\mathbb R_+) = L^2(\R_+^2)$. \\
Consequently
we may write
\begin{displaymath}
  u_1(\sigma,\tau)= w_1(\sigma)v_1(\tau) + \hat u_1(\sigma,\tau)  \,,
\end{displaymath}
where $\hat u_1 \in(I-\Pi_1)L^2({\R_+^2})$ and $w_1\in \mathcal S(\R)$ is
left arbitrary (and should be obtained from higher order balances).
We set $w_1=0$ in the sequel, as a two-term expansion satisfies our
needs in the next sections. \\
With Fredholm alternative in mind, we look for $u_1(\sigma,\tau)$ in the
form
$$
u_1(\sigma,\tau)= - i \sigma ^2 w_0 (\sigma) u_{11} (\tau)  + w_0 (\sigma) \big(\beta u_{12} (\tau)+\omega u_{13}(\tau)\big) \,,
$$
where $u_{11}(\tau)$ is the unique solution in $\Im (I-\Pi_1)$ of
$$
(\LL^+ -\lambda_0) u_{11}(\tau)= (\tau-e^{-i\pi/3}\tau_m) v_1(\tau)\,,\, u_{11}(0)=0\,,
$$
$u_{12}(\tau)$ is the unique solution in $\Im (I-\Pi_1)$ of
$$
(\LL^+ -\lambda_0) u_{12}(\tau)=  (\tau^2-i\tau_{m,2}) v_1(\tau)\,,\, u_{12}(0)=0\,.
$$
and $u_{13}(\tau)$ is the unique solution in $\Im (I-\Pi_1)$ of
\begin{displaymath}
  (\LL^+ -\lambda_0) u_{13}(\tau)=  v_1^\prime(\tau)\,,\, u_{13}(0)=0\,.
\end{displaymath}
The above equations are uniquely solvable, since their right hand
sides are all orthogonal to $\bar v_1$, (and hence both lie in ${\rm
  Im}(I-\Pi_1)$), and since $ (\LL^+ -\lambda_0)_{/\Im (I-\Pi_1)}$ is
invertible.  It is not difficult to show that $u_{1,j}$ belongs to
$\mathcal S(\overline{\R_+})$ for $j=1,2,3\,$.
\\
We have thus determined $\lambda_1$ and $u_1\in \mathcal
S(\overline{\R^2_+})$, providing sufficient accuracy for the
derivation of the upper bound in the last section.
\begin{remark}
  The above expansion is similar to the one given in \cite{GH}. Following
  the same steps detailed there, one can formally construct an
  approximate solution, of arbitrary algebraic accuracy (i.e. of
  $\OO(\ego^p)$ for any $p>0$).
\end{remark}

\subsection{Quasimode  and remainder}
We can now set the approximate eigenpair $(U^1,\Lambda^1)$ to be given by
  \begin{equation}\label{eq:45}
 \Lambda^1(\ego)=\lambda_0+\ego\lambda_1 \quad ; \quad \check U^1(\sigma,\tau) =\eta_\ego(\sigma,\tau) \left(\check u_0(\sigma,\tau)+\ego \,\check u_1(\sigma,\tau)\right)  \,,
\end{equation}
where the accent $\check{\cdot}$  is used to  denote
functions of $(\sigma,\tau)$, and $\eta_\ego\in C^\infty(\R^2,[0,1])$  is a cut-off function
supported in a neighborhood of $x_0$
\begin{displaymath}
  \eta_\ego (\sigma,\tau)  =
  \begin{cases}
    1 & \mbox{for } |\sigma^2+ \tau^2|\leq\ego^{-1/4} \\
    0 & \mbox{for }  |\sigma^2 + \tau^2|>2\ego^{-1/4} \,.
  \end{cases}
\end{displaymath}
For latter reference we define
\begin{equation}
  \label{eq:46}
\Lambda^1_\gamma(\ego)  =\lambda_0+ \gamma\,\ego\, \lambda_1\,,
\end{equation}
for some $0\leq\gamma\leq1$. \\
We finally state
\begin{proposition}
\label{prop3.2}
  Let $x_0\in \Sg^m$ and  $(\check U^1,\Lambda^1)$ be given by
  \eqref{eq:45}.   Let for $c(x_0)>0$
  $$ 
  U_h^1(x)= \widetilde U_h^1(s,\rho)=\check U^1\Big(  \Big[\frac{\alpha^3_m}{8J_mh^4} \Big]^{1/12}s
  ,\Big[\frac{J_m}{h^2}\Big]^{1/3}\rho\Big)\,.
  $$ 
and 
  \begin{equation}\label{eqdefhatLambda}  
  \hat{\Lambda}^1(h,x_0)=i V(x_0)+
(J_m h)^\frac 23  \Lambda^1(\ego(h))\,,\, \ego(h) = \frac{\alpha_m^{1/2}}{ 2^{1/2} [J_m]^{5/6}} h^{2/3}\,.
\end{equation}
For $c(x_0)<0$ set $U_h^1= \overline{\widetilde U_h^1}$ and 
\begin{displaymath}
  \hat{\Lambda}^1(h,x_0)= i V(x_0)+
(J_m h)^\frac 23  \overline{\Lambda^1(\ego(h))}\,.
\end{displaymath}
  Then, there exist $C>0$ and $h_0 >0$ such that, for all $h\in (0,h_0)$,  
  \begin{equation}
    \label{eq:47}
\|(\A_h - \hat{\Lambda}^1(h,x_0))U^1_h\|_2\leq C \, h^{5/3}  \| U^1_h\|_2\,.
  \end{equation}
\end{proposition}
The proof follows from the preceding asymptotic expansion, the fact
that $\supp \widetilde U^1_h$ belongs to $(-s_0,s_0) \times [0,\rho_0)$, and
the exponential decay of $u_0$ and $u_1$ in $\R^2_+$ which implies 
that $ \| (1-\eta_\ego) u_j\|_{B^3(\R^2_+)}=\OO (\ego^{+\infty})=\OO (h^{+\infty})\,$.

\section{Quasimode construction - Type V2}
\label{sec:quazimode-V2}
In this section we present a similar construction to the previous
section for type V2 potentials. \\
Let $x_0\in\hat{\Se}^m$ which for type V2 potentials is a corner point.
We use the curvilinear system of coordinates $(s,\rho)$ given by
\eqref{eq:24}.  The corner is set to be the origin, and $(s,\rho)$ varies
in a neighborhood of $(0,0)$ in $Q= [0,+\infty)\times [0,+\infty)$.  We then use
the diffeomorphism $\mathcal{F}_{x_0}(h)$ given by (\ref{eq:24}b). The
potential $V$ is given in the vicinity of $x_0$ by
\begin{equation}
\label{eq:48}
  V(x)= \widetilde V (s,\rho) = V(x_0) + c\,\rho \,,
\end{equation}
where $c=\pm |\nabla V(x_0)|$. \\

As in the previous section we assume, without
any loss of the generality of the proof, that $c>0$, otherwise we move to consider, as before,
$\bar{\A}_h$ instead of $\A_h$.\\ 
The Laplacian operator is given according to \eqref{eq:25} by
\begin{displaymath}
  \Delta=\tilde g_c \, (\partial^2_\rho+\partial^2_s)\,,
\end{displaymath}
where $$g_c(x)=|\nabla V|^2/c^2= \widetilde g_c(s,\rho)\,.
$$
By the smoothness of $V$, $\widetilde g$ admits, in the
vicinity of $(0,0)$ the expansion
\begin{equation}\label{eq:36a}
  \widetilde g_c (s,\rho) = 1+ \tilde \alpha_{x_0} s + \tilde \beta_{x_0} \rho+\OO(s^2+\rho^2)\,,
\end{equation}
where $\tilde \alpha_{x_0}>0$ since $x_0$ is a minimum of $|\partial V/\partial\nu|$ on
$\overline{\partial\Omega_D}$. 
 Note that by \eqref{eq:26}, at every corner
  we have 
  \begin{equation}
\label{linkalphahatalpha}
  \tilde \alpha_{x_0}=2\hat{\alpha}(x_0)/c>0\,,
  \end{equation}
  where $\hat{\alpha}$ is given
  by \eqref{defhatalpha}. Since $|c|=J_m$ on $\Se$ it follows 
  \begin{equation}\label{hatalpha2}
\tilde \alpha:=  \tilde \alpha_{x_0}=2\hat{\alpha}_m/J_m\,,\, \mbox{ for } x_0\in \hat \Sg^m\,.
  \end{equation}
We now apply the transformation
\begin{equation}\label{eq:49-1}
  \tau = \Big[\frac{J_m}{h^2}\Big]^{1/3}\rho \quad ; \quad \sigma= \Big[\frac{2^3 \hat \alpha_m^3 }{J_mh^4}\Big]^{1/9}s\,,
\end{equation}
to obtain from \eqref{eq:25} that
\begin{displaymath}
  h^{4/3}J_m^{-2/3}\Delta = (1+\varepsilon\sigma+\OO(\varepsilon^\frac 32))\Big(\frac{\partial^2}{\partial\tau^2}+ \varepsilon\frac{\partial^2}{\partial\sigma^2}\Big)\,.
\end{displaymath}
In the above, 
\begin{equation}
\label{eq:49}
  \varepsilon(h) =\Big[2^6 \hat \alpha_m^6J_m^{-8}\Big]^{1/9}\, h^{\frac 49}\,.
\end{equation}
Consequently, we may write 
\begin{equation}
  \label{eq:50} 
(hJ_m)^{-2/3}\A_h u =
-(1+\varepsilon\sigma) \frac{\partial^2\check u}{\partial\tau^2}+i \,\sign c\,\tau \check u -\varepsilon\frac{\partial^2\check u}{\partial\sigma^2}+\delta  u \,,
\end{equation}
where, for $u$ such that $\supp \widetilde u \subset (-s_0,s_0) \times [0,\rho_0)$,
\begin{displaymath}
  \|\delta u\|_2\leq C\,\varepsilon^\frac 32 \, \|\check u\|_{B^4(Q)} \,.
\end{displaymath}
We now continue as in the previous section.  The eigenvalue problem
can be formulated, for $c>0$,  as finding an approximate pair $(\check u,\lambda)$ such that
\begin{equation}
\label{eq:51}
  \begin{cases}
 -(1+\varepsilon\sigma)\check u_{\tau\tau} + i\tau \check u - \varepsilon \check u_{\sigma\sigma}
    + \OO(\varepsilon^\frac 32 ) = \lambda \check u  & \text{in } Q\,, \\
\check u (\sigma,0)=0 &\mbox{ for }  \sigma\in\R_+\,, \\
\frac{\partial \check u}{\partial\sigma}(0,\tau)=0 &\mbox{ for }  \tau\in\R_+ \,.
  \end{cases}
\end{equation}
Note that for $c<0$ we obtain \eqref{eq:51} once again by taking the
complex conjugate of the approximate problem.  

Omitting the accent
$\check{\cdot}$\,, we first assume
\begin{displaymath}
  u = u_0+\varepsilon \,u_1 \quad ; \quad
  \lambda=\lambda_0+\varepsilon\, \check \lambda_1 \,,
\end{displaymath}
with $u_0$ and $u_1$ in $\mathcal S(\overline{Q})$.\\
 The leading order balance  is precisely \eqref{eq:45}, and hence, as before, 
\begin{equation}\label{deflambda0}
  u_0 = v_1(\tau)\, w_0^+(\sigma) \quad ; \quad \lambda_0 = - e^{-i2\pi/3}\nu_1 \,,
\end{equation}
with $w_0$ arbitrary in $\mathcal S (\overline \R_+)$, as long as it
satisfies, the Neumann condition at $\sigma=0$.\\

{\bf The next order balance} assumes the form
\begin{displaymath}
  (\LL^+-\lambda_0) u_1  = \Big(\frac{\partial^2}{\partial\sigma^2} + \sigma\frac{\partial^2}{\partial\tau^2} + \check \lambda_1\Big)u_0 \quad;\quad u_1(0,\sigma)=0 \,,
\end{displaymath}
where $\LL^+$ is defined by \eqref{eq:38}.\\
 As
\begin{displaymath}
  \frac{\partial^2u_0}{\partial\tau^2}= (i\tau-\lambda_0)u_0 \,, 
\end{displaymath}
we obtain that 
\begin{equation}
\label{eq:52}
  (\LL^+-\lambda_0) u_1  =- \Big(\PP_+(\tau)  - \check \lambda_1\Big)u_0 \quad;\quad u_1(0,\sigma)=0 \,,
\end{equation}
where
  \begin{equation}\label{eq:53}
 \PP_+(\tau)  = -\frac{\partial^2}{\partial\sigma^2} - (i\sigma\tau-\lambda_0\sigma)\,.
\end{equation}
Taking the inner product of \eqref{eq:40} with $\bar{v}_1$ in
$L^2(\R_+,\C)$ we obtain that
\begin{displaymath}
  -\frac{\partial^2w_0^+}{\partial\sigma^2} +(\theta_0\sigma- \check \lambda_1)w_0^+=0\,, \quad (w_0^+)^\prime(0)=0\,,
\end{displaymath}
where 
\begin{equation}
  \label{eq:54}
\theta_0=\lambda_0-e^{i\pi/6}\tau_m\,,
\end{equation}
in which $\tau_m$ is given by
\eqref{eq:41}. As
\begin{displaymath}
  \theta_0 =
  -\int_{\R_+}(i\tau-\lambda_0)v_1^2(\tau)\,d\tau =
  \int_{\R_+}(v_1^\prime(\tau))^2\,d\tau=e^{i\pi/3}\int_{\R_+}(\Ai^\prime(e^{i\pi/6}\tau+\nu_1^\prime))^2\,d\tau\,,
\end{displaymath}
it easily follows that  $\arg \theta_0=\pi/6\,$. 

As a Neumann realization of a complex Airy operator on  $\R_+$, the
spectrum of the operator   $ -\partial^2/\partial\sigma^2+\theta_0\sigma$ is discrete and
$\lambda_1$ can explicitly be found  as function of the zeros of the
derivative of the Airy function. Thus,
\begin{equation}\label{deflambda1}
  w_0^+ (\sigma) = C_0 A_i (\theta_0^{1/3}\sigma+\nu_1^\prime) \quad ; \quad \check \lambda_1 =-
  \theta_0^{2/3}\nu_1^\prime \,,
\end{equation}
where $C_0$ is chosen so that 
$$
\int_{\R_+} w_0^+ (\sigma)^2 \, d\sigma =1\,,
$$ 
and $\nu'_1$ is the first zero of $A_i'\,$. 

The problem for $u_1$ then assumes the form
\begin{displaymath}
  (\LL^+-\lambda_0) u_1  = -i\sigma(\tau-e^{-i\pi/3}\tau_m)u_0 \quad;\quad u_1(\sigma,0)=0 \,.
\end{displaymath}
As in the previous section it follows that there exists a unique
solution to the above problem in $(I-\Pi_1)L^2(\R_+,\C)$, which in
addition is in $\mathcal S (\overline{Q})$ and satisfies the
Dirichlet-Neumann condition. 
We can now set
\begin{equation}\label{eq:55}
   \Lambda^2(\varepsilon) =\lambda_0+\varepsilon \check \lambda_1 \quad ; \quad \check U^2=(u_0+\varepsilon u_1)\, \eta_\varepsilon \,,
\end{equation}
where
\begin{displaymath} 
  \eta_\varepsilon(\sigma,\tau)  =
  \begin{cases}
    1 &\mbox{ for }  |\sigma^2 + \tau^2|\leq\varepsilon^{-1/4}\,, \\
    0 & \mbox{ for } |\sigma^2+ \tau^2|>2\varepsilon^{-1/4} \,,
  \end{cases}
\end{displaymath}
to obtain the approximate eigenpair $(\check U^2,\Lambda^2)$. In a similar
manner to the previous
section we define, for later reference
\begin{equation}
\label{eq:56}
\Lambda^2_\gamma(\varepsilon) =\lambda_0+ \gamma\,\varepsilon\, \check \lambda_1\,
\end{equation}
with $\lambda_0$ and $\check \lambda_1$ defined in \eqref{deflambda0} and \eqref{deflambda1}.\\

The following proposition is an immediate  consequence of the foregoing
discussion.
\begin{proposition} 
\label{prop4.1} Let $x_0\in \hat{\Sg}^m$.
  Let $ (\check U^2,\Lambda^2)$ be given by \eqref{eq:55}. Let for $c(x_0)>0$
  $$
  U_h^2(x)= \widetilde U_h^2(s,\rho)=\check U^2\Big(\Big[\frac{8\hat{\alpha}^3_m}{J_mh^4}\Big]^{1/9}s
  ,\Big[\frac{J_m}{h^2}\Big]^{1/3}\rho\Big)\,,
$$
and 
  \begin{equation} 
\label{eq:44a} 
  \hat \Lambda^2 (h,x_0)  = i V(x_0) + (J_mh)^\frac 23 \Lambda^2(\varepsilon(h)) \,, \, \varepsilon(h)  =  \Big[2^6 \hat \alpha_m^6J_m^{-8}\Big]^{1/9}\, h^\frac 49\,.
  \end{equation}  
For $c(x_0)<0$, set $U_h^2= \overline{\widetilde U_h^2}$ and 
  \begin{displaymath}
  \hat \Lambda^2 (h,x_0)  = i V(x_0) + (J_mh)^\frac 23 \overline{\Lambda^2(\varepsilon(h))} \,.
\end{displaymath}
Then, we have
  \begin{equation}
  \label{eq:57} 
\|(\A_h - \hat \Lambda^2 (h,x_0) ) U_h^2\|_2\leq C\, h^{4/3}\, \| U_h^2 \|_2  \,.
  \end{equation}
 \end{proposition}
 Remark \ref{rem:sign-c} still holds for $x_0\in \hat{\Sg}^m$.

\section{V1 potentials: 1D operators}
\label{sec:v1-potentials:-1d}
\subsection{Motivation}

To prove the existence of an eigenvalue of $\A_h$ in the vicinity of
the approximated value $\hat \Lambda^1(h)$, one needs an estimate of
$\|(\A_h-\lambda)^{-1}\|$ for $\lambda$ in an annulus whose interior circle
encloses $\hat \Lambda^1 (h)$.  The relevant eigenmode is expected to decay
exponentially fast away from a point $x_0$ on $\Se^m$. We thus replace
the type V1 potential by the leading orders in its Taylor expansion
around $x_0$ as in \eqref{deftildeV} and renormalize the operator by
considering an approximation of $\check \A_h = (J_m h)^{-\frac 23} (\A_h-i
V(x_0))$. The spectral parameter $\lambda$ is thus replaced by $\check \lambda =
(J_mh)^{-\frac 23} (\lambda -iV(x_0))$ 
and the parameter $\ego\sim h^\frac 23$
given by \eqref{eq:35} is introduced.  In the next two sections we
estimate the resolvent of the ensuing approximate operator, after the
dilation \eqref{eq:34} centered at $x_0$ is applied, and the blowup
coordinates $(\sigma,\tau)$ are being introduced. The estimation of the error
generated through the use of an approximate potential, instead of $V$,
is left to the last section.

A necessary first step towards the above mentioned resolvent estimate
is to consider a one-dimensional simplification of it. We recall from
\eqref{eq:37}, where we state the eigenvalue problem for the
approximate operator after dilation, that the approximate potential
includes the term $\ego\beta\tau^2$. Compared to the leading order term,
$\tau$, this term is much smaller for all $\tau\ll\ego^{-1}$.  However, any
attempt to drop this term completely from the expansion and to account
for the error afterwards would fail, as by \eqref{eq:42} it has an
$\OO(\ego)$ effect on $\Lambda^1(\ego)$.  Since we seek an estimate of
$\|(\check \A_h-\check \lambda )^{-1}\|$ on a circle centered at $
\Lambda_1(\ego)$ of radius much smaller than $\ego$, a complete
neglect of this term seems impossible. However, since we consider
$\tau\in\R_+$, it makes sense to avoid problems resulting from the fact
that for $\beta<0$, $\tau+\ego\beta\tau^2$ changes sign for sufficiently large
$|\tau|$. We thus multiply $\ego\beta\tau^2$ by an appropriate cutoff
function, so that the error generated by it need not be accounted for
in the last section, considering the fact that the resolvent is
multiplied there by a cutoff function as in Section \ref{s2}.

\subsection{Realization on the entire real line }

Let, for $\ego >0$, $\LL_2(\ego)$ be given by 
\begin{equation}
\label{eq:58}
  \LL_2(\ego) = -\frac{d^2}{d\tau^2} + i(\tau+\ego\beta\chi(\ego^b\tau)\tau^2)\,,
\end{equation}
where $\beta\in\R$, $1/2<b<3/4$, and $\chi\in C_0^\infty(\R,[0,1])$ is chosen such that
\begin{equation}
  \label{eq:59}
\chi(x) =
\begin{cases}
  1 & |x|<1 \,, \\
  0 & |x|>2\,,
\end{cases}
\end{equation}
and so that
$$
\check{\chi} =
\sqrt{1-\chi^2} \mbox{ in } \mathbb R\,,
$$
is in $C^\infty(\R,[0,1])$.
  We shall frequently drop the reference to
$\ego$ in $\LL_2(\ego)$ and
write instead $\LL_2$ when no ambiguity is expected.

Clearly, $\LL_2$ is a closed operator whose domain is given by
\begin{displaymath}
  D(\LL_2) = \{ u\in H^2(\R) \,| \, \tau u\in L^2(\R)\}\,.
\end{displaymath}
We now need to establish that $\LL$ and $\LL_2$ share some properties
in common. In particular, from \cite{al08,Hen} we know that $\LL$ has
a compact resolvent, empty spectrum and that, for all $\mu_0\in \mathbb
R$, the resolvent norm is uniformly bounded in the half space $\Re \lambda
\leq \mu_0\,$:
\begin{equation}
  \label{eq:60}
\sup_{\Re\lambda\leq \mu_0} \|(\LL-\lambda)^{-1}\| < +\infty\,.
\end{equation}
Moreover, we will make use of the following regularity property for
$\LL\,$  (cf. \cite[Proposition 5.4]{AGH}),
\begin{equation}\label{regLL}
\| \tau u\| \leq C \, ( \| \LL u\| + \|u\|), \forall u \in D(\LL)\,,
\end{equation}
which implies together with \eqref{eq:60}, 
\begin{equation}\label{equt}
\| \tau u\| \leq C_{\mu_0}\,  \| (\LL-\mu) u\|, \forall u \in D(\LL)\,,\, \forall \mu \in [-1, \mu_0]\,.
\end{equation}
It can also be easily verified, by integrating by parts $\Re \langle  \LL -\mu) u\,,\, u\rangle$, that
\begin{displaymath}
  |\mu|\,\|u\|_2\leq C \| (\LL-\mu) u\|, \forall u \in D(\LL)\,,\, \forall \mu < -1\,,
\end{displaymath}
which implies  \eqref{equt} for $\mu < -1$ using again \eqref{regLL}.

Similarly, we get the following properties   for $\LL_2\,$:
\begin{proposition}
  \label{lem:1d-operators-R} 
For any $\ego >0$,  $\LL_2=\LL_2(\ego)$ has a compact resolvent. Moreover, for all $\mu_0 \in\R\,$,
  there exists $\ego_0>0$ and $C_{\mu_0}>0$ such that for all $0<\ego\leq
  \ego_0$ , the spectrum of $\LL_2(\ego)$ lies outside $\{\Re \lambda \leq \mu_0\}$ and 
  \begin{equation}
    \label{eq:61}
\sup_{\Re\lambda\leq \mu_0} \|(\LL_2(\ego)-\lambda)^{-1}\| \leq C_{\mu_0} \,.
  \end{equation}
\end{proposition}
 \begin{proof}~\\
 The first statement is an immediate consequence of the boundedness of $\LL -\LL_2$.\\
 To prove (\ref{eq:61}) we first show that
\begin{equation}\label{eq:170a}
\sup_{\{ \Re\lambda\leq \mu_0\}\cap\rho(\LL_2(\ego))} \|(\LL_2(\ego)-\lambda)^{-1}\| \leq C_{\mu_0}\,,
\end{equation}
where $\rho(\LL_2)$ denotes the resolvent set of $\LL_2$.\\
The spectrum of $\LL_2$ being discrete, this uniform bound
implies that $\sigma(\LL_2)$ lies outside $\{\Re \lambda <
\mu_0\}$, and hence, it also implies \eqref{eq:61}.\\
 
{\bf  Consider first the case $|\Im \lambda|\leq\ego^{-(1+2b)/3}$.}\\
Let $\lambda= \mu + i \nu\in\rho(\LL_2)$, $w\in D(\LL_2)$ and $g= (\LL_2-\lambda)w $.
It follows that
    \begin{equation}\label{c1z}
    | \nu | \leq\ego^{-(1+2b)/3}\,.
    \end{equation}
Let further
  $$
   \nu_\ego:= \nu + \ego \beta \chi(\ego^b \nu)  \nu_1(\ego)\,,
   $$
   and
   $$
   \nu_1(\ego)= \frac{\nu^2}{1-2\beta\ego\nu\chi(\ego^b\nu)}\,,
   $$
which is by \eqref{c1z}  well defined when 
  $4 \,\ego^{1-b} |\beta| < 1\,.
  $\\
We assume  in the sequel that
  \begin{equation}\label{condego0}
   4 \,\ego_0^{1-b} |\beta| \leq  \frac 12\,,\mbox{ and } 0< \ego_0 \leq 1\,.
  \end{equation}
  Note that under these assumptions, we have
  \begin{equation}\label{nuegonu}
  \nu_\ego = \nu \, (1+ \OO (\ego^{1-b}))\,.
  \end{equation}

Applying the transformation 
   \begin{displaymath}
    \tau^\prime = \tau-\nu_\ego \,,
  \end{displaymath}
yields (dropping the superscript $\prime$)
\begin{equation*}
  (\LL - \mu )w = g
  -i\ego\beta\Big( (\tau+\nu_\ego) ^2 \varphi_{\ego,\nu} (\tau)  - \nu_1(\ego) \chi (\ego^b \nu)
 \Big)w  \,,
\end{equation*}
where we have introduced
$$
\varphi_{\ego,\nu}(\tau) = \chi (\ego^b (\tau + \nu_\ego) )\,.
$$

By \eqref{eq:60}, \eqref{equt} and \eqref{nuegonu}, there exists, for any $\mu_0\in
\mathbb R$, a constant $C_{\mu_0}$ such that, for $\mu \leq \mu_0$ and $\ego \in (0,\ego_0]$,   
\begin{equation}\label{5.8}
  \|w\|_2 + \|\tau w\|_2  \leq C_{\mu_0} \Big(\|g\|_2
 +  \ego  \|\tau^2\varphi_{\ego,\nu} w\|_2+  \ego |\nu| \|\tau\varphi_{\ego,\nu} w\|_2 + \ego \,\| \left(\nu_{\ego}^2 \varphi_{\ego,\nu} -\nu_1(\ego) \chi(\ego^b\nu)\right) w    \| \Big)\,.
 \end{equation}
 To estimate the second term of the right hand side, we first observe that, for some constant $C_0>0$, 
\begin{displaymath}
  |\tau \varphi_{\ego,\nu}(\tau) |\leq |\nu_\ego |+2\ego^{-b} \leq C_0 (|\nu| + \ego^{-b}) \,,
\end{displaymath}
and hence, using the assumptions on $\nu$, $b$ and $\ego$
\begin{equation}
  \label{eq:62} 
  \ego\|\tau^2 \varphi_{\ego,\nu}  w\|_2 \leq C_0 \ego ( |\nu| +  \ego^{-b}) \|\tau w\|_2  \leq 2C_0   \ego^{2(1-b)/3}\|\tau w\|_2
  \,. 
\end{equation}
For the third term, we simply observe that 
\begin{equation}\label{eq:171aa}
\ego |\nu| \|\tau\varphi_{\ego,\nu} w\|_2\leq   \ego^{2(1-b)/3}\|\tau w\|_2\,.
\end{equation}
It remains to obtain a bound for the last term on the right-hand-side
of \eqref{5.8}
$$
r_\ego:= \ego \| \left( \nu_{\ego}^2  \varphi_{\ego,\nu} -\nu_1(\ego) \chi(\epsilon^b \nu) \right) w    \| \,.
$$
To this end we first observe that
\begin{equation}\label{eq:172a}
r_\ego \leq \ego \nu_{\ego}^2   \|  (\varphi_{\ego,\nu}- \chi(\epsilon^b \nu))  w\| +
\ego  | \nu_\ego^2    -\nu_1(\ego)| \chi(\ego^b \nu)  \| w    \| \,.
\end{equation}
Using the fact that
\begin{displaymath} 
  |\varphi_{\ego,\nu}(\tau)  -\chi(\ego^b\nu)|= |\chi (\ego^b \tau +\ego^b \nu_\epsilon) -  \chi(\ego^b\nu)| \leq \epsilon^b \, ( \sup |\chi'|)\, (|\tau | + |\nu_\epsilon-\nu|)\,,
\end{displaymath}
Equation \eqref{nuegonu} and the assumptions on $\ego,\nu, b$, 
yield for the first term on the right-hand-side of \eqref{eq:172a}
\begin{equation}\label{172aa}
\ego \nu_{\ego}^2   \|  (\varphi_{\ego,\nu}- \chi(\epsilon^b \nu))  w\|
\leq
 C\, \ego^{(1-b)/3} \left( \|\tau w\|_2 +  \ego^{(1-b)/3} \|w\|_2\right) \,.
\end{equation}
For the second term on the right-hand-side of \eqref{eq:172a}, we get
from the identity
  \begin{displaymath}
    \nu^2 + 2\beta\ego\frac{\nu^3\chi(\ego^b\nu)}{1-2\beta\ego\nu\chi(\ego^b\nu)}= \frac{\nu^2}{1-2\beta\ego\nu\chi(\ego^b\nu)}= \nu_1(\ego)\,,
  \end{displaymath}
  the estimate
  \begin{equation}\label{172ab}
  \ego  | \nu_\ego^2    -\nu_1(\ego)| \, |\chi(\ego^b \nu)|\,  \| w    \| \leq C\,  \ego^3 \nu^4 |\chi(\ego^b \nu)|   \| w    \| \leq C\, \ego^{3-4b}  \| w    \| \,.
  \end{equation}
Using \eqref{5.8}-\eqref{172ab}, we obtain
\begin{equation*}
   \|w\|_2 + \|\tau w\|_2 \leq \hat C_{\mu_0} \big(\|g\|_2+
  \ego^{(1-b)/3}\| \tau   w\|_2+ (\ego^{3-4b} + \ego^{\frac{2(1-b)}{3}})\| w\|_2 \big) \,,
  \end{equation*}
  and choosing sufficiently small $\ego_0$ (which could depend on
  $\mu_0$) we finally obtain, for $b <\frac 34$,
  the existence of $C_{\mu_0}$ such that, for any $\ego \in (0,\ego_0]$,
  any $\lambda$ s.t.  $ \Re \lambda \leq \mu_0$, and any $w\in \mathcal D (\LL)$,
  \begin{equation}
\label{eq:63}
   \|w\|_2\leq \hat C_{\mu_0} \|g\|_2 \,.
\end{equation}
Consequently,
\begin{equation}\label{eq:64}
  \sup_{
    \begin{subarray}
  \quad   \quad \quad \quad  \lambda \in \rho(\LL_2(\ego) ) \\
    \quad  \quad \quad   \mu \leq\mu_0 \\
    \quad \quad \quad   |\nu| \leq \ego^{-(1+2b)/3}
    \end{subarray}}\|(\LL_2(\ego)-\lambda)^{-1}\|\leq C_{\mu_0} \,.
\end{equation}

{\bf Consider now $\lambda$ such that $\Re \lambda\leq\mu_0$ and $|\Im \lambda| >
\ego^{-(1+2b)/3}$.\\}

As before let $w\in D(\LL_2)$  and $g= (\LL_2-\lambda)w $. Let $\chi_2(\tau)
=\chi(\ego^b\tau/2)$ and \break $\check  \chi_2(\tau) =\check
\chi(\ego^b\tau/2)$. Hence we have  
$\check{\chi}_2^2+\chi_2^2=1$. 
Clearly,
\begin{displaymath}
   \Im\langle\chi_2^2w,(\LL_2-\lambda)w\rangle=
  - \nu \|\chi_2w\|_2^2
   +\langle\tau(1+\ego\beta\tau\chi)\chi_2^2w,w\rangle+2\, \Im\langle\chi_2^\prime w,(\chi_2w)^\prime\rangle\,.
\end{displaymath}
Consequently,
\begin{equation}
\label{eq:65}
  \ego^{-(1+2b)/3}\|\chi_2w\|_2^2\leq
  C(\ego^{-b}\|\chi_2w\|_2^2+ \ego^{b}\|w\|_2^2+  \ego^b\|(\chi_2w)^\prime\|_2^2 +
  \ego^b\|g\|_2^2 + \ego^{1-2b}\|\chi_2w\|_2^2)\,.
\end{equation}
Furthermore, as
\begin{displaymath}
    \Re\langle\chi_2^2w,(\LL_2-\lambda)w\rangle=\|(\chi_2w)^\prime\|_2^2 -
   \mu \|\chi_2w\|_2^2-\|\chi_2^\prime w\|_2^2\,,
\end{displaymath}
we obtain that
\begin{displaymath}
  \|(\chi_2w)^\prime\|_2^2 \leq C_{\mu_0} (\|\chi_2w\|_2^2+\ego^{2b}\|w\|_2^2+\|g\|_2^2)\,.
\end{displaymath}
Substituting the above into \eqref{eq:65} yields for a new constant $C_{\mu_0}$
\begin{equation}
  \label{eq:66}
\|\chi_2w\|_2 \leq C_{\mu_0} \, \ego^{(1+5b)/6}(\|g\|_2+\|w\|_2)\,.
\end{equation}
We now write, observing that $\check \chi_2(\tau) \, \chi (e^b \tau) =0\,$ on $\mathbb R\,$, 
\begin{displaymath} 
(\LL_2-\lambda)(\check{\chi}_2w) =  (\LL-\lambda)(\check{\chi}_2w)=\check{\chi}_2g + 2\check{\chi}_2^\prime w^\prime+ \check{\chi}_2^{\prime\prime}w \,.
\end{displaymath}
Hence, using \eqref{eq:60}, 
\begin{equation}
\label{eq:67}
  \|\check{\chi}_2w\|_2 \leq C_{\mu_0} (\|g\|_2 + \ego^b\|w^\prime\|_2 + \ego^{2b}\|w\|_2) \,.
\end{equation}
As
\begin{displaymath}
   \Re\langle w,(\LL_2-\lambda)w\rangle=\|w^\prime\|_2^2 -
   \mu \|w\|_2^2 \,,
\end{displaymath}
we easily obtain that
\begin{displaymath}
  \|w^\prime\|_2\leq C_{\mu_0} (\|g\|_2+\|w\|_2)\,.
\end{displaymath}
Substituting the above into \eqref{eq:67} yields
\begin{displaymath} 
   \|\check{\chi}_2w\|_2 \leq C_{\mu_0} (\|g\|_2 + \ego^b\|w\|_2)\,,
\end{displaymath}
which combined with \eqref{eq:66} yields 
\begin{equation}\label{case2}
  \sup_{
    \begin{subarray}{c}
    \lambda \in \rho (\LL_2) \\
      \mu \leq\mu_0 \\
    \quad   |\nu| > \ego^{-(1+2b)/3}
    \end{subarray}}\|(\LL_2-\lambda)^{-1}\|\leq C_{\mu_0} \,.
\end{equation}
The above together with \eqref{eq:64} yields \eqref{eq:170a}. 
\end{proof}
\subsection{Dirichlet realization in the half-line}
Denote the Dirichlet realization of \eqref{eq:58} in $\mathbb R_+$ by
$\LL_2^+(\ego)$ (or $\LL_2^+$ for simplicity). Its domain is given by
$D(\LL^+)$ (see \eqref{eq:38}). We recall that $\LL^+$ has compact
resolvent and that its spectrum consists of eigenvalues with
multiplicity $1$. In the sequel we denote these eigenvalues (ordered
by non decreasing real part) by $\{\vartheta_n\}_{n=1}^\infty$ and their associated
eigenfunctions by $\{v_n\}_{n=1}^\infty$ (recall that $\vartheta_n=|\nu_n|e^{i\pi/3}$).
\begin{proposition}\label{lem:1d-operators-R+}
Let $\mu_0 < \Re \vartheta_2$, $\delta_0 >0$  and
\begin{equation}
  \label{eq:68}
\Lambda (\ego,\delta,\mu_0))=\{\lambda \in \mathbb C\,,\, -1 \leq \Re\lambda\leq \mu_0 \mbox{
  and } | \lambda -  \vartheta_1 - \ego \beta \tau_{m,2}|\geq \delta \} \,. 
\end{equation}
  There exist positive $\ego_0$ and $C$ such that $ (\LL_2^+(\ego)
  -\lambda)$ is invertible whenever $\lambda\in\Lambda (\ego,\delta,\mu_0)$\,, for all  $\ego \in
  (0,\ego_0]$ and $\ego^{2-b}\leq \delta \leq \delta_0$. Moreover
  \begin{equation}
    \label{eq:69}
\sup_{\lambda\in\Lambda (\ego,\delta,\mu_0)}
\|(\LL_2^+(\ego) -\lambda)^{-1}\|\leq  \frac{C}{\delta}\,.
\end{equation}
  \end{proposition}
\begin{proof}
  Let $\lambda\in\rho(\LL_2^+)\cap\Lambda (\ego,\delta,\mu_0)$.  Let $w\in D(\LL_2^+) $\,, $g =
  (\LL_2^+-\lambda) w$ and $\tau_{m,2}$ be given by \eqref{eq:41}. Write
  \begin{equation}
\label{eq:70} 
    (\LL^+-\lambda-\ego\beta\tau_{m,2})w = g - \ego\beta(i\chi_\ego \tau^2-\tau_{m,2})w \,,
  \end{equation}
  with
  $$ 
  \chi_\ego (\tau) = \chi (\ego^b \tau)\,.
  $$
  Recall the definition of $\Pi_k$ from \eqref{eq:44}. Applying $\Pi_1$
  to \eqref{eq:70}, we obtain
\begin{equation}
\label{eq:71} 
  (\vartheta_1-\lambda-\ego\beta\tau_{m,2})\Pi_1w=\Pi_1g - \ego\beta\Pi_1\big((\chi_\ego i\tau^2-\tau_{m,2})w\big)\,.
\end{equation}
From the definition of $\tau_{m,2}$ in \eqref{eq:41} we have
\begin{equation}
\label{eq:72}
  \Pi_1\big((i\tau^2-\tau_{m,2})w\big)=\Pi_1\big((i\tau^2-\tau_{m,2})(I-\Pi_1)w\big)\,.
\end{equation}
Furthermore, from the definition of $\Pi_1$ we have
\begin{equation}
\label{eq:73}
  \big\|\Pi_1\big((\chi_\ego i\tau^2-\tau_{m,2})w\big)-
  \Pi_1\big((i\tau^2-\tau_{m,2})w\big)\big\|_2\leq C \exp
  \Big\{-\frac{\sqrt{2}}{3}|\ego|^{-\frac {3b}{2}}\Big\}\|w\|_2\,.
\end{equation}
Here we have used the decay properties of the Airy function $v_1$ as
  $\tau\to+\infty\,$. See for example  \cite{AS}.\\
Consequently, we may write
\begin{displaymath}
    \big\|\Pi_1\big((\chi_\ego i\tau^2-\tau_{m,2})w\big)\big\|_2\leq C\Big(\|(I-\Pi_1)w\|_2+ \exp
  \Big\{-\frac{\sqrt{2}}{3}|\ego|^{-\frac {3b}{2}}\Big\}\|w\|_2\Big).
\end{displaymath}
By the above and \eqref{eq:71} we then have
\begin{equation}
\|\Pi_1w\|_2\leq\frac{C}{|\lambda-\vartheta_1-\ego \beta\tau_{m,2}|}\Big(\|\Pi_1g\|_2+ \ego\|(I-\Pi_1)w\|_2+\exp
  \Big\{-\frac{\sqrt{2}}{3}|\ego|^{-\frac {3b}{2}}\Big\}\|w\|_2\Big).
\end{equation}
Since $\lambda\in\Lambda(\ego,\delta,\mu_0))$, we obtain
\begin{equation}
  \label{eq:74} 
\|\Pi_1w\|_2\leq\frac{C}{|\lambda-\vartheta_1-\ego \beta\tau_{m,2}|}(\|g\|_2+ \ego\|(I-\Pi_1)w\|_2 ).
\end{equation}
Next, we apply $(I-\Pi_1)$ to both sides of \eqref{eq:70} to obtain
\begin{displaymath}
     (\LL^+ -\lambda+ \ego\beta\tau_{m,2})(I-\Pi_1)w = (I-\Pi_1)g  - \ego\beta (I-\Pi_1)\big((\chi_\ego i\tau^2-\tau_{m,2})w\big) .
\end{displaymath}
It has been established in \cite{AGH} that
\begin{equation}\label{eq:AGH}
 \|
(\LL^+-\hat \lambda)^{-1}(I-\Pi_1)\|\leq C \mbox{  for  }\Re \hat \lambda \leq \mu_0\,.
\end{equation} 
Hence,
\begin{displaymath}
  \|(I-\Pi_1)w\|_2 \leq  C(\|(I-\Pi_1)g\|_2 + \ego\|\chi_\ego \tau^2w\|_2+\ego\|w\|_2) \,.
\end{displaymath}
Having in mind the support  of $\chi_{\ego}$ we  obtain
\begin{equation}
\label{eq:75}
  \ego\|\chi_\ego \tau^2w\|_2\leq 2 \ego^{1-b}\|\tau w\|_2 \,,
\end{equation}
and hence the existence of $C>0$ such that, for all $w\in D(\LL^+)$, 
\begin{equation}
  \label{eq:76}
\|(I-\Pi_1)w\|_2\leq  C\, (\| g\|_2 + \ego^{1-b}\|\tau w\|_2+\ego\|w\|_2) \,.
\end{equation}

Since $w\in D(\LL^+)$, we can apply \cite[Proposition 5.8]{AGH}
and  \eqref{regLL} in the case of $\LL$ to
\eqref{eq:70} to obtain
\begin{displaymath}
  \|\tau w\|_2 \leq C(\|g\|_2 + |\lambda| \|w\|_2+ \ego\|\chi_\ego\tau^2w\|_2)\,.
\end{displaymath}
From \eqref{eq:75} we then get for $\ego_0$ small enough and $\ego\in (0,\ego_0)$
\begin{equation}
\label{eq:77}
   \|\tau w\|_2 \leq C(\|g\|_2 + |\lambda| \|w\|_2)\,.
\end{equation}
Combining \eqref{eq:77} with \eqref{eq:74} and \eqref{eq:76} yields
$$
\begin{array}{ll}
  \|w\|_2&\leq \|\Pi_1w\|_2 + \|(I-\Pi_1)w\|_2 \\ & \leq
C \Big(  \frac{1}{|\lambda-\vartheta_1-\ego \beta\tau_{m,2}|} +1\Big) \left(\|g\|_2+\ego^{2-b}[| \lambda| +1] \|w\|_2]\right) ,
\end{array}$$
 which implies the existence of $\ego_0$ and $\hat C >0$ such that,
 for $\ego\in (0,\ego_0]$ and \break  $|\Im \lambda| \leq \frac{1}{\hat C}
 \ego^{b-2}$,  
 \begin{equation}\label{eq:180a} 
  \|w\|_2\leq \hat C   \Big(  \frac{1}{|\lambda-\vartheta_1-\ego \beta\tau_{m,2}|} +1\Big)  \|g\|_2 \,.
  \end{equation}
 
  To complete the proof of \eqref{eq:69} we need to consider the case
  $|\Im \lambda| > \hat C^{-1} \ego^{b-2}$. To this end we need only observe
  that the proof of \eqref{case2}, where a bound on $\|(\LL_2-\lambda)^{-1}\|$
  is obtained, can be adapted without any changes for the Dirichlet
  realization $\LL_2^+$, whenever $|\Im \lambda| \geq \ego^{-(1+2b)/3}$. As $
  \ego^{b-2}>\ego^{-(1+2b)/3}$ we may conclude that there exist $C>0$
  and $\ego_0$ such that for $(\delta,\ego)$ satisfying the assumption of
  the proposition we have
   \begin{equation}
\label{eq:78}
\sup_{
  \begin{subarray}
   \,\lambda \in  \rho(\LL_2^+(\ego)) \\
   \lambda \in\Lambda (\ego,\delta,\mu_0) 
  \end{subarray}}
\|(\LL_2^+(\ego)-\lambda)^{-1}\|\leq \frac{C }{\delta}\,.
  \end{equation}
  
By the discreteness of $ \sigma(\LL_2^+(\ego))$ it now follows that
$\sigma(\LL_2^+(\ego))\cap\Lambda(\ego,\delta,\mu_0)=\emptyset\,$. This completes the proof of the proposition. 
\end{proof}

Using  the quantitative version of the  Gearhart-Pr\"uss Theorem (see
\cite{Sj} or \cite{Helbook}), we immediately obtain  
 \begin{corollary}
\label{cor:1d-operators-semigroup}
Let $b\in (\frac 12,\frac 34)$.  Let $e^{-t\LL_2^+(\ego)}$ denote the
semigroup associated with $-\LL_2^+(\ego)$.  There exist positive
$\ego_0$ and $C>0$ such that, for $\ego \in (0,\ego_0]$ and $
\ego^{2-b}\leq \delta\leq\Re\vartheta_1+1$, we have
  \begin{equation}
    \label{eq:79}
\|e^{-t\LL_2^+(\ego)}\|\leq \frac{C}{\delta}e^{-t(\Re\vartheta_1-\delta)}\,.
  \end{equation}
\end{corollary}

   \begin{remark}\label{rem5.6}
     Note that $C$ and $\ego_0$, in both \eqref{eq:79} and
     \eqref{eq:69} a priori depend on $\beta$. Nevertheless, as the proof
     of \eqref{eq:69} simply assumes that $\beta$ is bounded, we may drop
     this dependence by confining $\beta$ to a bounded interval.
   \end{remark}

   By \cite[Example 4.1.2]{AH} $\LL_2^+(\ego)$ possesses a complete
   system of generalized eigenfunctions in $L^2(\R_+)$. Denote by
   $\{\tilde{\vartheta}_k\}_{k=1}^{+\infty}$ the sequence of distinct eigenvalues
   ordered by non decreasing real part and the corresponding
   projection operators by $\tilde{\Pi}_k$.  Whenever an emphasis of
   the dependence on $\ego$ is necessary, we use the notation
$$
\tilde{\vartheta}_k=\vartheta_{k,\ego}  \mbox{ and } \tilde{\Pi}_k = \Pi_{k,\ego} \,.
$$

We now attempt to obtain a bound for  the variation
of the eigenvalues and eigenfunctions of $\LL_2^+(\ego)$ as function
of $\ego\,$.
\begin{proposition}
  \label{lem:1d-eigenvalues} For any $\beta_0>0$, $b\in (\frac 12,\frac
  34)$ and $\mu_0 < \Re \vartheta_2$, there exists a positive $\ego_1$ such
  that, for all $\beta\in[-\beta_0,\beta_0]$ and $\ego \in (0, \ego_1]$,
  $\LL_2(\ego)$ has in the half plane $\{\Re \lambda \leq \mu_0 \}$  at
    most a single eigenvalue of multiplicity $1$ denoted by
  $\vartheta_{1,\ego}$ which satisfies
      \begin{equation}\label{eq:184a}
  |\vartheta_{1,\ego}- \vartheta_1 - \, \ego \beta \tau_{m,2}| \leq \ego^{2-b}\,.  
\end{equation}
\end{proposition}
\begin{proof}
  By \eqref{eq:69} there exists $\ego_1 >0$ such that for any $\ego \in
  (0,\ego_1]$, the set \break $\sigma(\LL_2^+(\ego))\cap\{\Re\lambda\leq\mu_0\}$ is either empty
  or
\begin{equation}
\label{eq:80}
\sigma(\LL_2^+(\ego))\cap\{\Re\lambda\leq\mu_0\}\subset B(\vartheta_1+\ego\beta\tau_{m,2},\ego^{2-b})\,.
\end{equation}
Throughout the proof we use the notation $\LL_2(\ego, \beta)$ instead of
$\LL_2(\ego)$ in order to emphasize the dependence on $\beta$ of $\LL_2$.
Let $\Pi_1(\ego,\beta)$ denote the projector associated with the
spectrum of $\LL_2^+(\ego, \beta)$ in the disk $D(\vartheta_1 +  \ego \tau_{m,2},
\ego^{2-b})$. For fixed $\ego \in (0, \ego_1]$, the operator-valued
function $\beta \mapsto ( \widehat \LL_2^+(\ego,\beta) -\lambda)^{-1}\in \LL(L^2(\mathbb
R_+))$, in view of Remark~\ref{rem5.6}, is uniformly continuous for
$\lambda$ on $\partial D (\vartheta_1+  \ego \tau_{m,2},\ego^{2-b})$ and $\beta\in [-\beta_0,\beta_0]$.  We
indeed observe that
$$
\begin{array}{l}
 ( \LL_2^+(\ego,\beta) -\lambda)^{-1} -  ( \LL_2^+(\ego,\beta') -\lambda)^{-1}\\
 \qquad = i \ego (\beta-\beta^\prime) ( \LL_2^+(\ego,\beta) -\lambda)^{-1} \circ
 (\chi(\ego^b\tau) i\tau^2-\tau_{m,2}) \circ ( \LL_2^+(\ego,\beta') -\lambda)^{-1}\,, 
\end{array} \,.
$$
The projector $\Pi_1(\ego,\beta)$ (which can be expressed by a Cauchy
integral of the resolvent along $\partial D (\vartheta_1+  \ego
\tau_{m,2},\ego^{2-b})$) is Lipschitz continuous in $\beta$ in $[-\beta_0,\beta_0]$,
and its rank is therefore a continuous integer valued function of $\beta$.
This rank is consequently constant and, noting that for $\beta =0$ we
have $\Pi_1(\ego,0) = \Pi_1$, must be equal to one. Hence
$\sigma(\LL_2^+(\ego,\beta)) \cap D(\vartheta_1+\ego\beta\tau_{m,2},\ego^{2-b})$ contains
precisely one eigenvalue of multiplicity one for all $\beta\in[-\beta_0,\beta_0]$.
Moreover $\Pi_1(\ego,\beta) = \widetilde\Pi_1= \Pi_{1,\ego}\,$.
\end{proof}
\begin{remark}
  Note that, expressing $\Pi_{1,\ego}$ by a Cauchy integral along a
  fixed circle centered at $\vartheta_1$ and contained in the half space
  $\{\Re \lambda < \Re \vartheta_2\} $, we obtain by (\ref{eq:69}) that there exists
  a constant $C >0$ such that, $\forall \ego \in (0,\ego_1]$,
\begin{equation}\label{esttildev1}
1 \leq \|\Pi_{1,\ego} \|\leq C\,.\end{equation}
\end{remark}

We can now obtain the following complement to Proposition
  \ref{cor:1d-operators-semigroup}: 
\begin{proposition}
\label{lem:1d-operators-semigroup2}
For every $\mu_0 <\Re \vartheta_2 $  there exists $M_{\mu_0} >0$ and $\ego_{\mu_0} >0$ such that, $\forall \ego \in (0,\ego_{\mu_0}]$, 
  \begin{equation}
    \label{eq:81}
\|e^{-t\LL_2^+(\ego)}(I-\Pi_{1,\ego})\|\leq M_{\mu_0} \,e^{-t \mu_0 }\,.  \end{equation}
\end{proposition}
\begin{proof}
Let $\Re \vartheta_1<\mu_0<\Re\vartheta_2$. Recalling \eqref{eq:69}, we observe that
\begin{displaymath}
  r(\mu_0) := \sup_{\Re\lambda=\mu_0}\|(\LL_2^+(\ego)-\lambda)^{-1}\| < + \infty\,.
\end{displaymath}
By \cite[Theorem 1.6]{HeSj} we have
\begin{displaymath}
  \|e^{-t\LL_2^+(\ego)}(I-\Pi_{1,\ego})\|\leq\frac{e^{-\mu_0t}}{r(\mu_0)\int_0^{t/2}\|e^{-s\LL_2^+(\ego)}\|^{-2}e^{-2\mu_0s}\,ds} \, \|\,I- \Pi_{1,\ego}\|\,. 
\end{displaymath}
We may now use \eqref{eq:69}, \eqref{esttildev1} and \eqref{eq:79}
to prove that \eqref{eq:81} holds uniformly for $t\in \R_+$. 
\end{proof}
As
  \begin{displaymath}
    (\LL_2^+(\ego) -\lambda)^{-1}(I-\Pi_{1,\ego})= \int_0^\infty e^{-t(\LL_2^+(\ego) -\lambda)}(I-\Pi_{1,\ego})\,dt\,,
  \end{displaymath}
we obtain from \eqref{eq:81}  the following corollary:
\begin{corollary}
For every $\mu_0 <\Re \vartheta_2 \,,$  there exists $M_{\mu_0} >0$ and $\ego_{\mu_0} >0$ such that, $\forall \ego \in (0,\ego_{\mu_0}]$, 
\begin{equation}\label{eq:82}
\sup_{\Re \lambda \leq\mu_0}\| (I-\Pi_{1,\ego})(\LL_2^+(\ego)-\lambda)^{-1}\| \leq M_{\mu_0}\,.
\end{equation}
\end{corollary}
We conclude by obtaining the dependence on $\ego$ and the decay as
$\tau \to + \infty$ of the eigenfunction $\tilde v_1:=v_{1,\ego}$.

\begin{proposition} Under the assumptions of Proposition
  \ref{lem:1d-eigenvalues},   the corresponding eigenfunction
  $v_{1,\ego}$ can be normalized such that 
\begin{equation}\label{normtildev1}
\int_0^{+\infty}  v_{1,\ego}^2(\tau)\, d\tau =1\,,
\end{equation}
and 
\begin{equation}\label{eq:184c}
  \|v_{1,\ego} -v_1\|_2\leq C_1  \ego\,.
\end{equation}
Moreover, for any $ \Upsilon<\sqrt{2}/3\,$,   there
exists $C_\Upsilon >0$ and $\ego_1 >0$ such that,  for all $0<\ego\leq
\ego_1$, 
\begin{equation}
\label{eq:83}
  \|e^{\Upsilon \tau^{3/2}}\, v_{1,\ego}\|_2\leq C_\alpha\,.
\end{equation}
\end{proposition}
\begin{proof}
 We first observe, by \cite{AsDa}, that we can normalize the
 eigenfunction $\tilde v_1= v_{1,\ego}$ so that \eqref{normtildev1}
 holds.  Once this normalization is applied, we may write 
  \begin{equation}
  \label{eq:84}
\tilde{\Pi}_1 = \Pi_{1,\ego}= \langle \,\cdot\,,\, \bar{{v}}_{1,\ego}\rangle  \, v_{1,\ego}\,.
\end{equation}
Note that the above normalization determines $\tilde v_1$ up to a
multiplication by $\pm 1$. Since $\Pi_{1,\ego}$ is a rank one projection,
it follows by \cite{AsDa} 
$$
\|\Pi_{1,\ego}\| = \|v_{1,\ego} \|^2\,,
$$
which implies together with  \eqref{esttildev1} that
\begin{equation}\label{esttildev1a}
1 \leq \|v_{1,\ego}\|_2\leq C\,.
\end{equation}
To prove \eqref{eq:184c} we first observe that
  \begin{displaymath}
    (\LL_2^+(\ego) -\vartheta_1)v_1= -i\ego\beta\tau^2\chi v_1 \,.
  \end{displaymath}
By \eqref{eq:82} we then have
\begin{equation}
\label{eq:85}
  \|(I- \Pi_{1,\ego})v_1\|_2 \leq M_1\, \ego \,.
\end{equation}
Hence, for some $M_1>0$ and for any $\ego \in (0,\ego_1]$,
there exist $\alpha_{1,\ego}$ and a function  $\psi_\ego \in(I-\Pi_{1,\ego})L^2(\R_+)$ 
satisfying 
\begin{equation}\label{contalphaetpsi}
\|\psi_\ego\|_2\leq M_1 \ego\, ,\, |\alpha_{1,\ego}| \leq M_1\,,
\end{equation}
such that
\begin{equation}
  v_1=\alpha_{1,\ego} \, v_{1,\ego}+ \psi_\ego  \,.
\end{equation}
Taking the inner product with $\bar{v}_1$ and having in mind the
normalization of $v_1$ and $\tilde v_1$ yields  
\begin{displaymath}
 1 = \alpha_{1,\ego} \int_0^{+\infty}   v_{1,\ego}(\tau)  v_1(\tau)\,d\tau  + \int_0^{+\infty}  \psi_\ego(\tau) v_{1}(\tau)\, d\tau  \,.
\end{displaymath}
Taking the inner product with $\bar{v}_{1,\ego}$ yields 
\begin{displaymath}
 \int_0^{+\infty} v_{1,\ego} (\tau)  v_1 (\tau) \, d\tau   = \alpha_{1,\ego} \,.
\end{displaymath}
Consequently,
\begin{displaymath}
1 = \alpha_{1,\ego}^2 +  \int_0^{+\infty}  \psi_\ego (\tau) v_{1}(\tau)\, d\tau \,.
\end{displaymath}
By \eqref{contalphaetpsi}  we must therefore have 
$$
| \alpha_{1,\ego}^2 - 1|\leq C\, \ego\,.
$$
Possibly changing $v_{1,\ego}$ into $-v_{1,\ego}$ we get \eqref{eq:184c}.\\

To obtain the decay of $v_{1,\ego}$ we observe, following Agmon \cite{ag82}
that 
\begin{displaymath}
\begin{array}{ll}
  0& =\Re\langle e^{2\Upsilon\tau^{3/2}}\tilde{v}_1,(\LL_2^+-\tilde{\vartheta}_1)\tilde{v}_1\rangle\\ &  =
\|(e^{\Upsilon\tau^{3/2}}\tilde{v}_1)^\prime\|_2^2-\frac{9\Upsilon^2}{4}\|\tau^{1/2}e^{\Upsilon\tau^{3/2}}\tilde{v}_1\|_2^2-
\Re\tilde{\vartheta}_1\|e^{\Upsilon\tau^{3/2}}\tilde{v}_1\|_2^2 \,,
\end{array}
\end{displaymath}
and
\begin{equation}
\begin{array}{ll}
  0&=
  \Im\langle e^{2\Upsilon\tau^{3/2}} \tilde{v}_1,(\LL_2^+-\tilde{\vartheta}_1)\tilde{v}_1\rangle\\ & 
  =\|\tau^{1/2}e^{\Upsilon\tau^{3/2}}\tilde{v}_1\|_2^2- \Im\tilde{\vartheta}_1\|e^{\Upsilon\tau^{3/2}}\tilde{v}_1\|_2^2 \\
  &  \qquad +3\Upsilon\, \Im\langle\tau^{1/2}e^{\Upsilon\tau^{3/2}}\tilde{v}_1,
(e^{\Upsilon\tau^{3/2}}\tilde{v}_1)^\prime\rangle+\beta\ego\|\chi^{1/2}\tau e^{\Upsilon\tau^{3/2}}\tilde{v}_1\|_2^2\,.
\end{array}
\end{equation}
Combining the above identities yields
\begin{displaymath}
  \Big[1-\frac{9\Upsilon^2}{2}-C\ego^b\Big]\|\tau^{1/2}e^{\Upsilon\tau^{3/2}}\tilde{v}_1\|_2^2\leq C\|e^{\Upsilon\tau^{3/2}}\tilde{v}_1\|_2^2 \,.
\end{displaymath}
Consequently, for $0< \Upsilon < \sqrt{2}/3\,$, there exists $\hat C_\Upsilon >0$
and $\ego_1>0$ such that, for $\ego \in
(0,\ego_1]$, 
\begin{displaymath}
  \|e^{\Upsilon\tau^{3/2}}\tilde{v}_1\|_2^2 \leq \hat C_\Upsilon \, \|{\mathbf 1}_{\{\tau\leq \hat C_\Upsilon\}} \, \tilde{v}_1\|_2^2 \leq \hat C_\Upsilon  \,  \| \tilde v_1\| \,.
\end{displaymath}
We can then conclude by using \eqref{esttildev1}.
\end{proof}

 \begin{remark}More generally, one can prove that, 
  for every $k\in\N$ there exist positive $C_k$ and $\ego_k$ such that, 
  for all $\ego \in (0, \ego_k]\,$,  $\vartheta_{k,\ego} $ is simple  and satisfies
$$
|\vartheta_{k,\ego}-\vartheta_k|\leq C_k \, \ego \,.
 $$
 We can normalize the corresponding eigenfunction $v_{k,\ego}$ by
 $$
 \int_0^{+\infty}  v_{k,\ego}^2(\tau)\, d\tau =1\,,
 $$
 and with this normalization
$$
  \|v_{k,\ego} -v_k\|_2\leq C_k\, \ego\,.
$$
Moreover  for any $ \Upsilon<\sqrt{2}/3$ and any $k\in \mathbb N$,   there
exist $C_k^\Upsilon >0$ and $\ego_k >0$ such that,  for all $\ego\in (0,
\ego_k]$, 
$$
  \|e^{\Upsilon\tau^{3/2}}\, v_{k,\ego}\|_2\leq C_\Upsilon^k\,.
$$
 The proof for $k\geq2$ can indeed be
similarly obtained by considering
\begin{displaymath}
  (\LL_2^+(\ego) -\lambda)^{-1}\Big(I-\sum_{n=1}^{k-1}\tilde{\Pi}_n\Big)
\end{displaymath}
instead of $(\LL_2^+(\ego) -\lambda)^{-1} (I- \tilde{\Pi}_1)$.
\end{remark}

\subsection{Application to 2D separable operators}
The above-derived one-dimensional estimates can now be used to derive
similar estimates for some operators that can be represented as a sum
of $\LL_2^+$ and an operator that depends on $\sigma$ only (see the
definition of $(\sigma,\tau)$ in \eqref{eq:34}). We begin the estimation with
the following auxiliary lemma which will be used in the next section.
\begin{lemma}
\label{cor:taum}
  Let $g\in L^2(\R_+)$, and $\tau_m$ be given by \eqref{eq:41}. Then,
  \begin{equation}
    \label{eq:86}
\big\|\tilde{\Pi}_1\big((\tau-e^{-i\pi/3}\tau_m)g\big)-\tilde{\Pi}_1\big((\tau-e^{-i\pi/3}\tau_m)(I-\tilde{\Pi}_1)g\big)\big\|_2
\leq C\ego\, \|g\|_2 \,.
  \end{equation}
\end{lemma}
\begin{proof}
By (\ref{eq:184c}) we have 
\begin{equation}
\label{eq:87}
  \|\Pi_1-\tilde{\Pi}_1\|\leq C  \, \ego \,.
\end{equation}
From the definition of $\tau_m$ we have
\begin{displaymath}
  \Pi_1\big((\tau-e^{-i\pi/3}\tau_m)\Pi_1g\big)=  0\,.
\end{displaymath}
We may thus write
\begin{displaymath}
   \tilde{\Pi}_1\big((\tau-e^{-i\pi/3}\tau_m)\tilde{\Pi}_1g\big)=
   (\tilde{\Pi}_1-\Pi_1)\big((\tau-e^{-i\pi/3}\tau_m)\tilde{\Pi}_1g\big)+
   \Pi_1\big((\tau-e^{-i\pi/3}\tau_m)(\tilde{\Pi}_1-\Pi_1)g\big)\,.
\end{displaymath}
By \eqref{eq:87} and \eqref{eq:83} we have 
\begin{displaymath}
  \|(\tilde{\Pi}_1-\Pi_1)\big((\tau-e^{-i\pi/3}\tau_m)\tilde{\Pi}_1g\big)\|_2\leq C\ego\,\|(\tau-e^{-i\pi/3}\tau_m)\tilde{\Pi}_1g\|_2\leq C\ego\, \|g\|_2 \,,
\end{displaymath}
and since 
\begin{displaymath}
  \Pi_1\big((\tau-e^{-i\pi/3}\tau_m)(\tilde{\Pi}_1-\Pi_1)g\big)=\langle\bar{v}_1,(\tau-e^{-i\pi/3}\tau_m)(\tilde{\Pi}_1-\Pi_1)g\rangle v_1\,,
\end{displaymath}
we obtain from \eqref{eq:87} and the decay properties of the Airy
function that
\begin{displaymath}
  \|\Pi_1\big((\tau-e^{-i\pi/3}\tau_m)(\tilde{\Pi}_1-\Pi_1)g\big)\|_2\leq C\,\ego\, \|g\|_2 \,.
\end{displaymath}
\end{proof}
The next proposition will also be needed  in the next section.
\begin{proposition}
\label{cor:1d-operators-2d-one} 
Let, for $\ego >0$,  $\beta \in \mathbb R$,  $1/2<b<3/4$, and $I$  an
open interval in $\R$  (we may set $I=\R$ as well),  
\begin{subequations}
\label{eq:207}
    \begin{equation}
    \Mg(\ego,I,\beta,\chi) = \LL^+_2(\ego,\beta)  -\ego\partial^2_\sigma \,, 
  \end{equation}
be defined on
\begin{equation}
  D(\Mg(\ego,I,\beta,\chi))=\{u\in H^2(S_I)\cap H^1_0(S_I) \,| \, \tau u\in L^2(S_I)\} \,,
\end{equation}
\end{subequations}
where $S_I=I\times\R_+$. \\
 Then,
there exist $\ego_0 >0$ and $C>0$ such that, for any triple
$(\ego,\delta,I)$ satisfying $ \delta \in [\ego^{2-b}\,,\, \Re\tilde{\vartheta}_1+1)$ 
and $\ego \in (0,\ego_0]$, the spectrum $\Mg(\ego,I,\beta)$ lies outside $\{
\Re\lambda\leq\Re \tilde{\vartheta}_1-\delta\}$ and
\begin{equation}
  \label{eq:88}
\sup_{
    \Re\lambda\leq\Re \tilde{\vartheta}_1-\delta}
\|(\Mg(\ego,I,\beta,\chi)-\lambda)^{-1}\|\leq \frac{C}{\delta}\,.
\end{equation}
\end{proposition}
The proof can easily be obtained from \eqref{eq:69}
by using Fourier series in $\sigma$, or by using a Fourier transform in the
case $I=\R$ (see also in \cite[Lemma 4.12]{AGH}).\\
Finally we will also make use, in the next section, of the following
proposition:
\begin{proposition}
\label{cor:1d-2d-two}
Let $\tau_m$  and $\PP$ be defined by \eqref{eq:27a} and  \eqref{eq:41}. 
Let further
\begin{equation}
  \label{eq:89}
\Mg^2_\ego = \LL^+_2(\ego) +\ego\, \PP \,, 
\end{equation}
be the closed operator on $L^2(\R^2_+)$ with domain 
\begin{displaymath}
  D(\Mg^2_\ego)=\{  u\in H^2(\R^2_+)\cap H^1_0(\R_+^2)\,|\, (\sigma^2+\tau)u\in
  L^2(\R^2_+)\,\}\,. 
\end{displaymath}
Then, there exist $\ego_0$ and $C>0$, such that,  for all pairs
$(\ego,\delta)$ for which \break  $ \ego^{2-b}\leq \delta \leq \Re \vartheta_1+1$ and $\ego \in (0,\ego_0]\,$,  we
have 
\begin{subequations}
  \label{eq:90}
  \begin{equation}
\|e^{-t\Mg^2_\ego}\|\leq \frac{C}{\delta}e^{-t(\Re\vartheta_1-\delta+\ego\mu_1^r)}\,,
\end{equation}
where $\mu_1^r=\inf \Re\sigma(\PP)$. \\
Furthermore, for each
$\varpi<\Re \vartheta_2$, there exists $M_\varpi>0$ such that
\begin{equation}
  \|e^{-t\Mg^2_\ego}(I-\tilde{\Pi}_1)\|\leq M_\varpi e^{-t\varpi}\,.
\end{equation}
\end{subequations}
\end{proposition}
\begin{proof}
Note first that the potential
\begin{displaymath}
\hat V_\ego(\tau,\sigma) = \tau+\beta\ego\tau^2\chi(\ego^b\tau) + \ego\, e^{-i\pi/3}\sigma^2\tau_m 
\end{displaymath}
has positive real and imaginary parts for $\ego \in (0,\ego_0]$ (with
$\ego_0$ small enough) and satisfies for some $C(\ego,\beta,b)>0$
\begin{displaymath}
  |\nabla \hat{V}_\ego |\leq C\sqrt{1+\hat V_\ego^2} \quad \text{ in } \R^2_+ \,.
\end{displaymath}
Hence we may use the technique in \cite{AH} to obtain that
$\Mg^2_\ego:D(\Mg^2_\ego)\to L^2(\R^2_+)$ has a bounded inverse.  

  Since $\Mg^2_\ego$ is separable, we have (cf. \cite{AGH})
  \begin{displaymath}
    e^{-t\Mg^2_\ego}=  e^{-t\LL_2^+(\ego)}\otimes e^{-t \PP}\,.
  \end{displaymath}
  The proof of (\ref{eq:90}a) follows from \eqref{eq:79} and the
  fact, proven in \cite[Corollary 14.5.2]{da07}, 
\begin{equation}
\label{eq:91}
  \|e^{-t \PP}\|\leq C  e^{-t\,\mu_r^1}\,.
\end{equation}
\end{proof}

\begin{remark}
\label{rem:Dirichlet-separable}
  The validity of \eqref{eq:90} remains intact, if we define
  an operator $\Mg^2_\ego(L)$ as the Dirichlet realization of 
  $\Mg^2_\ego$ in the semi-infinite strip $S_L=(-L,L)\times\R_+$. \end{remark}

\section{V1 potentials: 2D simplification}
\label{sec:simplified}

In this section we estimate the resolvent of the operator appearing on
the left hand side of \eqref{eq:37}, which is obtained from the Taylor
expansion of $V$ near some $x_0\in\Se^m$. To remove the large $\tau$ effect
of the term $i\ego\beta\tau^2$, we attach to it a cutoff function
$\chi(\ego^b\tau)$ as in the previous section. Since the term $i \ego
\sigma^2\tau$ has a significant effect for large values of $|\sigma|$ we
separately estimate the resolvent in the region $|\sigma|\gg1$.  We address
the effect of the term $2 \ego \omega \partial_\tau$ in a later stage. Other error
terms appearing in \eqref{eq:36}-\eqref{eq:21a} will be treated in
Subsection \ref{sec:modification}.

\subsection{The operator $\B_\ego$}
Let 
\begin{equation}
\label{eq:92} 
  \B_\ego  = \LL_2^+(\ego,\beta)  + \ego\,  \Kg
\end{equation}
where $\LL_2^+(\ego,\beta)$ is defined by \eqref{eq:58}, with domain 
given by (\ref{eq:38}b), and
\begin{equation}
  \Kg = -\partial^2_\sigma + i\sigma^2\tau \,.
\end{equation}
We first give  a characterization of the domain of  the Dirichlet realization of
$\B_\ego$ in $\R^2_+$.  We may assume that $\ego_0 >0$ is small enough so that
$$
V_\ego (\tau,\sigma):= \tau (1 +\ego \sigma^2 + \ego \beta \tau  \chi(\ego^b \tau))
$$
is non negative for $\ego \in (0,\ego_0]$.\\
The operator $\B_\ego$ has   the form $-\Delta+iV_\ego$. 
It can be verified that $|\nabla V_\ego |^2+V_\ego^2$ tends to $+\infty$ as
$\sigma^2 +\tau^2$ tends to $+\infty$ and that there exist $C:= C(\ego,\beta,b)>0$
such that
\begin{displaymath}
  |D^2V_\ego |\leq C \sqrt{1+|\nabla V_\ego|^2+V_\ego^2} \quad \text{in } \R^2_+\,. 
\end{displaymath}
Hence 
\begin{equation}\label{domBego}
  D(\B_\ego) = \{u\in H^2(\R^2_+)\cap H^1_0(\R^2_+)\,|\,\tau(1+\sigma^2)u\in L^2(\R^2_+) \}\,,
\end{equation}
and the resolvent of $ \B_\ego$ is compact by \cite[Corollary 5.10]{AGH}.
Note that while the corollary in \cite{AGH} refers separable
operators, we may still apply it, given the Dirichlet
boundary conditions in (\ref{domBego}).  As a matter of fact, we can use all the
estimates of \cite[Section 5.2]{AGH}, obtained in the
absence of boundaries.  

\subsection{Large $|\sigma|$ simplification}
In the following we estimate the resolvent of \eqref{eq:92} for large
$|\sigma|$, or more precisely, in $(\ego^{-{\mathfrak a}}, +\infty)\times \mathbb
R_+$.  It is convenient to shift $\B_\ego$ to a fixed
domain $Q=\R_+\times\R_+$ by using the transformation
$\sigma\to\sigma+\ego^{-{\mathfrak a}}$.
\begin{proposition}
\label{lem:large-sigma} Let ${\mathfrak a} >0$ such that $1/6<{\mathfrak a}<1/4$, and let 
  \begin{subequations}
      \begin{equation}
\label{eq:93}
\CC_\ego  = \LL_2^+(\ego)  -\ego\partial^2_\sigma + i\ego((\sigma+\ego^{-{\mathfrak a}})^2\tau) \,,  
  \end{equation}
be defined on
\begin{equation}
  D(\CC_\ego) = \{u\in H^2(Q)\cap H^1_0(Q)\,| \,\tau(1+ \sigma^2)u\in L^2(Q) \}\,,
\end{equation}
  \end{subequations}
  where $Q=\R_+\times\R_+$. 
Then, for all $\gamma_0>0$ there exist positive $C(\gamma_0)$ and
  $\ego_0$ such that for all $\ego\in (0,\ego_0]$, we have
  $B(\vartheta_1,\gamma_0\ego) \subset \rho(\CC_\ego)$ ,
\begin{subequations}
  \label{eq:94}
  \begin{equation}
  \sup_{\lambda\in B(\vartheta_1,\gamma_0\ego)}   \|(\CC_\ego -\lambda)^{-1}\| \leq \frac{C}{\ego^{1-2{\mathfrak a}}} \,,
  \end{equation}
 and
\begin{equation}
   \sup_{\lambda\in B(\vartheta_1,\gamma_0\ego)}   \|\partial_\sigma(\CC_\ego-\lambda)^{-1}\| \leq \frac{C}{\ego^{3/2-2{\mathfrak a}}} \,.
\end{equation}
\end{subequations}
\end{proposition}
\begin{proof}~\\
  \paragraph{Partition of unity.}
   We begin the proof by introducing an
  appropriate partition of unity. Let $\{\phi_k\}_{k=0}^{+\infty}$ denote a
  sequence of cutoff function in $C^\infty(\R,[0,1])$ satisfying
\begin{subequations}\label{eq:95}
  \begin{equation}
    \phi_k(x) =
    \begin{cases}
      1 & |x-k)|<1/4 \\
      0 & |x-k|>3/4 \,,
    \end{cases} \mbox{ and }
|\phi_k^\prime|+ |\phi''_k| \leq C \,,
  \end{equation}
and
\begin{equation}
   \sum_{k=0}^{+\infty} \phi_k^2=1\,\mbox{ in } \mathbb R_+\,.
\end{equation}
\end{subequations}
For ${\mathfrak b}$ satisfying $1/6<{\mathfrak b}<{\mathfrak a}\,$, 
we introduce
\begin{displaymath}
 \phi_k^\ego(\sigma)=\phi_k(\ego^{\mathfrak b}\sigma)\,.
\end{displaymath}
Let
\begin{equation}\label{defSk}
S_k=((k-1)\ego^{-{\mathfrak b}},(k+1)\ego^{-{\mathfrak b}})\times\R_+
\end{equation} 
and $\CC_{\ego,k}$ denote the Dirichlet realization in $S_k$ 
of the differential operator given by (\ref{eq:93}).  Its domain, for
$k\geq 1$, is given by
\begin{subequations}
\label{eq:209a}
  \begin{equation}
   D(\CC_{\ego,k}) = \{u\in H^2(S_k)\cap H^1_0(S_k)\,| \,\tau u\in L^2(S_k) \}\,.
\end{equation}
For  $k=0$ the domain is given by
\begin{equation}
D(\CC_0) = \{u\in H^2(S_0^+)\cap H^1_0(S_0^+)\,| \,\tau^2u\in L^2(S_0^+) \}\,,
\end{equation}
\end{subequations}
where $S_0^+=S_0\cap Q$.\\
 We attempt to estimate $(\CC_\ego -\lambda)^{-1}$
by the following approximate resolvent
  \begin{equation}
\Rg^{app}_{\CC}= \sum_{k=0}^{+\infty} \phi_k^\ego(\CC_{\ego,k}-\lambda)^{-1}\phi_k^\ego \,.
\end{equation}
Clearly,
\begin{subequations}
  \label{eq:96}
\begin{equation}
  (\CC_\ego-\lambda) \Rg^{app}_{\CC} = I + \Eg_\CC \,,
\end{equation}
where
\begin{equation}
   \Eg_\CC = -\sum_{k=0}^{+\infty} \ego[\partial^2_\sigma,\phi_k^\ego](\CC_{\ego,k}-\lambda)^{-1}\phi_k^\ego \,.
\end{equation}
\end{subequations}
Note that since $\CC_\ego$ has a compact resolvent, boundedness
  of the right inverse of $(\CC_\ego-\lambda)$ immediately implies its surjectivity and
  injectivity, and hence an identity between its right and the left inverses.
To bound $\|(\CC_\ego-\lambda)^{-1}\|$ we have to establish that $\|\Eg_\CC\|
\to 0$ as $\ego\to 0$. To this end we need to show the existence of
  $\ego_0$ such that, for any $k$ and any $\ego\in (0,\ego_0]$, the
  disk  $B(\vartheta_1,\gamma_0\ego)$ belongs to  $\rho(\CC_{\ego,k})$ and to obtain 
  an estimate of $\|(\CC_{\ego,k}-\lambda)^{-1}\|$ in this disc.
 
\paragraph{Control of $ (\CC_{\ego,k}-\lambda)^{-1}$.}~\\

  Let $w\in D(\CC_{\ego,k})$ and  $g\in L^2(S_k)$ (or $L^2(S_0^+)$ when $k=0$) such that
\begin{equation}
  \label{eq:97}
(\CC_{\ego,k}-\lambda)w=g\,.
\end{equation}
We rewrite \eqref{eq:97} in the form
\begin{displaymath}
\begin{array}{l}
 \big(-\partial^2_\tau+i(\mathfrak d(\ego,k) ^3 +\ego\beta\chi(\ego^b\tau) \tau)\tau  -\ego\partial^2_\sigma-\lambda\big)w\\
 \qquad\qquad   =  g
 -i\big\{\ego[\sigma-k\ego^{-{\mathfrak b}}]^2+2[\ego^{1-{\mathfrak a}}+k\ego^{1-{\mathfrak b}}](\sigma-k\ego^{-{\mathfrak b}})\big\}\tau w \,,
 \end{array}
\end{displaymath}
where
$$
 \mathfrak d(\ego,k):= (1+\ego[\ego^{-{\mathfrak a}}+k\ego^{-{\mathfrak b}}]^2)^{1/3}\,.
$$
Using the dilation
\begin{equation}
\label{eq:98}
  (\sigma,\tau) \to \mathfrak d(\ego,k)(\sigma,\tau)
\end{equation}
yields, in $I (\ego,k,\mathfrak d (\ego,k))\times \mathbb R_+$, with 
$$ 
\begin{array}{l}
I(\ego,k,\mathfrak d) = \mathfrak d  ((k-1)\ego^{-{\mathfrak b}},(k+1)\ego^{-{\mathfrak b}}) \,, \\
\, \beta (\ego,k) = \beta \mathfrak
d(\ego, k)^{-4}\,,\, \chi^{\mathfrak d} (\tau) = \chi (\ego^b\mathfrak d(\ego, k)^{-1}\tau
)\,,
\end{array}
$$
and 
\begin{equation}\label{diff}
  \Big(  \Mg(\ego,I (\ego,k,\mathfrak d (\ego,k)),\beta(\ego,k) , \chi^{\mathfrak d(\ego,k)}) - \frac{\lambda}{\mathfrak d(\ego,k) ^2
  }\Big) \widehat w = \mathfrak d(\ego, k)^{-2} \widehat h\,.
\end{equation}
Here
\begin{equation}\label{nothat}
\begin{array}{l}
h=g
 -i\big\{\ego[\sigma-k\ego^{-{\mathfrak b}}]^2+2[\ego^{1-{\mathfrak a}}+k\ego^{1-{\mathfrak b}}](\sigma-k\ego^{-{\mathfrak b}})\big\}\tau w\,,\\
  \widehat w = w\circ \mathfrak d(\ego,k)^{-1}\,,\, \widehat h = h \circ \mathfrak d(\ego,k)^{-1}\,,
  \end{array}
 \end{equation}
and   $ \Mg (\ego, I, \beta, \chi )$ is the Dirichlet realization in the interval $I\times \mathbb R_+$ of
$$
 \Mg (\ego, I, \beta, \chi) := \LL_2^+(\ego,\beta, \chi)  -\ego \partial_{\sigma}^2\,.
$$
Since $\mathfrak d(\ego,k)\geq 1$, the new parameter $\beta (\ego,k) = \beta
\mathfrak d(\ego, k)^{-4}$ is bounded. More caution should be used
below while assessing the effect of \eqref{eq:98} on $\chi^{\mathfrak
  d(\ego,k)}$. Nevertheless, it is safe to apply \eqref{eq:88} as long as
$\mathfrak d (\ego ,k)$ remains in a bounded interval $[1,\mathfrak
d_0]$.
 
We now assume that $\lambda \in B(\vartheta_1,\gamma_0\ego)$. Observe that $$\mathfrak d
(\ego ,k) ^3 \geq 1 + \ego^{1-2 {\mathfrak a}} \,.$$
We first assume  $k\leq  \ego^{-(1/2-{\mathfrak b})}$, so that
$$ 
1 \leq \mathfrak d(\ego,k) \leq \mathfrak d_0:=(3 + 2\ego_0^{1-2{\mathfrak a}})^\frac 13 \,.
$$ 
We now attempt to apply \eqref{eq:88}, with $\delta=\Re\vartheta_1(1-{\mathfrak d}^{-2})/2$,
$\beta=\beta (\ego,k)$, and $I= I(\ego,k)$. Here we note that the constant $C$ in
\eqref{eq:88} is independent of $I$ and that, for $\ego$ small enough, we have
$$
\delta \geq \Re\vartheta_1  \mathfrak d_0^{-1} (\mathfrak d_0^2+\mathfrak d_0
+1)^{-1} (\mathfrak d^3-1)\geq \Re\vartheta_1  \mathfrak d_0^{-1} (\mathfrak
d_0^2+\mathfrak d_0 +1)^{-1} \ego^{1-2\mathfrak a} \,.
$$
Obviously, $1-2\mathfrak a<1<2-b$, and hence we have, for $\ego$ small enough,
$$
\ego^{2-b}  \leq \delta \,.
$$
In addition, for sufficiently small $\ego$,
$$
 \mathfrak d(\ego,k) ^{-2} \, \Re \lambda\,   \leq  (\Re \vartheta_1 + \gamma_0 \ego)
 \mathfrak d(\ego,k) ^{-2} \leq  \Re \vartheta_1 - \delta\,,\, \forall k\,.
$$
Hence all the conditions,  needed for the sake of applying \eqref{eq:88}, are met,
and with the aid of the identity 
\begin{displaymath}
  \ego^{1-{\mathfrak
    a}}+k\ego^{1-{\mathfrak b}}=[\ego({\mathfrak d}^3-1)]^{1/2}\,.
\end{displaymath}
we obtain that
\begin{equation*}
\label{eq:99} 
  \|w\|_2\leq \frac{C}{{\mathfrak d}^2-1}
  \left(\|g\|_2+\ego^{1/2-{\mathfrak b}} ({\mathfrak
      d}^3-1)^{1/2}\|\tau w\|_2\right) 
  \,.
\end{equation*}
Clearly,
\begin{displaymath}
  \frac{[({\mathfrak
      d}^3-1)]^{1/2}}{{\mathfrak d}^2-1}= \frac{1}{[{\mathfrak
      d}-1]^{1/2}}\frac{[{\mathfrak d}^2+{\mathfrak
      d}+1]^{1/2}}{{\mathfrak d}+1}<\frac{1}{[{\mathfrak
      d}-1]^{1/2}}\,.
\end{displaymath}
Substituting the above into \eqref{eq:99} and taking account of the fact
that, for \break $k\leq \ego^{-(1/2-{\mathfrak b})}$, 
$$
{\mathfrak d(\ego,k)}^2-1\geq\ego^{1-2{\mathfrak a}}/(2 (\mathfrak d_0 +1))\,,
$$ 
we obtain
\begin{equation}
\label{eq:100}
  \|w\|_2\leq \frac{C}{\ego^{1-2{\mathfrak a}} } \left(\|g\|_2+\ego^{1-\mathfrak b-\mathfrak a}\|\tau w\|_2\right) =  C \ego^{1-2{\mathfrak a}} \|g\|_2+  C \,\ego^{\mathfrak a -\mathfrak b}\|\tau w\|_2
  \,.
\end{equation}
  
We now  consider the case
$k>\ego^{-(1/2-{\mathfrak b})}$. We begin by observing that
in this case 
$$
\mathfrak d^3 \geq 2 \quad \mbox{ and }  \quad  \frac{\Re \lambda}{\mathfrak d(\ego,k)^2} \leq (2^{-\frac 23}+ \gamma_0 \ego)\,  \Re
  \vartheta_1 \,.
$$
Hence there exists $\gamma_1< 1$ and $\ego_0$ such that for $\ego\in (0,\ego_0]$
$$
 \frac{\Re \lambda}{\mathfrak d(\ego,k)^2} \leq \gamma_1  \Re
  \vartheta_1 \,.
$$
We then use resolvent
estimates for $\Mg(\ego, I(\ego,k, \mathfrak d(\ego,k)) ,0)$. Thus, writing
\begin{equation*}
 \Big(\Mg(\ego, I(\ego,k, \mathfrak d(\ego,k)) ,0,0)-
  \frac{\lambda}{  \mathfrak d(\ego,k)^2}\Big)\widehat w= 
(  \mathfrak d (\ego,k))^{-2}\Big(\hat h  -  i\ego (  \mathfrak d
(\ego,k))^{-2}\beta  \chi^{\mathfrak d}(\tau) \tau^2 \, \widehat w \Big) \,,
\end{equation*}
we may use  \eqref{eq:88} (with $\beta=0$)  and the fact that
\begin{displaymath}
   (\mathfrak d
(\ego,k))^{-4}\|\chi^{\mathfrak d}(\tau) \tau^2\hat{w}\|_2\leq\ego^{-b}(\mathfrak d
(\ego,k))^{-3}\|\tau \hat{w}\|_2
\end{displaymath}
to obtain 
\begin{equation}
\label{eq:101}
  \|w\|_2\leq C \left( \|g\|_2+\ego^{1-b}\|\tau w\|_2\right) 
  \,.
\end{equation}
Since
\begin{displaymath}
\begin{array}{l}
  \Im\langle w,(\CC_{\ego,k}-\lambda)w\rangle\\
  \quad =\|\tau^{1/2}w\|_2^2+\ego\beta\|\chi^{1/2}(\ego^b \tau) \tau
  w\|_2^2+\ego\big\|[|\sigma+\ego^{-{\mathfrak a}}|^2\tau]^{1/2}w\big\|_2^2 -\Im\lambda\|w\|_2^2 \,,
  \end{array}
\end{displaymath}
we obtain,  with the aid of the inequality 
 $$
 \|\chi^{1/2}(\ego^b \tau) \tau w\|_2^2 \leq \sqrt{2}\, \ego^{-b} \, \| \tau^\frac 12 w\|^2_2\,,
$$
and the fact that $\Im \lambda \leq \gamma_0 \ego_0\,$,  the estimate
\begin{equation}
  \label{eq:102}
\|\tau^{1/2}w\|_2 \leq C(\|g\|_2+\|w\|_2)\,.
\end{equation}
Furthermore, as
\begin{displaymath}
   \Re\langle w,(\CC_{\ego,k}-\lambda)w\rangle=\|\partial_\tau w\|_2^2+\ego\|\partial_\sigma w\|_2^2-\Re\lambda\|w\|_2^2 \,,
\end{displaymath}
we readily deduce that
\begin{equation}
\label{eq:103}
\ego^{1/2}\|\partial_\sigma w\|_2+ \|\partial_\tau w\|_2\leq C(\|g\|_2+\|w\|_2)\,.
\end{equation}
   Finally, as
\begin{displaymath} 
\begin{array}{ll}
   \Im\langle\tau w,(\CC_{\ego,k}-\lambda)w\rangle
   & =\|\tau w\|_2^2+\ego\beta\|\chi^{1/2}(\ego ^b\tau) \tau^{3/2}w\|_2^2\\
   &\quad +\ego\big\|[|\sigma+\ego^{-{\mathfrak a}}| \tau] w\big\|_2^2\\ &\quad  -\Im\lambda\|\tau^{1/2}w\|_2^2 +
   2\, \Im\langle w,w_\tau\rangle\,, 
   \end{array}
\end{displaymath}
we may use \eqref{eq:102} and \eqref{eq:103} to 
establish that 
\begin{equation}
\label{eq:104}
\|\tau w\|_2 \leq C(\|g\|_2+\|w\|_2)\,.
\end{equation}
Substituting the above into either \eqref{eq:100} or \eqref{eq:101}
we get the existence of  $\ego_0>0$, such that for any $k$,  any $\ego
\in (0,\ego_0]$, and $\lambda \in B(\vartheta_1, \gamma_0 \ego)\cap \rho
(\CC_{\ego,k})$, 
\begin{equation}
  \label{eq:105}
\|(\CC_{\ego,k}-\lambda)^{-1}\|\leq \frac{C}{\ego^{1-2{\mathfrak a}}} \,.
\end{equation}
Using the discreteness of the spectrum in $B(\vartheta_1,
  \gamma_0 \ego)$ we indeed get  from the above  uniform estimate that
$\sigma(\CC_{\ego,k})\cap B(\vartheta_1,  \gamma_0 \ego)=\emptyset\,$.
By the above and \eqref{eq:103} we then have
\begin{equation}
  \label{eq:106}
\|\partial_\sigma(\CC_{\ego,k}-\lambda)^{-1}\|\leq \frac{C}{\ego^{3/2-2 {\mathfrak a}}} \,.
\end{equation}
Hence we have 
established
\begin{equation}
  \label{eq:107}
\|(\CC_{\ego,k}-\lambda)^{-1}\|+\ego^{1/2}\|\partial_\sigma(\CC_{\ego,k}-\lambda)^{-1}\|\leq \frac{C}{\ego^{1-2{\mathfrak a}}}\,.
\end{equation}

\paragraph{Estimation of  $\|(\CC_{\ego}-\lambda)^{-1}\|$ }~\\
From \eqref{eq:107} it follows
that there exists $\ego_0>0$, such that for any $k$ and any $\ego \in (0,\ego_0]$
\begin{displaymath}
  \|\ego(\partial^2_\sigma\phi_k^\ego) (\CC_{\ego,k}-\lambda)^{-1}\phi_k^\ego g\|_2\leq
  C\,\ego^{2{\mathfrak b}+2{\mathfrak a}}\, \|{\mathbf 1}_{S_k}g\|_2 \,,
\end{displaymath}
whereas from \eqref{eq:106} it follows that (for $k=0$ we write
$S_0^+$ instead of $S_k$)
\begin{displaymath}
    \|\ego(\partial_\sigma\phi_k^\ego) \partial_\sigma(\CC_{\ego,k}-\lambda)^{-1}\phi_k^\ego g\|_2 \leq  C\, \ego^{{\mathfrak b}+2{\mathfrak a}-1/2}\,\|{\mathbf 1}_{S_k}g\|_2 \,.
\end{displaymath}
Since
\begin{displaymath}
  \langle\ego[\partial^2_\sigma,\phi_k^\ego](\CC_{\ego,k}-\lambda)^{-1}\phi_k^\ego
  g,\ego[\partial^2_\sigma,\phi_m^\ego](\CC_{\ego,k}-\lambda)^{-1}\phi_m^\ego g\rangle =0
\end{displaymath}
whenever $|k-m|\geq2$, we may conclude that 
\begin{displaymath}
\begin{array}{ll}
  \|\sum_{k=0}^\infty\ego[\partial^2_\sigma,\phi_k^\ego](\CC_{\ego,k}-\lambda)^{-1}\phi_k^\ego g\|_2^2 & \leq
   4 C\,\ego^{ 2({\mathfrak b}+2{\mathfrak a}-1/2)}\,\Big(\|{\mathbf 1}_{S_0^+}g\|_2^2+\sum_{k=1}^\infty\|{\mathbf
    1}_{S_k}g\|_2^2\Big) \\ &  \leq 4 C\, \ego^{2({\mathfrak b}+2{\mathfrak a}-1/2)}\, \|g\|_2^2 \,.
    \end{array}
\end{displaymath}
Consequently, by (\ref{eq:96}c) we obtain that
\begin{equation}
\label{eq:108}
  \lim_{\ego\to0} \Eg_\CC =0 \,.
\end{equation}

To complete the proof we use  the fact that by \eqref{eq:108} the operator
$I+\Eg_{\CC}$ is invertible for sufficiently small $\ego$ to obtain
\begin{displaymath}
  (\CC_\ego -\lambda)^{-1}=\Rg^{app}_{\CC}(I+\Eg_{\CC})^{-1}\,.
\end{displaymath}
Hence, we can conclude from \eqref{eq:107} the existence of $\ego_0>0$
and $C>0$ such that, for $\ego \in (0,\ego_0]$,
\begin{displaymath}
  \|(\CC_\ego-\lambda)^{-1}\|\leq 2\,  \|\Rg^{app}_{\CC}\|\leq 8 \sup_{k\geq 0}
  \|(\CC_{\ego,k}-\lambda)^{-1}\|\leq \frac{C}{\ego^{1-2{\mathfrak a}}} \,.
\end{displaymath}
This completes the proof of (\ref{eq:94}a). The proof of
(\ref{eq:94}b) easily follows from the fact that
\begin{displaymath}
   \Re\langle w,(\CC_\ego-\lambda)w\rangle=\|w_\tau\|_2^2+\ego\|w_\sigma\|_2^2-\Re\lambda\|w\|_2^2 \,,
\end{displaymath}
for all $w\in D(\CC_\ego)$. 
\end{proof}

\subsection{Resolvent estimates for $\B_\ego$ }
\label{sec:first-simpl-oper}
Let 
$$
\Lambda_\gamma^1(\ego)=\vartheta_1 +\gamma \ego \lambda_1\mbox{  for some } \gamma\in[0,1]\,,
$$ 
where $\lambda_1$ is given by \eqref{eq:42}. Let further $b$, in the
definition of $\B_\ego$ (see  \eqref{eq:92}) satisfy 
\begin{equation}\label{condsurb}
\frac 12 < b <  \frac 34\,,
\end{equation}
 and $r(\ego)$
satisfy, for some $q <1/6\,$,
\begin{equation}\label{condsurr}
\lim_{\ego \ar 0} r(\ego) =0\,\mbox{ and }  \lim_{\ego \ar 0} \ego^{-q} r(\ego)=+\infty\,.
\end{equation}
In the following we prove the inclusion of $\partial
B(\Lambda_\gamma^1(\ego),r(\ego) \ego)$ in the resolvent set of $\B_\ego$
and obtain a bound on the resolvent norm  there.

\begin{proposition}
  \label{prop:simplified-V1}
Under the previous conditions, there
  exist positive $C$ and $\ego_0$ such that $\partial B (\Lambda_\gamma^1,r(\ego)\ego)\subset
  \rho (\B_\ego)$ for all $\ego\in (0,\ego_0]$ and $\gamma \in [0,1]$. Furthermore, the
  inequality
\begin{equation}
\label{eq:109}
    \|( \B_\ego -\lambda)^{-1}\| \leq \frac{C}{(r(\ego) +1-\gamma)\ego} \,,
  \end{equation}
    holds true.
\end{proposition}
\begin{proof}~
\paragraph{Construction of the right approximate resolvent.}~\\
We  introduce a $C^\infty$ partition of unity $(\zeta_-,\eta , \zeta_+)$ of $\mathbb
R$ such that 
  \begin{equation}
\label{eq:110}
    \eta(x)=
    \begin{cases}
      1 &\mbox{ if } |x|<1\,,\\
      0 &\mbox{ if }  |x|>2 \,,
    \end{cases}
  \end{equation}
and
\begin{equation}
\label{eq:111}
\zeta_+(x)=0 \mbox{ if } x <1\,,\,\zeta_-(x)=\zeta_+(-x)\,.
\end{equation}
Let further 
\begin{equation}\label{eq:53a}
\eta_\ego(\sigma)=\eta(\ego^{\mathfrak a}\sigma)
\mbox{ and } \zeta_\ego^\pm(\sigma)=\zeta_\pm(\ego^{\mathfrak a}\sigma)
\end{equation} for some $\mathfrak a$ satisfying
\begin{equation}\label{condsurmathfraka}
 1/6<{\mathfrak a}<(1-q)/4\,.
 \end{equation}

Next, let $S_N=(-2\ego^{-{\mathfrak a}},2\ego^{-{\mathfrak a}})\times\R_+$
and $\CC_N$ denote the Dirichlet realization in $S_N$ associated with
the differential operator
given by \eqref{eq:92}\,.\\
Let further $S_D^+= (\ego^{-{\mathfrak a}},+\infty)\times\R_+$, $S_D^-=
(-\infty,-\ego^{-{\mathfrak a}})\times\R_+$, and $\CC_D^\pm$ denote the
Dirichlet realizations in $S_D^\pm$ associated with the differential
operator given by \eqref{eq:92}.  The corresponding domains are
\begin{displaymath}
  D(\CC_D^\pm) = \{u\in H^2(S_D^\pm)\cap H^1_0(S_D^\pm)\,| \,\sigma^2(1+\tau)u\in
  L^2(S_D^\pm) \}\,.
\end{displaymath}
We can now  formally introduce the approximate resolvent in the form
\begin{equation}
\label{eq:112}
  \Rg^{app}_\B = \eta_\ego(\CC_N-\lambda)^{-1}\eta_\ego+
  \zeta_\ego^+(\CC_D^+-\lambda)^{-1}\zeta_\ego^+ + \zeta_\ego^-(\CC_D^--\lambda)^{-1}\zeta_\ego^- \,.
\end{equation}
\paragraph{Estimation of $ \|(\CC_D^\pm -\lambda)^{-1}\|$.}~\\
By \eqref{eq:94}, observing that $B(\Lambda_\gamma^1, r(\ego)\ego) \subset B(\vartheta_1,\gamma_0 \ego)$ for some $\gamma_0 >\gamma$ and $\ego_0$ small enough,  we have, for all  $\ego\in (0,\ego_0]$,  
\begin{equation}
\label{eq:113}
  \|(\CC_D^\pm-\lambda)^{-1}\|+ \ego^{1/2} \|\partial_\sigma(\CC_D^\pm-\lambda)^{-1}\|\leq
  \frac{C}{\ego^{1-2{\mathfrak a}}}\,.
\end{equation}
Note that the estimates for $\CC_D^-$ are deduced from the estimates
for $\CC_D^+$ by using the intertwining relation
$\CC_D^+=\Rg^{-1}\CC_D^-\Rg\,$, where $\Rg$ represents the
reflection $\sigma\to-\sigma\,$.
\paragraph{Estimation of $ \|(\CC_N-\lambda)^{-1}\|$.}~\\
It remains to obtain an estimate for $\|(\CC_N-\lambda)^{-1}\|$. 
 Let 
$w\in D(\CC_N)$ and $g\in L^2(S_N)$ satisfy
\begin{equation}
\label{eq:114}
   (\CC_N-\lambda)w=g\,.
\end{equation}
Let further 
\begin{displaymath}
  \tilde{\CC}_N =\LL_2^+(\ego) + \ego (\PP -\beta \tau_{m,2})= \LL_2^+(\ego) -\ego\partial^2_\sigma + e^{i\pi/6}\ego\sigma^2\tau_m \,,  
\end{displaymath}
where $\tau_m$ is given by \eqref{eq:41} and $\PP$ is given by \eqref{eq:27a}.\\
 We can now write
\eqref{eq:114} in the following form
\begin{displaymath}
   (\tilde{\CC}_N-\lambda)w=g-i\ego\sigma^2(\tau-e^{-i\pi/3}\tau_m)w \,.
\end{displaymath}
Applying the projection $\tilde{\Pi}_1$, given by \eqref{eq:84} (which
stands for ${\rm Id} \, \widehat \otimes\, \tilde{ \Pi} _1$  as in Section \ref{sec:2}) ,
to the above
balance yields
\begin{equation}
\label{eq:115}
  (\tilde{\CC}_N-\lambda)\tilde{\Pi}_1w = \tilde{\Pi}_1g
  -i\ego\sigma^2\tilde{\Pi}_1\big((\tau-e^{-i\pi/3}\tau_m)w\big) \,.
\end{equation}
From \eqref{eq:86}  it follows that
\begin{displaymath}
 \| \tilde{\Pi}_1\big((\tau-e^{-i\pi/3}\tau_m)\phi \big)-
 \tilde{\Pi}_1\big((\tau-e^{-i\pi/3}\tau_m)(I-\tilde{\Pi}_1)\phi \big)\|_2\leq C\ego\|\phi \|_2\,,\,  \forall \phi\in L^2(S_N)\,.
\end{displaymath}
Let $\phi(\sigma,\tau)= \sigma^2 w(\sigma,\tau)$. We can now conclude, as $|\sigma|\leq2\ego^{-{\mathfrak a}}$ in $S_N$, that
\begin{equation}
\label{eq:116}
\begin{array}{ll}
  \big\|\sigma^2\tilde{\Pi}_1\big((\tau-e^{-i\pi/3}\tau_m)w\big)\big\|&  \leq C(
  \| \tilde \Pi_1 (\tau-e^{-i\pi/3}\tau_m)\sigma^2(I-\tilde{\Pi}_1)w\|_2+\ego\|\sigma^2w\|_2)\\ & \leq
  C\, \ego^{-2{\mathfrak a}}\, (\|  (I-\tilde{\Pi}_1)w\|_2+\ego\|w\|_2)\,. 
\end{array}
\end{equation}
Since
\begin{equation}\label{eq:194a}
   (\tilde{\CC}_N-\lambda)\tilde{\Pi}_1=\left( -\ego\partial^2_\sigma +
   e^{i\pi/6}\ego\sigma^2\tau_m+(\tilde{\vartheta}_1-\lambda)\right)\, \tilde{\Pi}_1\,,
\end{equation}
we obtain from \eqref{eq:184a},   \eqref{eq:115}, \eqref{eq:116}, and the Riesz-Schauder
theory, that for $\lambda \in \partial B (\Lambda^1_\gamma, r(\ego)\ego)$,
\begin{equation}
\label{eq:117}
  \|(\tilde{\CC}_N-\lambda)^{-1} \widetilde \Pi_1\|\leq \frac{C}{(r(\ego) +1-\gamma)\ego} \,.
\end{equation}
By (\ref{eq:90}b) and Remark \ref{rem:Dirichlet-separable} we have, 
\begin{equation}
\label{eq:118} 
   \|(\tilde{\CC}_N-\lambda)^{-1}(I-\widetilde \Pi_1)\|\leq C \,. 
\end{equation}
Applying   (\ref{eq:117}) to \eqref{eq:115} yields, 
with the aid of \eqref{eq:116},
\begin{equation}
\label{eq:119}
  \|\tilde{\Pi}_1w\|_2\leq \frac{C}{(r(\ego) +1-\gamma)\ego}(\|g\|_2 +
  \ego^{1-2{\mathfrak a}}\|(I-\tilde{\Pi}_1)w\|_2+\ego^{2-2{\mathfrak a}}\|w\|_2)\,.
\end{equation}
 We now apply $I-\tilde{\Pi}_1$ to \eqref{eq:115} to obtain 
\begin{equation}
\label{eq:120}
  (\tilde{\CC}_N-\lambda)(I-\tilde{\Pi}_1)w = (I-\tilde{\Pi}_1)g -
  i\ego\sigma^2(I-\tilde{\Pi}_1)\big((\tau-e^{-i\pi/3}\tau_m)w\big) \,.
\end{equation}

Since the norm of $I-\tilde{\Pi}_1$ is uniformly bounded (see
\eqref{esttildev1}), we have
\begin{displaymath} 
  \big\|\sigma^2(I- \widetilde \Pi_1)\big((\tau-e^{-i\pi/3}\tau_m)w\big)\big\|_2\leq C\, \|\sigma^2(\tau-e^{-i\pi/3}\tau_m)w\|_2\leq
 C\ego^{-2{\mathfrak a}}\|(\tau-e^{-i\pi/3}\tau_m)w\|_2\,.
\end{displaymath}
In  the same manner we have obtained
\eqref{eq:104}  we can now obtain
\begin{equation}
\label{eq:121}
  \|\tau w\|_2 \leq C\, (\|w\|_2+\|g\|_2)\,. 
\end{equation}
Consequently,
\begin{equation}\label{eq:6.35}
   \big\|\sigma^2(I-\tilde{\Pi}_1)\big((\tau-e^{-i\pi/3}\tau_m)w\big)\big\|_2
  \leq C\ego^{-2{\mathfrak a} }(\|w\|_2+\|g\|_2)\,.
\end{equation}
We now apply  (\ref{eq:118})  to \eqref{eq:120} to obtain, with the aid
of (\ref{eq:116})  and the above inequality,
\begin{equation}
\label{eq:122}
   \|(I-\tilde{\Pi}_1)w\|_2\leq C(\|g\|_2 +  \ego^{1-2{\mathfrak a}}\|w\|_2)\,.
\end{equation}
Substituting the above into \eqref{eq:119} yields
\begin{displaymath}
  \|\tilde{\Pi}_1w\|_2\leq \frac{C}{(r(\ego) +1-\gamma)\ego}\, \|g\|_2 +  \frac{C\ego^{1-4{\mathfrak a}}}{r(\ego) +1-\gamma}\, \|w\|_2\,.
\end{displaymath}
The above together with \eqref{eq:122},
\eqref{condsurr}, and \eqref{condsurmathfraka} 
yield the existence of $\ego_0>0$ and $
C$ such that, for all $\ego\in (0,\ego_0]$, $\partial B (\Lambda_\gamma^1,r(\ego)\ego)$ belongs to 
$\rho (\CC_N)$ and
\begin{equation}
\label{eq:123}
  \|(\CC_N-\lambda)^{-1}\|\leq \frac{C}{(r(\ego) +1-\gamma)\ego}\,,\, \forall \lambda \in\partial B (\Lambda_\gamma^1,r(\ego)\ego)\,.
\end{equation}
\paragraph{A bound on $\partial_\sigma(\CC_N-\lambda)^{-1}$}~\\
As in the proof of Proposition \ref{lem:large-sigma} we need an
estimate for $\partial_\sigma(\CC_N-\lambda)^{-1}$. While \eqref{eq:103} still holds,
it is unsatisfactory in the present context.  Let $(w,g)$ satisfy
\eqref{eq:114}. To achieve a better estimate of $\|w_\sigma\|_2$, we
separately estimate $ \|\tilde{\Pi}_1w_\sigma\|_2$ and $ \|
(I-\tilde{\Pi}_1)w_\sigma\|_2$. \\
  
  To facilitate the estimation of $ \|\tilde{\Pi}_1w_\sigma\|_2$, we rewrite
  \eqref{eq:115}-\eqref{eq:194a} in the following manner
\begin{equation}\label{eq:64a}
   \ego \Big(-\partial^2_\sigma+ e^{i\pi/6}\tau_m\sigma^2 -\frac{\lambda-{\tilde \vartheta}_1}{\ego}\Big)\tilde{\Pi}_1w = \tilde{\Pi}_1g -i\ego\sigma^2\tilde{\Pi}_1\big((\tau-e^{-i\pi/3}\tau_m)w\big) \,.
\end{equation}

Taking the inner product of \eqref{eq:64a} with $\tilde{\Pi}_1w$
we obtain from the real part and \eqref{eq:116}
\begin{displaymath}
  \ego\|\tilde{\Pi}_1w_\sigma\|_2^2 \leq C \ego\|\tilde{\Pi}_1w\|_2^2 +
  \|\tilde{\Pi}_1w\|_2\,\left(\|\tilde{\Pi}_1g\|_2+\ego^{1-2{\mathfrak a}}
  \|(I-\tilde{\Pi}_1)w\|_2+\ego^{2-2{\mathfrak a}}\|w\|_2\right)\,.
\end{displaymath}
Hence,
\begin{displaymath}
  \|\tilde{\Pi}_1w_\sigma\|_2 \leq \hat C \left(  \|\tilde{\Pi}_1w\|_2 + \ego^{-1} \|\tilde{\Pi}_1g\|_2+\ego^{-2{\mathfrak a}}
  \|(I-\tilde{\Pi}_1)w\|_2+\ego^{1-2{\mathfrak a}}\|w\|_2 \right)\,.
\end{displaymath}
Using \eqref{eq:122} and \eqref{condsurmathfraka} we then obtain
\begin{displaymath}
\begin{array}{ll}
  \|\tilde{\Pi}_1w_\sigma\|_2 & \leq C\, \left(\|\tilde{\Pi}_1w\|_2+\ego^{1-4{\mathfrak a}}\|w\|_2+\ego^{-1}\|g\|_2 \right)\\
   & \leq  \tilde C  \, ( \| w\|_2 +\ego^{-1}  \| g\|_2))\,,
   \end{array}
\end{displaymath}
from which we deduce, with the aid of \eqref{eq:114} and  \eqref{eq:123}, 
\begin{equation}
\label{eq:124}
  \|\tilde{\Pi}_1w_\sigma\|_2 \leq \frac{C}{(r(\ego) +1-\gamma)\ego}\|g\|_2\,.
\end{equation}

To estimate $\|(I- \tilde{\Pi}_1)w_\sigma\|_2$, 
we now take the inner product of \eqref{eq:120} with $(I-\tilde{\Pi}_1)w$ to
obtain, with the aid of \eqref{eq:121}\,,
\begin{displaymath}
  \ego\|(I-\tilde{\Pi}_1)w_\sigma\|_2^2 \leq C \|(I-\tilde{\Pi}_1)w\|_2^2 +
  C\|(I-\tilde{\Pi}_1)w\|_2(\|g\|_2+\ego^{1-2{\mathfrak a}}  \|w\|_2)\,.
\end{displaymath}
Making use of \eqref{eq:122} then yields 
\begin{displaymath}
  \|(I-\tilde{\Pi}_1)w_\sigma\|_2\leq C(\ego^{-1/2}\|g\|_2+\ego^{1/2-2{\mathfrak a}}\|w\|_2)\leq \check C (\ego^{-1} \|g\|_2+ \|w\|_2) \,,
\end{displaymath}
which as above leads to
\begin{equation}
\label{eq:65b}
  \|(I-\tilde{\Pi}_1) w_\sigma\|_2 \leq \frac{C}{(r(\ego) +1-\gamma)\ego}\|g\|_2\,.
\end{equation}
Then  \eqref{eq:124} and \eqref{eq:65b} give  the existence of $C$ and
$\ego_0$ such that, for all $\ego\in (0,\ego_0]$, 
\begin{equation}
\label{eq:125}
  \|\partial_\sigma(\CC_N-\lambda)^{-1}\|\leq \frac{C}{(r(\ego) +1-\gamma)\ego}\,.
\end{equation}
\paragraph{The approximate resolvent}~\\
The preceding paragraphs prove that the approximate resolvent
$\Rg^{app}_{\B} $, introduced in \eqref{eq:112} is well defined. We now
prove that it serves as a good approximation for the resolvent. We note
that
\begin{subequations}
  \label{eq:126}
\begin{equation}
  ( \B_\ego -\lambda) \Rg^{app}_{\B} = I + \Eg_\B \,,
\end{equation}
where
\begin{equation}
   \Eg_\B = -\ego[\partial^2_\sigma,\eta_\ego](\CC_N-\lambda)^{-1}\eta_\ego
   -\ego[\partial^2_\sigma,\zeta_\ego^+](\CC_D^+-\lambda)^{-1}\zeta_\ego^+-\ego[\partial^2_\sigma,\zeta_\ego^-](\CC_D^--\lambda)^{-1}\zeta_\ego^-\,.
\end{equation}
\end{subequations}
As
\begin{displaymath}
  \ego[\partial^2_\sigma,\zeta_\ego^\pm](\CC_D^\pm-\lambda)^{-1}\zeta_\ego^\pm =
  \left(\ego^{1+2{\mathfrak a}}(\zeta^{\pm,\prime\prime})_\ego+ 2 \ego^{1+{\mathfrak a}}(\zeta^{\pm,\prime})_\ego \partial_\sigma\right)(\CC_D^\pm-\lambda)^{-1}\zeta_\ego^\pm \,,
\end{displaymath}
we obtain by \eqref{eq:113} that
\begin{equation}
\label{eq:127}
  \|\ego[\partial^2_\sigma,\zeta_\ego^\pm ](\CC_D^\pm-\lambda)^{-1}\zeta_\ego\|\leq C\ego^{3 {\mathfrak a}-1/2} \,.
\end{equation}
Furthermore, since
\begin{displaymath} 
  \ego[\partial^2_\sigma,\eta_\ego](\CC_N-\lambda)^{-1}\eta_\ego =
  (\ego^{1+2{\mathfrak a}}(\eta^{\prime\prime})_\ego+2 \ego^{1+{\mathfrak a}}(\eta^\prime)_\ego\,  \partial_\sigma)(\CC_N-\lambda)^{-1}\eta_\ego \,,
\end{displaymath}
we obtain from \eqref{eq:123} and \eqref{eq:125} 
\begin{displaymath}
  \|\ego[\partial^2_\sigma,\eta_\ego](\CC_N-\lambda)^{-1}\eta_\ego\|\leq C\frac{\ego^{\mathfrak a}}{r(\ego) +1-\gamma}
  \,.
\end{displaymath}
The above, together with \eqref{eq:127} and (\ref{eq:126}b), yields
\begin{displaymath}
  \|\Eg_\B\|\leq C \Big( \frac{\ego^{\mathfrak a}}{r(\ego) +1-\gamma}+ \ego^{3{\mathfrak a}-1/2}\Big)\,.
\end{displaymath}
Hence, \eqref{condsurr}, and \eqref{condsurmathfraka} imply that $
\|\Eg_\B\|$ tends to $0$ as $\ego\to 0$. Consequently, for sufficiently
small $\ego$, $I+\Eg_\B$ is invertible and we may use (\ref{eq:126}a)
to obtain the right inverse to $(\B_\ego -\lambda)$. \\
For $\lambda \in \rho(\B_\ego) \cap \partial B (\Lambda_\gamma ^1, r(\ego)\ego)$, this right
inverse is identical with the left inverse and we get
\begin{displaymath}
   \|( \B_\ego -\lambda)^{-1}\|\leq C\, \Rg^{app}_{\B} \leq \frac{\widehat C}{(r(\ego) +1-\gamma)\ego} \,.
\end{displaymath}
The spectrum of  $\B_\ego$ being discrete, we may conclude from the
above estimate that $\sigma(\B_\ego)\cap\partial B (\Lambda_\gamma^1, r(\ego)\ego)=\emptyset\,$,
which completes the proof of the proposition.
\end{proof}

For later reference we separately estimate the $\sigma$-derivatives of $(
\B_\ego -\lambda)^{-1}$
\begin{proposition}
  Under the conditions of Proposition \ref{prop:simplified-V1}, for
  any \break $\mathfrak a \in(1/6, (1-q)/4)$, there exists $\ego_0$ and
  $C_{\mathfrak a}$ such that,   for all $\lambda \in \partial B (\Lambda_\gamma^1, \ego
  r(\ego))$, we have  
  \begin{subequations} 
 \begin{equation}  
\label{eq:128}
 \| \partial_\sigma (\B_\ego-\lambda)^{-1}\| \leq C_{\mathfrak a} \Big(\frac{1}{\ego^{3/2-2{\mathfrak
       a}}} + \frac{1}{(r(\ego) +1-\gamma)\ego}\Big)\,,
 \end{equation}
and
\begin{equation}   
\label{eq:129}
\|\partial^2_{\sigma\sigma} (\B_\ego -\lambda)^{-1}\| \leq C_{\mathfrak a} \Big(\frac{1}{\ego^{2-2{\mathfrak
      a}}}+ \frac{1}{(r(\ego)+ 1-\gamma) \ego^{3/2}}\Big)\,.
\end{equation}
  \end{subequations}
\end{proposition}
\begin{proof}\strut
\paragraph{Estimation of $\partial_\sigma (\B_\ego -\lambda)^{-1}$}~\\
     Let $w\in D( \B_\ego)$ and $g\in L^2(\R^2_+)$
     satisfy $ ( \B_\ego -\lambda)w=g \,.  $ Clearly,
 \begin{displaymath}
     ( \B_\ego -\lambda)(\eta_\ego w)= \eta_\ego g- 2\ego^{1+{\mathfrak a}}(\eta^\prime)_\ego w_\sigma  -  \ego^{1+2{\mathfrak a}}(\eta^{\prime\prime})_\ego \, w  \,,
 \end{displaymath}
 where $(\eta^\prime)_\ego(\sigma)=\eta^\prime(\ego^{\mathfrak a}\sigma)$ and
 $(\eta^{\prime\prime})_\ego(\sigma)=\eta^{\prime\prime}(\ego^{\mathfrak a}\sigma)$. \\
 By \eqref{eq:125} we then have
 \begin{equation}
 \label{eq:130}
   \|(\eta_\ego w)_\sigma\|_2\leq \frac{C}{(r(\ego) +1-\gamma)\ego}[ \|\eta_\ego g\|_2+
   \ego^{1+{\mathfrak a}}\|(\eta^\prime )_\ego\, w_\sigma\|_2  +  \ego^{1+2{\mathfrak a}}\| (\eta^{\prime\prime})_\ego\, w\|_2] \,.
 \end{equation}
 Similarly, as
 \begin{equation*}
     ( \B_\ego -\lambda)(\zeta^\pm_\ego w)= \zeta^\pm_\ego g- 2\ego^{1+{\mathfrak
         a}}(\zeta^{\pm,\prime})_\ego w_\sigma  -  \ego^{1+2{\mathfrak
         a}}(\zeta^{\pm,\prime\prime})_\ego w  \,, 
 \end{equation*} 
 we may use \eqref{eq:113} to obtain 
 \begin{equation}
\label{eq:131}
   \|(\zeta^\pm_\ego w)_\sigma\|_2\leq \frac{C}{\ego^{3/2-2{\mathfrak a}}}[ \|\zeta^\pm_\ego g\|_2+
   \ego^{1+{\mathfrak a}}\|(\zeta^{\pm,\prime})_\ego w_\sigma\|_2  +  \ego^{1+2{\mathfrak a}}\|(\zeta^{\pm,\prime\prime})_\ego \,w\|_2] \,.
 \end{equation}
 Combining \eqref{eq:130} and \eqref{eq:131} yields (recalling that
 ${\mathfrak a}>1/6$ and $q < \frac 1 6$)
  \begin{displaymath}
 \|w_\sigma\|_2\leq C\Big(\frac{1}{\ego^{3/2-2{\mathfrak a}}} + \frac{1}{(r(\ego) +1-\gamma)\ego}\Big)[
 \|g\|_2 +  \ego^{1+2{\mathfrak a}}\|w\|_2] \,. 
 \end{displaymath}
 With the aid of \eqref{eq:109} we then obtain, for any pair $(w,g)$ satisfying $
    ( \B_\ego -\lambda)w=g \,,$
 \begin{equation*}
 \|w_\sigma\|_2\leq C\Big(\frac{1}{\ego^{3/2-2{\mathfrak a}}} +
 \frac{1}{(r(\ego) +1-\gamma)\ego}\Big) 
 \|g\|_2 \,,
 \end{equation*}
from which~\eqref{eq:128} easily follows.

\paragraph{Estimation of $\partial^2_{\sigma\sigma} (\B_\ego -\lambda)^{-1}$}~\\
For the same pair $(w,g)$, an integration by parts  yields
 \begin{displaymath}
   -\Re\langle w_{\sigma\sigma},( \B_\ego-\lambda)w\rangle = \|w_{\tau\sigma}\|_2^2 +
   \ego\|w_{\sigma\sigma}\|_2^2 + 2\ego\Im\langle w_\sigma,\sigma \tau w \rangle -\Re\lambda \|w_\sigma\|_2^2\,.
 \end{displaymath}
Note here that $<w_{\sigma\sigma}, w_{\tau\tau}> = \| w_{\sigma\tau}\|^2$ for all
  $w\in H^2(\R^2_+)\cap H^1_0(\R^2_+)$ and hence also for all $w\in D(\B_\ego)$.\\
  Hence,
 \begin{equation}
 \label{eq:132}
   \|w_{\sigma\sigma}\|_2\leq \frac{C}{\ego^{1/2}}(\|w_\sigma\|_2
   +\ego\|\sigma\tau w\|_2+\ego^{-1/2}\|g\|_2)\,.  
 \end{equation}
 As
 \begin{displaymath}
    \Im\langle\tau w,( \B_\ego-\lambda)w\rangle = \Im\langle w,w_\tau\rangle  + \|\tau
    w\|_2^2+\beta\ego\|\tau^{3/2}\chi(\ego^b \tau) w\|_2^2 \,+
    \ego\|\sigma\tau w\|_2^2 -\Im\lambda \|\tau^{1/2}w\|_2^2\,,
 \end{displaymath}
and since both \eqref{eq:103} and \eqref{eq:104} hold in this case as well,
 we easily obtain, in view of the fact that $\tau\chi (\ego^b\tau) \leq2\ego^{-b}$, that
 \begin{equation}\label{eq:78a}
   \|\sigma\tau w\|_2 \leq \frac{C}{\ego^{1/2}}(\|g\|_2 +\|w\|_2)\,.
 \end{equation}
 Substituting \eqref{eq:128}  and \eqref{eq:78a} into
 \eqref{eq:132} then yields
 \begin{equation*}
  \|w_{\sigma\sigma}\|_2\leq C\Big(\frac{1}{\ego^{2-2{\mathfrak
        a}}}+\frac{1}{(r(\ego)+1-\gamma)\ego^{3/2}}\Big)\|g\|_2\,, 
 \end{equation*}
for any pair $(w,g)$ which satisfies  $( \B_\ego -\lambda)w=g \,,$
which completes the proof of \eqref{eq:129}.
  \end{proof}

For later reference we also need the following additional estimate:
\begin{proposition}
\label{cor:large-sigma}
  Under the conditions preceding Proposition \ref{prop:simplified-V1},  for
  all $ {\mathfrak a}$  in\break  $(1/6,(1-q)/4)$, there exists
  $C_{\mathfrak a}>0$ and $\ego_0>0$ 
  such that, for any $\ego\in (0,\ego_0]$, 
  \begin{equation}
    \label{eq:133}
  \|{\mathbf 1}_{|\sigma|\geq 2\ego^{-{\mathfrak a}}}( \B_\ego-\lambda)^{-1}\|+ \ego^{1/2} \|{\mathbf 1}_{|\sigma|\geq 2\ego^{-{\mathfrak a}}}\partial_\sigma( \B_\ego-\lambda)^{-1}\|\leq
  \frac{C_{\mathfrak a}}{\ego^{1-2{\mathfrak a}}}\,.
  \end{equation}
\end{proposition}
\begin{proof}
  Since for sufficiently small $\ego$ we have
  \begin{displaymath}
    ( \B_\ego-\lambda)^{-1} = \Rg^{app}_{\B}(I+\Eg_\B)^{-1}\,,
  \end{displaymath}
it follows by \eqref{eq:112} that
\begin{displaymath}
 {\mathbf 1}_{\sigma\geq 2\ego^{-{\mathfrak a}}} ( \B_\ego-\lambda)^{-1} =  {\mathbf 1}_{\sigma\geq
   2\ego^{-{\mathfrak a}}} (\CC_D^+-\lambda)^{-1}\zeta_\ego^+ (I+\Eg_\B)^{-1}\,. 
\end{displaymath}
By \eqref{eq:113} we then have
\begin{displaymath}
   \|{\mathbf 1}_{\sigma\geq 2\ego^{-{\mathfrak a}}}( \B_\ego-\lambda)^{-1}\|\leq   \frac{C_{\mathfrak a}}{\ego^{1-2{\mathfrak a}}}\,.
\end{displaymath}
In a similar manner we write 
\begin{displaymath}
 {\mathbf 1}_{\sigma\geq 2\ego^{-{\mathfrak a}}} \partial_\sigma( \B_\ego-\lambda)^{-1} =  {\mathbf 1}_{\sigma\geq
   2\ego^{-{\mathfrak a}}} \partial_\sigma(\CC_D^+-\lambda)^{-1}\zeta_\ego^+ (I+\Eg_\B)^{-1}\,. 
\end{displaymath}
Once again by \eqref{eq:113} we obtain
\begin{displaymath}
\ego^{\frac 12}\,  \|{\mathbf 1}_{\sigma\geq 2\ego^{-{\mathfrak a}}}\partial_\sigma( \B_\ego-\lambda)^{-1}\|\leq
  \frac{C_{\mathfrak a}}{\ego^{1-2{\mathfrak a}}}\,.
\end{displaymath}
\end{proof}

\subsection{Curvature effects}
\label{sec:modification}
In the following, we estimate the effect of some of the error terms 
in \eqref{eq:36} and \eqref{eq:21a}.  Since the
  estimation of these terms is complex, it is preferable to consider
  them as modifications of $\B_\ego$ and not in the context of the
  original operator $\A_h$, which is addressed in Section \ref{sec:upper}.

\subsubsection{Effect 1} The first effect is generated by the
first error term in \eqref{eq:21a}. 
   \begin{proposition}
\label{lem:modified}
Consider on $D( \B_\ego)$ the operator
 \begin{equation}
   \label{eq:134}
 \hat{\B}_\ego = \B_\ego - \theta\ego^2\tau\chi(\ego^{\tilde b}\tau)\partial^2_\sigma \,,
 \end{equation}
 where $\theta\in\R$, $\tilde{b}\in (0,1/2 -q)$ and  $\chi\in C^\infty(\R_+,[0,1])$ is given by
 \eqref{eq:59}.
 Then, there exist positive $C$ and
 $\ego_0$ such that, for every $\ego \in (0, \ego_0]$, $\partial
 B(\Lambda_\gamma^1,r(\ego) \ego)\cap\sigma(\hat{\B}_\ego )=\emptyset$ and 
 \begin{subequations}
   \label{eq:135}
    \begin{equation}
  \sup_{\lambda\in\partial B(\Lambda_\gamma^1,r(\ego) \ego)}\|(\hat{\B}_\ego -\lambda)^{-1}\| \leq \frac{C}{(r(\ego) +1-\gamma)\ego}  \,,\,\forall \lambda \in
  \partial B(\Lambda_\gamma^1,r(\ego) \ego)\,. 
 \end{equation}
Furthermore, we have, for all $\lambda \in
  \partial B(\Lambda_\gamma^1,r(\ego) \ego)$,
\begin{equation}
   \|\partial_\tau(\hat{\B}_\ego -\lambda)^{-1}\|  + \ego^{1/2}\|\partial_\sigma(\hat{\B}_\ego -\lambda)^{-1}\| \leq \frac{C}{(r(\ego) +1-\gamma)\ego}\,,
\end{equation}
and
\begin{equation}
   \|\partial_{\tau\tau}(\hat{\B}_\ego -\lambda)^{-1}\|  + \ego\|\partial_{\sigma\sigma}(\hat{\B}_\ego -\lambda)^{-1}\| \leq \frac{C}{(r(\ego) +1-\gamma)\ego}\,.
\end{equation}
 \end{subequations}
   \end{proposition}
   \begin{proof}
     For sake of brevity we use the notation
     $\tilde{\chi}(\tau)=\chi(\ego^{\tilde b}\tau)$ where $\chi$ is given by
     \eqref{eq:59}.  We keep the same notation as in the previous
     subsection for the cut-off functions given by
     \eqref{eq:110}-\eqref{condsurmathfraka} .

 Let $u\in D( \B_\ego)$ and $g \in L^2(\R^2_+)$ satisfy
 \begin{displaymath}
    (\hat{\B}_\ego -\lambda)u=g \,.
 \end{displaymath}
 We rewrite the above balance in the following manner
 \begin{equation}
 \label{eq:136}
    ( \B_\ego -\lambda)u= g + \theta\ego^2\tau\tilde{\chi}u_{\sigma\sigma}\,.
 \end{equation}
Keeping in mind that $|\tau \tilde \chi (\tau)| \leq 2 \ego^{-\tilde b}$ we use \eqref{eq:129}
with $\mathfrak a \in (\frac{1}{4}-\frac q2 , \frac{1}{4}-\frac q4)$ to obtain  
 \begin{displaymath}
    \|u_{\sigma\sigma}\|_2   \leq
   C\Big(\ego^{2{\mathfrak
       a}}+\frac{\ego^{1/2}}{(r(\ego)+1-\gamma)}\Big)(\ego^{-2}\|g
   \|_2+\ego^{-\tilde{b}}\|u_{\sigma\sigma}\|_2)\,. 
  \end{displaymath}
 Since $\tilde{b}\in (0,1/2 -q)$ we may conclude that
 \begin{equation}
   \label{eq:137}
 \|u_{\sigma\sigma}\|_2 \leq C\Big(\frac{1}{\ego^{2-2{\mathfrak
       a}}}
 +\frac{1}{(r(\ego)+1-\gamma)\ego^{3/2}}\Big)\| g \|_2\,.   
 \end{equation}
 
 Applying \eqref{eq:109} to \eqref{eq:136} yields
 \begin{displaymath}
   \|u\|_2 \leq\frac{C}{(r(\ego) \ego+1-\gamma)}(\|g\|_2+\ego^{2-\tilde{b}}\|u_{\sigma\sigma}\|_2)\,.
 \end{displaymath}
 We first establish (\ref{eq:135}a) for $\lambda\in \rho (\hat \B_\ego)\cap \partial
 B(\Lambda^1_\gamma, \ego r(\ego))$, by substituting \eqref{eq:137} (observing
 that $ \tilde b < \frac12 - q < 2 \mathfrak a$) into the above
 inequality. Since the spectrum of $\hat\B_\ego$ is discrete,
 we can deduce, as for $\B_\ego$, that $\sigma(\hat\B_\ego)\cap\partial B(\Lambda^1_\gamma,
 \ego r(\ego))=\emptyset$, and hence, that (\ref{eq:135}a) is satisfied
 without restriction.

The proof of (\ref{eq:135}b) follows immediately from (\ref{eq:135}a)
and the identity
\begin{displaymath}
  \Re\langle u,  (\hat{\B}_\ego -\lambda)u\rangle = \|u_\tau\|_2^2 + \ego \|u_\sigma\|_2^2 +
  \theta\ego^2 \|[\tilde{\chi}\tau]^{1/2}u_\sigma\|_2^2-  \Re\lambda\|u\|_2^2 \,,
\end{displaymath}
which holds for every $u\in D(\hat{\B}_\ego )$. To prove (\ref{eq:135}c)
we use \eqref{eq:137} and the following identity, that holds for
every $u\in D(\hat{\B}_\ego )$,
\begin{displaymath}
\begin{array}{l}
  -\Re\langle u_{\tau\tau},  (\hat{\B}_\ego -\lambda)u\rangle \\
  \qquad \quad  = \|u_{\tau\tau}\|_2^2 + \ego
  \|u_{\tau\sigma}\|_2^2 +
  \theta\ego^2 \big(\|[\tilde{\chi}\tau]^{1/2}u_{\sigma\tau}\|_2^2-
  \Re\langle[\tilde{\chi}\tau]^\prime u_\sigma,u_{\tau\sigma}\rangle\big)-\Re\lambda\|u\|_2^2 \,,
  \end{array}
\end{displaymath}
together with (\ref{eq:135}a,b) and the fact that
$\|[\tilde{\chi}\tau]^\prime\|_\infty\leq1+2\|\chi^\prime\|_\infty$. 
  \end{proof}
We shall also need in the sequel the following estimate
  \begin{proposition}
  Under the conditions of Proposition \ref{lem:modified}, for any 
   ${\mathfrak a}$ in the interval  $(1/6,(1-q)/4)$ there exists $C_{\mathfrak a}>0$ and $\ego_0>0$
   such that for any $\ego\in (0,\ego_0]$,  and $\lambda\in\partial B(\Lambda_\gamma^1,r(\ego) \ego)$,
  \begin{equation}
\label{eq:138}
  \|{\mathbf 1}_{|\sigma|\geq 2\ego^{-{\mathfrak
        a}}}(\hat{\B}_\ego-\lambda)^{-1}\|+ \ego^{1/2} \|{\mathbf
    1}_{|\sigma|\geq 2\ego^{-{\mathfrak a}}}\partial_\sigma(\hat{\B}_\ego-\lambda)^{-1}\|\leq
  \frac{C_{\mathfrak a}}{\ego^{1-2{\mathfrak a}}}\,. 
  \end{equation}
\end{proposition}
\begin{proof}
 Let, as in the previous proof,  $u\in D(\hat{\B}_\ego)$, $\lambda \in \mathbb
 C$ and  $ g\in L^2(\R^2_+)$ 
 such that $ g=(\hat{\B}_\ego -\lambda)u$.  Since 
\begin{displaymath}
     ( \B_\ego -\lambda)u= g - \theta\ego^2\tau\tilde{\chi}u_{\sigma\sigma}  \,,
 \end{displaymath} 
we  obtain from \eqref{eq:133} that
\begin{displaymath}
  \|{\mathbf 1}_{|\sigma|\geq 2\ego^{-{\mathfrak a}}}u\|_2 \leq
  \frac{C}{\ego^{1-2{\mathfrak a}}}(\| g\|_2+\ego^{2-\tilde b}\|u_{\sigma\sigma}\|)\,.
\end{displaymath}
By \eqref{eq:137} we then have, using the fact that $\tilde  b<1/2-q\,$, 
\begin{displaymath}
   \|{\mathbf 1}_{|\sigma|\geq 2\ego^{-{\mathfrak a}}}u\|_2 \leq
   C\Big(\frac{1}{\ego^{1-2{\mathfrak a}}}+\frac{1}{r(\ego) \ego^{ 1/2+\tilde b-2{\mathfrak a}}}\Big)\|g\|_2 \leq  \frac{\tilde C}{\ego^{1-2{\mathfrak a}}}\, \| g\|_2 \,.
\end{displaymath}
In a similar manner we show that
\begin{displaymath}
   \ego^{1/2} \|{\mathbf 1}_{|\sigma|\geq 2\ego^{-{\mathfrak a}}}u_\sigma\|_2\leq
  \frac{C_{\mathfrak a}}{\ego^{1-2{\mathfrak a}}}\, \|g\|_2\,.
\end{displaymath}
\end{proof}

We finally establish an asymptotic estimate, which is needed in
Section \ref{sec:upper}. It is valid in a region where $\tau$ is large, but the
the cutoff function $\chi$ is still $1$ (i.e. $1\ll\tau\leq\ego^{-b}$).
\begin{proposition}
\label{lem:large-tau-V1}
Let $0<a<b$. 
Then, there exist positive $C_a$
and $\ego_0$ such that for all $\ego\in (0,\ego_0]$ and $\lambda\in\partial
B(\Lambda_\gamma^1,r(\ego) \ego)$,
\begin{subequations}
\label{eq:139}
  \begin{multline}
    \|{\mathbf 1}_{\tau\geq\ego^{-a}}(\hat{\B}_\ego -\lambda)^{-1}\| +
    \ego^{a/2}\|{\mathbf 1}_{\tau\geq\ego^{-a}}\partial_\tau(\hat{\B}_\ego
    -\lambda)^{-1}\| + \\
   \ego^{(a+1)/2}\|{\mathbf 1}_{\tau\geq\ego^{-a}}\partial_\sigma(\hat{\B}_\ego -\lambda)^{-1}\|
   \leq C_a\ego^a \,,
  \end{multline}
and
\begin{equation}
   \|{\mathbf 1}_{\tau\geq\ego^{-a}}\partial_\tau^2(\hat{\B}_\ego -\lambda)^{-1}\| +
   \ego\|{\mathbf 1}_{\tau\geq\ego^{-a}}\partial_\sigma^2(\hat{\B}_\ego -\lambda)^{-1}\| \leq
   C_a \,.
\end{equation}
\end{subequations}
\end{proposition}
\begin{proof}
Let $\zeta_+$ be given by \eqref{eq:111} and for some $\xi >0$, 
     $\zeta_\xi(\tau)=\zeta_+(2 \tau/\xi)$. 
   Let $u \in D(\hat{\B}_\ego)$, $\lambda\in\partial
B(\Lambda_\gamma^1,r(\ego) \ego)$, and
 $g\in L^2(\R^2_+)$ satisfy
 \begin{displaymath}
    (\hat{\B}_\ego -\lambda)u=g \,.
 \end{displaymath}
 \paragraph{Proof of (\ref{eq:139}a)}
As the identity
\begin{multline}
\label{eq:140}
\Re\langle (\hat{\B}_\ego -\lambda)u\,,\, \zeta^2u\rangle + \Im
\langle(\hat{\B}_\ego -\lambda)u\rangle\,,\, \zeta^2u\rangle\\  =
\|\partial_\tau(\zeta u)\|_2^2
-(\Im\lambda+\Re\lambda)\|\zeta u\|_2^2 +
\ego\|\zeta\partial_\sigma
u\|_2^2+\theta\ego^2\|\tau^{1/2}\zeta\tilde{\chi}^{1/2}\partial_\sigma
u\|_2^2\\
\qquad \qquad   -\|\zeta^\prime u\|_2^2  + 2\Im\langle
\zeta^\prime u,\partial_\tau(\zeta u)\rangle +
\|\tau^{1/2}\zeta
u\|_2^2+ \beta\ego\|\tau \zeta ({\chi(\ego^b\tau))^\frac 12}
u\|_2^2 +\ego\|\tau^{1/2}\sigma\zeta u\|_2^2\,,
\end{multline}
holds for any $C^\infty$ function $\zeta$ with support in $\mathbb R_+$, we get, for $\zeta=\zeta_\xi$, 
\begin{displaymath} 
  \|\tau^{1/2}\zeta_\xi u\|_2^2- (\Im\lambda+\Re\lambda)\|\zeta_\xi u\|_2^2 \leq
  4 \|\zeta_\xi^\prime u  \|_2^2+4 \|\zeta_\xi u\|_2\|\zeta_\xi g\|_2\,.
\end{displaymath}
Observing that $|\zeta'_\xi|\leq \frac{C_0}{\xi}$,  we deduce
\begin{displaymath} 
\left(\frac{\xi}{2} - (\Im\lambda+\Re\lambda)\right)  \|\zeta_\xi u\|_2^2\leq
  \frac{4C_0}{\xi ^2} \| {\mathbf
    1}_{\tau\geq\xi/2}u  \|_2^2+4 \|\zeta_\xi u\|_2\|g\|_2\,.
\end{displaymath}
Hence, for $\xi\geq4(\Im\lambda+\Re\lambda)\geq \Re \vartheta_1 >0\,$, 
\begin{displaymath}
  \|{\mathbf 1}_{\tau\geq\xi}u\|_2 \leq \frac{C}{\xi}(\|g\|_2+ \xi^{-\frac 12} \|{\mathbf
    1}_{\tau\geq\xi/2}u\|_2) \leq   \frac{\hat C}{\xi}(\|g\|_2+  \|{\mathbf
    1}_{\tau\geq\xi/2}u\|_2) \,.
\end{displaymath}
Applying the above inequality $k$ times recursively yields,  for any $\xi$ satisfying \break
 $\xi\geq4^k(\Im\lambda+\Re\lambda)$\,,
\begin{equation}
\label{eq:141}
  \|{\mathbf 1}_{\tau\geq\xi}u\|_2 \leq \frac{C_k}{\xi}(\|g\|_2 +\xi^{-(k-1)}\|{\mathbf
    1}_{\tau\geq\xi/2^k}u\|_2)\,.
\end{equation}
Choosing $\xi=\ego^{-a}$, we obtain,  for $\ego^a\, 4^k(\Im\lambda+\Re\lambda)\leq 1$,
\begin{displaymath}
  \|{\mathbf 1}_{\tau\geq\ego^{- a}}u\|_2 \leq C_k\ego^a\, (\|g\|_2 +\ego^{a(k-1)}\|u\|_2)\,.
\end{displaymath}
Choosing $k\geq  \frac{3}{a}$ yields, 
\begin{displaymath}
    \|{\mathbf 1}_{\tau\geq\ego^{-a}}u\|_2 \leq C_a(\ego^a\|g\|_2\,.
    +\ego^3\|u\|_2)\,,
\end{displaymath}
With the aid of \eqref{eq:135}, we may now conclude the existence, of $\ego_0$ and $C$ such
that, for any $\ego\in(0,\ego_0]$,
\begin{equation}\label{eq:231aa}
    \|{\mathbf 1}_{\tau\geq\ego^{-a}}u\|_2 \leq  C \, \ego^a\|g\|_2\,.
\end{equation}
An additional conclusion that can be drawn from  \eqref{eq:140} is
$$
\| \partial_\tau (\zeta_\xi u)\|^2_2 + \ego \| \zeta_\xi  \partial_\sigma u \|^2_2 \leq   8 \|\zeta_\xi^\prime u  \|_2^2+ \frac{C}{\xi} \|\zeta_\xi g\|^2_2\leq \frac{\hat C}{ \xi} (\|g\|^2_2 + \xi^{-1} \|\mathbf{1}_{\tau \geq \xi/2} u\|^2_2) \,,
$$
which leads to
\begin{equation*}
  \|{\mathbf 1}_{\tau\geq\xi}\partial_\tau u \|_2 + \ego^{1/2}\|{\mathbf
    1}_{\tau\geq\xi}\partial_\sigma u \|_2\leq \frac{C}{\xi^{1/2}}(\|g\|_2+ \xi^{-\frac 12} \|{\mathbf
    1}_{\tau\geq\xi/2}u\|_2)\,.
\end{equation*}
Using \eqref{eq:141} with $\xi$ replaced by $\xi/2$, we obtain  for any
$k$  the existence of positive constants $C_k$ and $\xi_k$ such that
for  all $\xi \geq \xi_k$   
\begin{equation}
\label{eq:142}
  \|{\mathbf 1}_{\tau\geq\xi}\partial_\tau u \|_2 + \ego^{1/2}\|{\mathbf
    1}_{\tau\geq\xi}\partial_\sigma u \|_2\leq \frac{C_k}{\xi^{1/2}}(\|g\|_2+ \xi^{-\frac 12 -k} \| u\|_2)\,.
\end{equation}
Setting once again $\xi=\ego^{-a}$, and $k$ to be sufficiently large  completes the proof of
(\ref{eq:139}a).
 \paragraph{Proof of (\ref{eq:139}b)}~\\
 We first prove that  
\begin{equation}
  \label{eq:143}
\|{\mathbf 1}_{\tau\geq\ego^{-a}}\partial_\tau^2(\hat{\B}_\ego -\lambda)^{-1}\|  \leq C_a \,.
\end{equation}
Let, for some $C^\infty$ function $\zeta$ supported in $\mathbb R^+$, 
\begin{equation}
  \label{eq:144}
\kern -4px \Gg:=\kern -3px -\kern-1px  \langle(\hat{\B}_\ego -\lambda)u, \zeta^2u_{\tau\tau}\rangle \kern -2px =\kern -3px
-\kern -1px \langle (\LL_2^+
-\lambda)u,\zeta^2u_{\tau\tau}\rangle \kern -1px + \kern -1px \ego  \langle(1+\ego \theta \tilde \chi \tau ) u_{\sigma\sigma}
\kern -0.5px - \kern -0.5px i  \sigma^2 \tau    u ,  \zeta^2  u_{\tau\tau}  \rangle \,.
\end{equation}

We estimate each term separately, repeatedly applying the following integration by parts formula
$$
\langle   \tau  v ,  \hat {\zeta}^2  v_{\tau\tau}  \rangle = - \| \tau^\frac 12 \hat \zeta v_\tau\|^2 - \langle  v , \pa_\tau(\tau \hat \zeta^2) v_\tau \rangle\,,
$$
for various choices of $v$ and $\hat \zeta$. We thus have 
$$
  - \langle \LL_2^+(\ego ) u, \zeta^2u_{\tau\tau} \rangle =\| \zeta u_{\tau\tau} \|^2 - i \langle \tau (1 + \beta\ego  \tau \chi(\ego^b\tau) )u , \zeta^2 u_{\tau\tau}\rangle \,,
  $$
  $$
  \begin{array}{l}
   - i \langle \tau (1 + \beta\ego  \tau \chi(\ego^b\tau) )u , \zeta^2 u_{\tau\tau}\rangle \\
   \qquad = i \left(  \|\tau^\frac 12  (1 + \beta\ego  \tau \chi(\ego^b\tau)^\frac 12  \zeta u_\tau\|^2 + \langle u\,,\,    \pa_\tau(\tau (1 + \beta\ego  \tau \chi(\ego^b\tau)) \zeta^2) \, u_\tau \rangle  \right)\,,
   \end{array}
  $$
$$
  \langle \lambda u, \zeta^2u_{\tau\tau} \rangle =  \lambda \, \left(  - \| \pa_\tau (\zeta u) \|^2  +  \|\zeta' u \|^2 + \langle \pa_\tau (\zeta u) , \zeta' u\rangle  - \langle \zeta' u, \pa_\tau (\zeta u) \rangle  \right)\,,
  $$
  $$
 - i \ego \langle \sigma^2  \tau  u ,  \zeta^2  u_{\tau\tau}  \rangle = i\ego \left(    \| \tau^\frac 12 \sigma  \zeta u_\tau\|^2  + \langle \sigma^2  u , \pa_\tau(\tau  \zeta^2) u_\tau \rangle   \right)\,,
  $$
and
  $$
  \begin{array}{ll}
   \langle( - (1+\ego \theta \tilde \chi \tau ) u_{\sigma\sigma} ,  \zeta^2  u_{\tau\tau}  \rangle  &=   - \langle(  (1+\ego \theta \tilde \chi \tau ) u_{\sigma} ,  \zeta^2  (u_{\sigma})_{\tau\tau}  \rangle \\
   & = \|  (1+\ego \theta \tilde \chi \tau ) ^\frac 12 \zeta u_{\sigma\tau} \|^2  + \langle u_\sigma\,,\, \pa_\tau ((1+\ego \theta \tilde \chi \tau)\zeta^2) u_{\sigma\tau}\rangle\,.
   \end{array}$$
We now decompose $\Gg$ in the following manner
$$
\Gg= \Gg_1 + i \Gg_2 +\Gg_3\,,
$$  
where $\Gg_1$ and $\Gg_2$ are positive terms
defined by 
\begin{subequations}
\label{eq:145}
  \begin{align}
\Gg_1&=  \|\zeta  u_{\tau\tau}\|_2^2 + \Re \lambda \| \zeta' u\|^2  +  \ego  \|  (1+\ego \theta \tilde \chi \tau ) ^\frac 12 \zeta u_{\sigma\tau} \|^2\,, \\
\Gg_2&=  \|\tau^\frac 12  (1 + \beta\ego  \tau \chi(\ego^b\tau)^\frac 12  \zeta
u_\tau\|^2   + \Im \lambda  \| \zeta' u\|^2  + \ego    \| \tau^\frac 12 \sigma  \zeta
u_\tau\|^2\,,\\
\intertext{and $\Gg_3$ is given by}
\Gg_3&= i  \langle u\,,\,    \pa_\tau(\tau (1 + \beta\ego  \tau \chi(\ego^b\tau) \zeta^2))
\, u_\tau \rangle \notag \\
&\quad +  \lambda \, \left(  - \| \pa_\tau (\zeta u) \|^2  + <\langle \pa_\tau (\zeta u) ,
  \zeta' u\rangle  - <\langle \zeta' u, \pa_\tau (\zeta u) \rangle  \right) \notag  \\
& \quad + i\ego \langle
\sigma^2  u , \pa_\tau(\tau  \zeta^2) \, u_\tau \rangle \notag \\
&\quad + \ego \langle u_\sigma\,,\, \pa_\tau ((1+\ego \theta \tilde \chi \tau)\zeta^2) \,
u_{\sigma\tau}\rangle\,. 
\end{align}
\end{subequations}

For the first term on the right-hand side of (\ref{eq:145}c), we have
$$
 |  \langle u\,,\,    \pa_\tau(\tau (1 + \beta\ego  \tau \chi(\ego^b\tau) ) \zeta^2) \, u_\tau \rangle| \leq  C \left( \| \zeta u\|^2_2 + \| \zeta u_\tau\|^2_2 + \| \tau \zeta'u_\tau\|^2_2 \right)\,.
 $$
 It is readily verified that the second line is bounded by $ C(\| \zeta'
 u\|^2_2 +  \|\partial_\tau(\zeta u)\|_2^2 )$. Let $\epsilon>0$. 
For the third line of (\ref{eq:145}c) we have
\begin{displaymath}
  | \langle \sigma  \tau^{\frac 12} \zeta  u_\tau,\tau^{-\frac 12} \zeta\ \sigma u\rangle +   2  \langle \zeta  \sigma \tau^\frac 12 u_\tau, \zeta' \tau^\frac 12  \ \sigma u\rangle | \
\leq \epsilon \|\sigma\tau^{1/2}\zeta  u_\tau\|^2_2 + \frac{C} {\epsilon} \left(   \| \zeta' \tau^\frac 12  \ \sigma u\|^2_2 + \| \tau^{-\frac 12} \zeta \sigma u\|^2_2\right).
\end{displaymath}
Finally for  the forth term, we have for some $C_\epsilon>0$
\begin{displaymath}
  \ego |\langle u_\sigma\,,\, \pa_\tau ((1+\ego \theta \tilde \chi \tau)\zeta^2) u_{\sigma\tau}\rangle|
\leq \ego \epsilon   \|  (1+\ego \theta \tilde \chi \tau ) ^\frac 12 \zeta u_{\sigma\tau}
\|^2_2  + \ego C_\epsilon  \| \partial_\tau(\zeta [1+\ego\theta\tau\tilde{\chi}]^{1/2})\,u_\sigma
\|^2_2 \,. 
\end{displaymath}
Combining the above yields that for every $\epsilon>0$ there exists $C_\epsilon>0$
such that
\begin{equation}
  \label{eq:146}
|\Gg_3| \leq\epsilon \, (\Gg_1+\Gg_2) + C_\epsilon\, \Gg_4
\end{equation}
where
\begin{equation}
\label{eq:233}
  \Gg_4=\|{\mathbf 1}_{\tau\geq\xi/2}u\|_2^2 +
\ego\|{\mathbf 1}_{\tau\geq\xi/2}u_\sigma\|_2^2+ \|\zeta u_\tau\|_2^2+ \ego\|{\mathbf
  1}_{\tau\geq\xi/2}\tau^{1/2}\sigma u\|_2^2 \,.
\end{equation}
To obtain (\ref{eq:146}) we have used the pointwise inequality 
\begin{displaymath}
  |\zeta_\xi|+|\zeta^\prime_\xi|+\ego|\zeta(\tau\tilde{\chi})^\prime|\leq C {\mathbf
  1}_{\tau\geq\xi/2}\,.
\end{displaymath}

We now observe that, by \eqref{eq:144}, 
$$
 \|\zeta  u_{\tau\tau}\|^2_2 \leq \Gg_1 +\Gg_2\leq \sqrt{2}|\Gg_1+i \Gg_2| \leq
 \sqrt{2}(|\Gg| + | \Gg_3|) \leq \epsilon (\Gg_1+\Gg_2 +\|\zeta  u_{\tau\tau}\|^2_2 )
 + C_\epsilon (\Gg_4 +\|g\|^2_2)\,, 
$$
which implies
\begin{equation}
\label{eq:234}
\|\zeta  u_{\tau\tau}\|^2_2 \leq C (\|g\|^2_2 + \Gg_4)\,.  
\end{equation}
A proper bound of $\Gg_4$ would thus complete the proof of
\eqref{eq:143}. To this end, we now show that there exists $C$ and $\ego_0$ such
that we have, with $\zeta=\zeta_\xi$, $\ego \in (0,\ego_0]$ and $\xi =
\ego^{-\mathfrak a}$,
\begin{equation}\label{estg4}
 \Gg_4 \leq C \,  \|g\|^2_2\,.
\end{equation}

The first term appearing on the right-hand-side of \eqref{eq:233},
 may be estimated by using \eqref{eq:231aa} (which remains valid if
 ${\mathbf 1}_{\tau\geq\ego^{-a}}$ is replaced by  ${\mathbf 1}_{\tau\geq\ego^{-a}/2}$).
To estimate the second term and the third term on
the right-hand-side of \eqref{eq:233}, we
use \eqref{eq:142} (with $\xi$ replaced by $\xi/2$). 
Finally, to estimate the last term on the right-hand-side of \eqref{eq:233},
   we may use \eqref{eq:140} to obtain
   $$
\ego  \| \tau^\frac 12 \sigma {\bf 1}_{\tau \geq  \xi} u\|^2   \leq  \ego  \|
\tau^\frac 12 \sigma \zeta u\|^2 \leq C\,  (\|g\|^2 + \|\zeta u\|^2 + \| \zeta'u\|^2)\leq
\hat C\,  \|g\|^2\,. 
   $$
Consequently, by \eqref{eq:233} we have \eqref{estg4} which when
substituted into \eqref{eq:234} yields \eqref{eq:143}.

Note, for future use, that the proof provides us,
for sufficiently small $\ego_0$,
\begin{equation}\label{eq:add}
\|\tau^\frac 12 \zeta u \|^2 \leq 2 \Gg_2 \leq C\,  \| g\|^2\,.
\end{equation}

\paragraph{Estimation of $\|{\mathbf 1}_{\tau\geq\ego^{-a}}\partial_\sigma^2(\hat{\B}_\ego -\lambda)^{-1}\| $}~\\
 To complete the proof of (\ref{eq:139}b), it remains necessary to show that
\begin{equation}\label{eq:125ba}   \ego\|{\mathbf 1}_{\tau\geq\ego^{-a}}\partial_\sigma^2(\hat{\B}_\ego -\lambda)^{-1}\| \leq
   C_a \,. 
  \end{equation}
  To this end we write 
  $$
\hat \Gg:=- \langle(\hat{\B}_\ego -\lambda)u, \zeta^2u_{\sigma\sigma }\rangle =   \langle (\LL_2^+(\ego ) -\lambda)u _\sigma ,\zeta^2u_{\sigma}\rangle - \ego  \langle( - (1+\ego \theta \tilde \chi \tau ) u_{\sigma\sigma} + i  \sigma^2 \tau    u ),  \zeta^2  u_{\sigma\sigma}  \rangle \,.
$$
We have
$$
\begin{array}{ll}
\hat\Gg &= \left(\ego \|  (1+\ego \theta \tilde \chi \tau )^\frac 12  \zeta u_{\sigma\sigma}\|^2_2
 + \| \zeta u_{\tau\sigma} \|^2_2\right) \\ & \quad + i \left( \ego \| \sigma \tau^\frac 12   \zeta u_\sigma\|^2_2 + \| \tau^\frac 12  (1 + \beta\ego  \tau \chi(\ego^b\tau) )^\frac 12 \zeta u_\sigma \|^2_2 \right) \\
&\quad  -\lambda \|\zeta u_\sigma\|^2_2  +2  i\ego  \langle \tau^\frac 12 \sigma    \zeta  u ,  \zeta   \tau^\frac 12 u_{\sigma}  \rangle + 2 \langle \zeta u_{\tau\sigma}\,,\,\zeta'  u_\sigma\rangle \\
&:= \hat \Gg_1 + i \,\hat \Gg_2 + \hat \Gg_3\,,
\end{array}
$$
from which we obtain as in the proof of \eqref{eq:143} (recall that
$\|\ego\theta\tau\tilde{\chi}\|_\infty\leq C\ego^{1-\tilde b}$ and $|\beta\ego \tau
\chi(\ego^b\tau)| \leq C \ego^{1-b}$)
\begin{equation}
\label{eq:147}
  \ego\|\zeta u_{\sigma\sigma}\|^2_2 \leq C \left( \ego^{-1} \|g\|^2_2 +   \|\zeta^\prime u_\sigma\|^2_2 +
   \|\zeta u_\sigma\|^2_2 +\ego^2 \|  \tau^{1/2} \zeta u \|^2_2\right)\,.
\end{equation}
To bound the last  term on the right-hand-side we use 
\eqref{eq:add}, whereas for first two terms we obtain from 
\eqref{eq:142} (with $\xi$ replaced by $\xi/2$)
$$
  \|\zeta^\prime u_\sigma\|_2^2 +  \|\zeta u_\sigma\|_2^2     \leq C \| {\bf 1}_{\tau \geq \xi/2} u_\sigma\|^2_2\leq \hat C \ego^{-1}\|g\|^2_2\,.
$$
 Hence,
\begin{displaymath}
   \ego\|{\mathbf 1}_{\tau\geq\ego^{-a}}u_{\sigma\sigma}\|_2 \leq C_a\|g\|_2\, ,
\end{displaymath}
which completes the proof of
(\ref{eq:125ba}).
 \end{proof}

\subsubsection{Effect 2}
We now address an additional  modification of  $ \B_\ego$, 
resulting from the third term on the right-hand-side of
\eqref{eq:36}. 
 \begin{proposition}
\label{prop:effect-2}
 Let, for $\omega \in \R$, 
   \begin{equation}
 \label{eq:148}
 \tilde{\B}_\ego =\hat{\B}_\ego - 2\omega\ego\partial_\tau  \,,
 \end{equation}
 be defined on $e^{-\ego\omega\tau}D( \B_\ego)$. Let further, for some $0<a < a' <1$,
 $I_\ego=(\ego^{-a},\ego^{-a'})$.

 Then, there exist positive $C$ and
 $\ego_0$ such that, for every $\ego\in (0,\ego_0]$, the circle $ \partial B(\Lambda_\gamma^1, \ego r(\ego))$ is included in $\rho(\tilde{\B}_\ego)$,  and,  for $\lambda \in
 \partial B(\Lambda_\gamma^1, \ego r(\ego))$\,, 
 \begin{subequations}
 \label{eq:149}
   \begin{multline}
  \|{\mathbf 1}_{\tau\leq\ego^{-a'}}(\tilde{\B}_\ego -\lambda)^{-1}{\mathbf
    1}_{\tau\leq\ego^{-a'}}\| + \|{\mathbf
    1}_{\tau\leq\ego^{-a'}}\partial_\tau(\tilde{\B}_\ego -\lambda)^{-1}{\mathbf
    1}_{\tau\leq\ego^{-a'}}\|  \\ + \ego^{1/2}\|{\mathbf
    1}_{\tau\leq\ego^{-a'}}\partial_\sigma (\tilde{\B}_\ego -\lambda)^{-1}{\mathbf
    1}_{\tau\leq\ego^{-a'}}\| + \|{\mathbf
    1}_{\tau\leq\ego^{-a'}}\partial_\tau^2(\tilde{\B}_\ego -\lambda)^{-1}{\mathbf
    1}_{\tau\leq\ego^{-a'}}\|  \\  + \ego\|{\mathbf
    1}_{\tau\leq\ego^{-a'}}\partial_\sigma^2(\tilde{\B}_\ego -\lambda)^{-1}{\mathbf
    1}_{\tau\leq\ego^{-a'}}\|\leq \frac{C}{(r(\ego) +1-\gamma)\ego} \,, 
 \end{multline} 
  \begin{multline}
 \ego^{-a}\|{\mathbf 1}_{\tau\in I_\ego}(\tilde{\B}_\ego -\lambda)^{-1}{\mathbf
    1}_{\tau\leq\ego^{-a'}}\| 
 + \ego^{-a/2}\|{\mathbf
     1}_{\tau\in I_\ego}\partial_\tau(\tilde{\B}_\ego -\lambda)^{-1}{\mathbf
     1}_{\tau\leq\ego^{-a'}}\| 
 \\ +  \ego^{1/2}\|{\mathbf
     1}_{\tau\in I_\ego}\partial_\sigma (\tilde{\B}_\ego -\lambda)^{-1}{\mathbf
     1}_{\tau\leq\ego^{-a'}}\| 
+ \|{\mathbf
    1}_{\tau\in I_\ego}\partial_\tau^2(\tilde{\B}_\ego -\lambda)^{-1}{\mathbf
     1}_{\tau\leq\ego^{-a'}}\| 
 \\ + \ego\|{\mathbf
     1}_{\tau\in I_\ego}\partial_\sigma^2(\tilde{\B}_\ego -\lambda)^{-1}{\mathbf
     1}_{\tau\leq\ego^{-a'}}\| \leq C \,,
 \end{multline}
 
and for every $1/6<{\mathfrak a}<1/4$
 \begin{multline}
    \| {\mathbf 1}_{\tau\leq\ego^{-a'}}{\mathbf 1}_{|\sigma|\geq
      2\ego^{-{\mathfrak a}}}(\tilde{\B}_\ego-\lambda)^{-1}{\mathbf
      1}_{\tau\leq\ego^{-a'}}\| \\ +  \ego^{1/2} \|{\mathbf 1}_{\tau\leq\ego^{-a'}}{\mathbf 1}_{|\sigma|
    \geq 2\ego^{-{\mathfrak a}}}\partial_\sigma(\tilde{\B}_\ego-\lambda)^{-1}{\mathbf 1}_{\tau\leq\ego^{-a'}}\|\leq
  \frac{C } {\ego^{1-2{\mathfrak a}}    }\,. 
 \end{multline}
 \end{subequations}
 \end{proposition}

 \begin{proof}
    It can be easily verified that
   \begin{displaymath}
     \tilde{\B}_\ego =  e^{-\ego\omega\tau} (\hat{\B}_\ego -\ego^2\omega^2 )e^{\ego\omega\tau}\,.
   \end{displaymath}
  For the first statement in (\ref{eq:149}a), we have
 $$
 \begin{array}{ll}
    \|{\mathbf 1}_{\tau\leq\ego^{-a'}}(\tilde{\B}_\ego -\lambda)^{-1}{\mathbf
    1}_{\tau\leq\ego^{-a'}}\| & =\| e^{-\ego\omega\tau}  {\mathbf 1}_{\tau\leq\ego^{-a'}} (\hat{\B}_\ego -\lambda -\ego^2\omega^2 )^{-1}{\mathbf
    1}_{\tau\leq\ego^{-a'}} e^{\ego\omega\tau}\|\\
    &\leq e^{2|\omega| \ego^{1-a'}}\|  (\hat{\B}_\ego -\lambda -\ego^2\omega^2 )^{-1}\|\,,
    \end{array}
    $$
    and we can use \eqref{eq:135} (note that  \eqref{eq:135} is valid
     in the ring 
\begin{displaymath}
  A_\ego=B (\Lambda^1_\gamma,2 \ego r(\ego))\setminus B (\Lambda^1_\gamma,\frac 12 \ego
     r(\ego))
\end{displaymath}
and that $\partial B (\Lambda^1_\gamma + \ego^2\omega^2 ,\ego r(\ego)) \subset
    A_\ego$     for $\ego_0>0$ small enough) to obtain
    \begin{displaymath}
      \|{\mathbf 1}_{\tau\leq\ego^{-a'}}(\tilde{\B}_\ego -\lambda)^{-1}{\mathbf
    1}_{\tau\leq\ego^{-a'}}\| \leq \frac{C}{(r(\ego) +1-\gamma)\ego} \,.
    \end{displaymath}
   The rest of the inequalities embedded in (\ref{eq:149}a) are
    similarly obtained by using
    (\ref{eq:135}b) and (\ref{eq:135}c).

To bound the first term on the right-hand-side of (\ref{eq:149}b), we use
    \eqref{eq:139} to obtain
\begin{displaymath}
   \|{\mathbf 1}_{\tau\in I_\ego}(\tilde{\B}_\ego
   -\lambda)^{-1}{\mathbf 1}_{\tau\leq\ego^{-(1-a)}}\| \leq e^{2|\omega| \ego^{1-a'}}  \|{\mathbf 1}_{\ego^{-a}\leq\tau}(\hat{\B}_\ego
   -\lambda-\ego^2\omega^2)^{-1}\| \leq C_a\,.
\end{displaymath}
As
\begin{displaymath}
  \partial_\tau(\tilde{\B}_\ego
   -\lambda)^{-1}= \partial_\tau e^{\ego\omega\tau} (\hat{\B}_\ego
   -\lambda-\ego^2\omega^2)^{-1}e^{-\ego\omega\tau}=e^{\ego\omega\tau} [\partial_\tau+\ego\omega](\hat{\B}_\ego
   -\lambda-\ego^2\omega^2)^{-1}e^{\ego\omega\tau}\,,
\end{displaymath}
 we may conclude, once again with the aid of \eqref{eq:139}, that 
\begin{displaymath}
   \|{\mathbf 1}_{\tau\in I_\ego}\partial_\tau(\tilde{\B}_\ego
   -\lambda)^{-1}{\mathbf 1}_{\tau\leq\ego^{-(1-a)}}\| \leq C_a\,  \ego^\frac a 2\,.
\end{displaymath} 
The rest of the inequalities in \eqref{eq:149} can be proved in a
similar manner. 
\end{proof}
\subsection{A linear potential estimate}
We conclude this section by the following  estimate,
which is somewhat similar to \cite[Lemma 7.5]{AGH}. Let
\begin{displaymath}
  \check \B_\delta= -\partial^2_\tau +i(1+\delta)\tau - \partial^2_\sigma \,,
\end{displaymath}
be defined on
\begin{displaymath}
  D(\check \B_\delta)=\{u\in H^2(\R^2_+)\cap H^1_0(\R^2_+) \,|\,\tau u\in L^2(\R^2_+)\,\}.
\end{displaymath}
\begin{proposition}
  For any $a>0$, there exist $\delta_0>0$ and $C>0$ such that for all $\delta\in
  (0,\delta_0]$, $p<\frac{2}{3\, \Re \vartheta_1}$, and $\Re\lambda\leq \Re\vartheta_1+p\delta$ we have
  \begin{subequations}
    \label{eq:150}
    \begin{equation}
\|{\mathbf 1}_{\tau\geq\delta^{-a}}(\check \B_\delta-\lambda)^{-1}\|+\|{\mathbf
  1}_{\tau\geq\delta^{-a}}\partial_\tau(\check \B_\delta-\lambda)^{-1}\|\leq C\,, 
    \end{equation}
and
\begin{equation}
  \|\partial_\sigma(\check \B_\delta-\lambda)^{-1}\| \leq \frac{C}{\delta^{1/2}} \,.
\end{equation}
  \end{subequations}
\end{proposition}
\begin{proof}
  Since the proof is similar to the proof of \cite[Lemma 7.5]{AGH} we
  bring only its outlines. We first apply the transformation $(t,
  s)=(1+\delta)^{1/3}(\tau,\sigma)$ and argue for $\lambda^\prime=(1+\delta)^{-2/3}\lambda\,$. The
  operator then assumes the form \break $\check \B_0= - \pa_t^2 + i t
  -\pa_s^2$, and $\lambda^\prime$ satisfies $\Re \lambda^\prime \leq (\Re \vartheta_1 + p \delta)
  (1+\delta)^{-\frac 23}$.

  We next observe that
  \begin{equation}
\label{eq:231}
    \|(\check \B_0-\lambda')^{-1}(I-\Pi_1)\|+ \| \nabla_{t,s} (\check
    \B_0-\lambda')^{-1}(I-\Pi_1)\|\leq C \,.
  \end{equation}
Consequently, the proof of (\ref{eq:150}a) follows immediately from
the decay 
of the Airy function $v_1$ and its derivative. To prove
(\ref{eq:150}b) we begin by writing
\begin{displaymath}
  (\check \B_0-\lambda^\prime) \Pi_1= (-\partial_s^2+ \vartheta_1-\lambda^\prime)\Pi_1\,.
\end{displaymath}
Integration by parts then yields, using the fact that $\Re \lambda^\prime < \Re
\vartheta_1$ for sufficiently small $\delta_0$,
\begin{displaymath}
  \|\partial_s(\check \B_0-\lambda')^{-1}\Pi_1\|_2^2 \leq \|(\check
  \B_0-\lambda')^{-1}\Pi_1\|\, \|\Pi_1\|  \,,
\end{displaymath}
 which together with (\ref{eq:231}) yields (\ref{eq:150}b) . 
\end{proof}

   \section{Simplified operators: V2 potentials}
   \label{sec:simplified-V2}
   In this section we estimate the resolvent norm of the operator whose
   eigenvalues were formally found in \eqref{eq:51}.  For convenience,
   we use an even extension to $(\sigma,\tau)\in\R\times\R_+$, instead of
   considering it on $\R_+\times\R_+$ with a Neumann boundary condition for
   $\sigma=0$, as in Section \ref{sec:quazimode-V2}.
  
   \subsection{Definition and preliminary estimates}

We begin by defining for $\varepsilon >0$
\begin{equation}
\label{eq:151}
  \Ug_\varepsilon   = -[1+\varepsilon |\sigma|]\partial^2_\tau -\varepsilon \partial^2_\sigma + i\tau\,.
\end{equation}
To define $\Ug_\varepsilon$ and characterize its domain, we look at the
typical case $\varepsilon=1$.  We start from the bilinear form given by
\begin{displaymath}
  a(u,v)= \langle u_\tau,(1+|\sigma|)v_\tau\rangle+ \langle u_\sigma,v_\sigma\rangle +i\langle u,\tau v\rangle\,, 
\end{displaymath}
defined on $\Vg\times\Vg$ where
\begin{displaymath}
  \Vg = \{ u\in H^1_0(\R^2_+)\,|\,| \tau^{1/2}u\in L^2(\R^2_+) \;;\;
  |\sigma|^{1/2}u_\tau\in L^2(\R^2_+) \} \,,
\end{displaymath}
is equipped with the norm
\begin{displaymath}
  \|u\|_\Vg^2= \|u_\sigma\|_2^2 + \|(1+|\sigma|)^{1/2}u_\tau\|_2^2 +
  \|(1+\tau^{1/2})u\|_2^2\,.
\end{displaymath}
It can be easily verified that
\begin{displaymath}
  |a(u,v)| \leq \|u\|_\Vg\|v\|_\Vg\,,
\end{displaymath}
and that there exists $c>0$ such that
\begin{displaymath}
   |a(u,u)|\geq c \|u\|_\Vg^2 \,,\, \forall u\in \mathcal V\,.
\end{displaymath}
It follows from the Lax-Milgram Theorem (cf. \cite{Helbook,AH}) that we can define 
$\mathcal U_1$ as a closed semibounded operator on $L^2(\R^2_+)$, whose
domain is given by
$$
\begin{array}{ll}
D(\mathcal U_1)&  =\{ u\in \Vg, s.t \,  \Vg \ni v\mapsto a(u,v) \mbox{ can be
  extended} \\ 
& \qquad \qquad \mbox{ as a continuous antilinear map on } L^2(\R^2_+)\}\,.
\end{array}
$$
Since we consider a Dirichlet problem, we then have 
\begin{equation}
  \label{eq:152}
D(\mathcal U_1) =\{u\in \Vg, \, s.t.\; \mathcal U_1 u \in L^2(\R^2_+)\}\,.
\end{equation}
Moreover  (see the proof of Lemma 3.3 in \cite{AH})  the
subspace  $\tilde{ \mathcal V}$ of the functions in  
  $ C^\infty(\overline{\R^2_+})\cap H^1_0(\R^2_+)$ compactly  supported in
  $\overline{\R^2_+}$ is dense in $\mathcal V$ and in  $D(\mathcal
  U_1)$ for the graph norm. 
\begin{lemma}
$\mathcal U_1$ has compact resolvent.
\end{lemma}
\begin{proof}
The operator being semi-bounded it is enough to prove the compact
injection of $\mathcal V$ into $L^2(\mathbb R^2_+)$. We observe that,
for all $u\in \tilde{ \mathcal V}$, 
  \begin{displaymath}
\begin{array}{ll}
\|u\|^2_{\mathcal V} & \geq     \int_{\R_+^2} (|\sigma|\, |u_\tau(\sigma,\tau) |^2 +
\tau|u (\sigma,\tau)|^2)\, d\sigma  d\tau\\
&\geq  \frac 12   \int_{\R_+} \left(\int_\mathbb R (|\sigma|\, |u_\tau(\sigma,\tau) |^2 +
\tau|u (\sigma,\tau)|^2)\, d\sigma \right) d\tau + \frac 12  \int_{\R_+^2}
\tau|u (\sigma,\tau)|^2\, d\sigma  d\tau  \\   &   \geq  \frac{|\nu_1|}{ 2} \, \int_{\R_+^2} | \sigma| ^{1/3} |u (\sigma,\tau)|^2 d \sigma \,d\tau\,
+\frac 12\,  \int_{\R_+^2} \tau |u(\sigma,\tau)|^2\,d\sigma d\tau    \,,
\end{array}
  \end{displaymath}
  where we recall that $|\nu_1|$ is the first eigenvalue of  the
  Dirichlet realization of the Airy  operator $D_\tau^2 +\tau$ in $\mathbb
  R^+$.\\ 
  By density the inequality is true for $u\in \mathcal V$ and proves  the
  continuous injection of $\Vg$ into an $L^2$ weighted space whose weight
  $(|\tau| + |\sigma|^\frac 13)$ tends to $+\infty$ as $(|\sigma| +|\tau|)$ tends to
  $+\infty$. This injection combined with the fact that $\mathcal V\subset
  H^{1,loc}(\overline{\mathbb R^2_+})$ completes the proof of the
  lemma.
\end{proof}

\begin{lemma}
\begin{equation}
\label{eq:153}
D(\mathcal U_1)= \{u\in \mathcal V\,,\, (1+|\sigma|)u_{\tau\tau} \in L^2\mbox{ and } u_{\sigma\sigma} \in L^2\}\,.
\end{equation}
\end{lemma}
\begin{proof}
It is enough to establish an inequality for $u\in \tilde{\mathcal
  V}$. We begin with the identity
\begin{equation}\label{starteq}
  \|\Ug_1u\|_2^2 = \| (1+|\sigma|)u_{\tau\tau}+u_{\sigma\sigma}\|_2^2 + \|\tau u\|_2^2
  -   2\Im\langle u_\tau,(1+|\sigma|)u\rangle\,.
\end{equation}
Then, we deduce from
\begin{displaymath}
\begin{array}{ll}
     \langle (1+|\sigma|)^{1/3}u,\Ug_1u\rangle & 
      =  \|
     (1+|\sigma|)^{\frac{2}{3}}u_\tau\|_2^2+ \|(1+|\sigma|^{1/6})u_\sigma\|_2^2 \\ &  \quad + i 
     \|(1+|\sigma|)^{1/6}\tau^{1/2}u\|_2^2  \\
     & \quad 
     +  \frac{1}{3} \langle\sign \sigma
     (1+|\sigma|)^{-2/3}u,u_\sigma\rangle \,,
     \end{array}
\end{displaymath}
that for any $\epsilon>0$ there exists $C_\epsilon>0$ such that 
$$
\begin{array}{l}
  \|
     (1+|\sigma|)^{\frac{2}{3}}u_\tau\|_2^2+ \|(1+|\sigma|^{1/6})u_\sigma\|_2^2 +    \|(1+|\sigma|)^{1/6}\tau^{1/2}u\|_2^2\\
     \quad   \leq | \frac{1}{3} \langle\sign \sigma
     (1+|\sigma|)^{-2/3}u,u_\sigma\rangle \,| + C_\epsilon \| \Ug_1u\|^2 + \epsilon \| (1+|\sigma|)^{1/3}u\|^2\,,
     \end{array}
$$
from which we conclude 
$$
\begin{array}{l}
  \|
     (1+|\sigma|)^{\frac{2}{3}}u_\tau\|_2^2+ \|(1+|\sigma|^{1/6})u_\sigma\|_2^2 +    \|(1+|\sigma|)^{1/6}\tau^{1/2}u\|_2^2\\
     \quad   \leq 
       C_\epsilon (\| \Ug_1u\|^2+ \|u\|^2)  + \epsilon \| (1+|\sigma|)^{1/3}u\|^2\,.
     \end{array}
$$
On the other hand, we have  (see the proof of the previous lemma)
\begin{displaymath}
  \|(1+|\sigma|)^{\frac{2}{3}}u_\tau\|_2^2+ \|(1+|\sigma|)^{1/6}\tau^{1/2}u\|_2^2
  \geq |\nu_1| \,\|(1+|\sigma|)^{1/3}u\|_2^2\,.
\end{displaymath}
For sufficiently small $\epsilon >0$ the above two inequalities imply
\begin{equation}\label{eq:154}
  \|(1+|\sigma|)^{1/3}u\|_2^2  +  \| (1+|\sigma|)^{\frac{2}{3}}u_\tau\|_2^2+ \|(1+|\sigma|^{1/6})u_\sigma\|_2^2 
     \leq 
       C  (\| \Ug_1u\|^2+ \|u\|_2^2) \,.
  \end{equation}
Returning to \eqref{starteq}, we estimate the third term of the right
hand side in the following manner
\begin{equation*}
  2|\langle u_\tau,(1+|\sigma|)u\rangle|\leq\|(1+|\sigma|)^{\frac{2}{3}}u_\tau\|_2^2+\|(1+|\sigma|)^{1/3}u\|_2^2 \leq
  C (\|\Ug_1u\|_2^2+ \|u\|^2_2)\,,
\end{equation*}
to obtain
\begin{equation}
\label{eq:155}
  \| (1+|\sigma|)u_{\tau\tau}+u_{\sigma\sigma}\|_2^2 + \|\tau u\|_2^2\leq C
  (\|\Ug_1u\|_2^2 +\|u\|_2^2)\,.
\end{equation}
Finally, since
\begin{displaymath}
  \Re\langle(1+|\sigma|)u_{\tau\tau},u_{\sigma\sigma}\rangle= \|[1+|\sigma|]^{1/2}u_{\tau\sigma}\|_2^2+\Re\langle u_{\tau\sigma},\sign \sigma\,
  u_\tau\rangle\,,
\end{displaymath}
(where we have used the fact that $u_\sigma$ and $u_{\sigma\sigma}$ vanish on $\partial\R^2_+$ when $u\in \tilde{\Vg}$)
we may conclude that
\begin{displaymath}
  \| (1+|\sigma|)u_{\tau\tau}+u_{\sigma\sigma}\|_2^2 \geq  \|
  (1+|\sigma|)u_{\tau\tau}\|_2^2+\|u_{\sigma\sigma}\|_2^2 - \|u_\tau\|_2^2 \,.
\end{displaymath}
By \eqref{eq:154} and \eqref{eq:155} we then obtain
\begin{equation}
\label{eq:156}
  \| (1+|\sigma|)u_{\tau\tau}\|_2^2+\|u_{\sigma\sigma}\|_2^2 \leq  C
  (\|\Ug_1u\|_2^2 +\|u\|_2^2)\,,\,\forall u\in \tilde{\Vg}\,.
\end{equation}
By density \eqref{eq:156} is extended to $u\in D( \Ug_1)$ establishing,
thereby, \eqref{eq:153}. 
\end{proof}

\subsection{Large $|\sigma|$ simplification}
In the similar fashion to \eqref{eq:93} we define  in $Q:=\mathbb
R_+\times \mathbb R_+$ the operator  
\begin{equation}  \label{eq:157}
\Tg_\varepsilon =  -[1+\varepsilon(\sigma+\varepsilon^{-{\mathfrak a}})]\partial^2_\tau - \varepsilon\partial^2_\sigma + i\tau\,,
\end{equation}
associated with  the bilinear form given by
\begin{displaymath}
  a^+(u,v)= \langle u_\tau,(1+(\sigma+\varepsilon^{-{\mathfrak a}}) )v_\tau\rangle+ \varepsilon  \langle u_\sigma,v_\sigma\rangle +i\langle u,\tau v\rangle\,, 
\end{displaymath}
defined on $\Vg^+\times\Vg^+$ where
\begin{displaymath}
  \Vg^+ = \{ u\in H^1_0(Q)\,|\, \tau^{1/2}u\in L^2(Q) \;;\;
  \sigma^{1/2}u_\tau\in L^2(Q) \} \,,
\end{displaymath}
is equipped with its natural Hilbertian norm. In the same manner we have
established \eqref{eq:154}, we can prove that  the domain of $\Tg_\varepsilon$
is 
\begin{equation}
  D(\Tg_\varepsilon) = \{u\in H^2(Q)\cap H^1_0(Q)\,| \,\tau u\in L^2(Q) \,;\,\sigma u_{\tau\tau}\in L^2(Q)\}\,.
\end{equation}
  We can also show as for $\mathcal U_\varepsilon$
that
$\Tg_\varepsilon$ has compact resolvent.\\
Let 
$$
\Lambda_\gamma^2(\varepsilon)=\lambda_0 + \gamma \varepsilon \check \lambda_1
$$ 
be given by \eqref{eq:56},
\eqref{deflambda0} and \eqref{deflambda1}, and let $r(\varepsilon)$ satisfy
\eqref{condsurr}.  We can now state and prove
\begin{proposition}
\label{lem:large-sigma-2}
  Let $1/4<{\mathfrak a}<(1-q)/2$. Then, there exist positive
$C$ and $\varepsilon_0$ such that for all $\varepsilon\in (0,\varepsilon_0]$,  $\partial B(\Lambda_\gamma^2,r(\varepsilon) \varepsilon)\subset\rho(\Tg_\varepsilon)$ and such that for  $\lambda\in\partial B(\Lambda_\gamma^2,r(\varepsilon) \varepsilon)$
\begin{subequations}
\label{eq:158}
  \begin{equation}
    \|(\Tg_\varepsilon -\lambda)^{-1}\| \leq \frac{C}{\varepsilon^{1-{\mathfrak a}}} \,,
  \end{equation}
    \begin{equation} 
    \| \partial_\tau\, (\Tg_\varepsilon -\lambda)^{-1}\| \leq \frac{C}{\varepsilon^{1-{\mathfrak a}}} \,,
  \end{equation}
   \begin{equation} 
    \| \partial_\tau^2\, (\Tg_\varepsilon -\lambda)^{-1}\| \leq \frac{C}{\varepsilon^{1-{\mathfrak a}}} \,,
  \end{equation}
  and
\begin{equation}
  \|\partial_\sigma(\Tg_\varepsilon-\lambda)^{-1}\| \leq \frac{C}{\varepsilon^{3/2-{\mathfrak a}}} \,.
\end{equation}
\end{subequations}
\end{proposition}
\begin{proof}
  The proof is similar to the proof of Proposition
  \ref{lem:large-sigma}, and we therefore bring only its outlines. We
  begin by defining the partition of unity \eqref{eq:95}, $S_k$ as in
  \eqref{defSk} and $\Tg_k$ as in \eqref{eq:209a} with $\ego$ replaced
  by $\varepsilon$, where we recall that $\frac 16 < \mathfrak b < \mathfrak
  a$. Then, we set
  \begin{equation}
\Rg^{app}_{\Tg}= \sum_{k=0}^\infty \phi_k^\varepsilon(\Tg_k-\lambda)^{-1}\phi_k^\varepsilon \,,
\end{equation}
yielding
\begin{subequations}
\label{eq:159}
\begin{equation}
  (\Tg_\varepsilon-\lambda) \Rg^{app}_{\Tg} = I + \Eg_\Tg \,,
\end{equation}
where
\begin{equation}
   \Eg_\Tg = -\sum_{k=0}^\infty \varepsilon[\partial^2_\sigma,\phi_k^\varepsilon](\Tg_k-\lambda)^{-1}\phi_k^\varepsilon \,.
\end{equation}
\end{subequations}
We now prove that $\Eg_\Tg \to0$ as $\varepsilon\to0$. To this end we set
\begin{displaymath}
  (\Tg_k-\lambda)w=g
\end{displaymath}
for $g\in L^2(S_k)$ and $w\in D(\Tg_k)$. As in the derivation of
\eqref{eq:100} and \eqref{eq:101} we write
\begin{displaymath}
  (-[1+\varepsilon(k\varepsilon^{-{\mathfrak b}}+\varepsilon^{-{\mathfrak a}})]\partial^2_\tau -
  \varepsilon\partial^2_\sigma + i\tau-\lambda)w=g+\varepsilon(\sigma-k\varepsilon^{-{\mathfrak b}}+\varepsilon^{-{\mathfrak a}})w_{\tau\tau}\,.
\end{displaymath}
This allows us to conclude by extending the proof of \eqref{eq:88},
using a dilation in the $\tau$ variable with $\ego$ replaced by $\varepsilon$ and
$\beta=0$, to get, under our assumption that $\lambda$ is $\mathcal
O(\varepsilon)$-close to $\lambda_0$,
\begin{equation}
\label{eq:160}
   \|w\|_2\leq \frac{C}{\varepsilon^{1-{\mathfrak a}}+k\varepsilon^{1-{\mathfrak b}}}(\|g\|_2+\varepsilon^{1-{\mathfrak b}}\|w_{\tau\tau}\|_2)
  \,.
\end{equation}
As
\begin{displaymath}
  \Re\langle (\Tg_k-\lambda)w, w \rangle=\| [1+\varepsilon(\sigma+\varepsilon^{-{\mathfrak a}})]^\frac 12
  w_\tau\|_2^2 +\varepsilon\|w_\sigma\|_2^2
  -\Re\lambda\|w\|_2^2 \,, 
\end{displaymath}
we  immediately  obtain 
\begin{equation}
\label{eq:161} 
  \|  [1+\varepsilon(\sigma+\varepsilon^{-{\mathfrak a}})]^\frac 12 w_\tau\|_2\leq C(\|w\|_2+\|g\|_2) \,.
\end{equation} 
Furthermore, as $w_\sigma|_{\tau=0}\equiv0$, we have 
\begin{equation}\label{eq:7.12} 
\begin{array}{ll}
-  \Re\langle(\Tg_k-\lambda)w, w_{\tau\tau}\rangle& =\| [1+\varepsilon(|\sigma|+\varepsilon^{-{\mathfrak
      a}})]^\frac 12  w_{\tau\tau}\|_2^2 +
  \varepsilon\|w_{\tau\sigma}\|_2^2\\ &\quad  -\Im\langle w_\tau,(1+ \varepsilon ( \sigma +\varepsilon^{-{\mathfrak a}}) ) w\rangle-\Re\lambda\|w_\tau\|_2^2 \,.
  \end{array}
\end{equation}
Consequently, by \eqref{eq:161} we obtain 
\begin{equation*}
  \|w_{\tau\tau}\|_2\leq C\big([1+\varepsilon^{1-{\mathfrak a}}+k\varepsilon^{1-{\mathfrak b}}]^{1/2}\|w\|_2+\|g\|_2\big)\,,
\end{equation*}
which when substituted into \eqref{eq:160} yields
\begin{equation}
  \label{eq:162}
\| (\Tg_k-\lambda)^{-1}\|\leq  \frac{C}{\varepsilon^{1-{\mathfrak a}}+k\varepsilon^{1-{\mathfrak b}}}\,.
\end{equation}
Note, for future use, that we may conclude in addition 
\begin{equation*}
  \|w_{\tau\tau}\|_2\leq C\big( \|w\|_2+\|g\|_2\big)\,,
\end{equation*}
and hence also that
\begin{equation}
  \label{eq:163}
\|\partial_\tau^2  (\Tg_k-\lambda)^{-1}\|\leq  \frac{C}{\varepsilon^{1-{\mathfrak a}}+k\varepsilon^{1-{\mathfrak b}}}\,.
\end{equation}
We now proceed as in the proof of Proposition \ref{lem:large-sigma} to
show that $\Eg_\Tg \to0$ as $\varepsilon\to0$. Then,  as 
\begin{displaymath}
  \|(\Tg_\varepsilon-\lambda)^{-1}\|\leq C\|\Rg^{app}_{\Tg}\|\leq C\sup_{k \geq 0}
  \|(\Tg_k-\lambda)^{-1}\|\leq \frac{C}{\varepsilon^{1-2{\mathfrak a}}} \,,
\end{displaymath}
we have established (\ref{eq:158}a). The proofs of (\ref{eq:158}b) and
(\ref{eq:158}c) are now respectively deduced from \eqref{eq:161} and
\eqref{eq:163}, and the proof of (\ref{eq:158}d) follows from the
identity
\begin{displaymath}
   \Re\langle (\Tg_\varepsilon-\lambda)w,w\rangle=\|(1+\varepsilon|\sigma|)^{1/2}w_\tau\|_2^2+\varepsilon\|w_\sigma\|_2^2-\Re\lambda\|w\|_2^2 \,,
\end{displaymath}
which holds for every $w\in D(\Tg_\varepsilon)$. 
\end{proof}
\subsection{Second simplified operator}
\label{secsecond-simplified}
Let $\Lambda_\gamma^2$ be given by \eqref{eq:56}, and let $r(\varepsilon)$ satisfy
  \eqref{condsurr}. 
\begin{proposition}
  \label{prop:simplified-V2}
There exist positive
$C$ and $\varepsilon_0>0$ such that for all $\varepsilon$ in $(0,\varepsilon_0]$,  $ \partial B(\Lambda_\gamma^2,r(\varepsilon) \varepsilon)$ is included in $\rho(\Ug_\varepsilon)$ and 
\begin{equation}
\label{eq:164} 
    \|(\Ug_\varepsilon -\lambda)^{-1}\| \leq \frac{C}{(r(\varepsilon) +1-\gamma)\varepsilon}\,,\, \forall \lambda \in \partial B(\Lambda_\gamma^2,r(\varepsilon) \varepsilon)  \,.
  \end{equation}
\end{proposition}
\begin{proof}
  Let $\eta$ and $\zeta_\pm$ be defined by \eqref{eq:110} and \eqref{eq:111}
  respectively. Let $1/4<{\mathfrak a}<(1-q)/2$. Next, let
  $S_N=(-2\varepsilon^{-{\mathfrak a}},2\varepsilon^{-{\mathfrak a}})\times\R_+$ and $\Tg_N$ denote the operator
  associated with the differential operator given by \eqref{eq:151} with domain 
\begin{displaymath}
  D(\Tg_N) = \{u\in H^2(S_N)\cap H^1_0(S_N)\,|\, \tau u\in L^2(S_N) \}\,.
\end{displaymath}
Let further $S_D^+= (\varepsilon^{-{\mathfrak a}},+\infty)\times\R_+$, $S_D^-=
(-\infty,-\varepsilon^{-{\mathfrak a}}) \times\R_+$, and $\Tg_D^\pm$ denote the operator associated with the differential operator 
given by \eqref{eq:92}, whose  domain can be characterized as  
\begin{displaymath}
  D(\Tg_D^\pm) = \{u\in H^2(S_D^\pm)\cap H^1_0(S_D^\pm)\,| \,\tau u \mbox{ and } \sigma u_{\tau\tau} \in
  L^2(S_D^\pm) \}\,.
\end{displaymath}
We can now define the approximate resolvent
\begin{equation}
\label{eq:165}
  \Rg^{app}_\Tg = \eta_\varepsilon(\Tg_N-\lambda)^{-1}\eta_\varepsilon+
  \zeta_\varepsilon^-(\Tg_D^+-\lambda)^{-1}\zeta_\varepsilon^- + \zeta_\varepsilon^+(\Tg_D^+-\lambda)^{-1}\zeta_\varepsilon^+ \,.
\end{equation}
By \eqref{eq:94} we have 
\begin{equation}
\label{eq:166}
  \|(\Tg_D^\pm-\lambda)^{-1}\|+ \varepsilon^{1/2} \|\partial_\sigma(\Tg_D^\pm-\lambda)^{-1}\|\leq
  \frac{C}{\varepsilon^{1-{\mathfrak a}}}\,.
\end{equation}
We seek an estimate for $\|(\Tg_N-\lambda)^{-1}\|$. 
Let $w\in D(\Tg_N)$
and $g\in L^2(S_N)$ satisfy
\begin{equation}
\label{eq:167}
   (\Tg_N-\lambda)w=g\,.
\end{equation}
Applying the projection $\Pi_1$, given by \eqref{eq:44},  to the above
balance yields
\begin{equation}
\label{eq:168}
\begin{array}{ll}
  (\LL^+ -\varepsilon\partial^2_\sigma-\lambda)\Pi_1w &= \Pi_1g
  +\varepsilon\, |\sigma| \, \Pi_1 \partial_{\tau\tau}^2 w  
  \\
  &= \Pi_1g
  +\varepsilon|\sigma|(i\Pi_1(\tau w)-\lambda_0\Pi_1w) 
  \end{array}\,.
\end{equation}
Let 
 \begin{displaymath}
   \tilde{\Tg}_N = \LL+\varepsilon(-\partial^2_\sigma + \theta_0|\sigma|) \,,  
 \end{displaymath}
where $\theta_0$ is given by \eqref{eq:54}. 
We now rewrite \eqref{eq:168} in the form
\begin{equation}
  \label{eq:169}
(\tilde{\Tg}_N-\lambda)\Pi_1w=\Pi_1g-i\varepsilon|\sigma|\Pi_1((\tau-e^{-i\pi/3}\tau_m)w) \,.
\end{equation}
From the definition of $\tau_m$ in \eqref{eq:41} it
follows that
\begin{displaymath}
 \Pi_1\big((\tau-e^{-i\pi/3}\tau_m)w\big)= \Pi_1\big((\tau-e^{-i\pi/3}\tau_m)(I-\Pi_1)w\big)\,.
\end{displaymath}
We can thus conclude that
\begin{equation}
\label{eq:170}
  \big\|\sigma\Pi_1\big((\tau-e^{-i\pi/3}\tau_m)w\big)\big\| \leq C\varepsilon^{-{\mathfrak a}}\|(I-\Pi_1)w\|_2\,.
\end{equation}
In  the same manner as it is established in \cite[Lemma
7.1]{AGH} or in Proposition~\ref{cor:1d-2d-two} we have
\begin{subequations}
\label{eq:171}
  \begin{equation}
   \|(\tilde{\Tg}_N-\lambda)^{-1}\Pi_1\|\leq \frac{C}{(r(\varepsilon) +1-\gamma)\varepsilon} \quad ; \quad
   \|(\tilde{\Tg}_N-\lambda)^{-1}(I-\Pi_1)\|\leq C \,. \tag{\ref{eq:171}a,b}
\end{equation}
\end{subequations}
Applying  (\ref{eq:171}a) to \eqref{eq:169} yields, with the aid of
\eqref{eq:170},
\begin{equation}
\label{eq:172}
  \|\Pi_1w\|_2\leq \frac{C}{(r(\varepsilon) +1-\gamma)\varepsilon}(\|g\|_2 +
  \varepsilon^{1-{\mathfrak a}}\|(I-\Pi_1)w\|_2)\,.
\end{equation}
We now apply $I-\Pi_1$ to \eqref{eq:167} to obtain 
\begin{equation}
\label{eq:173} 
  (\LL^+ -\varepsilon\partial^2_\sigma-\lambda)(I-\Pi_1)w = (I-\Pi_1)g
  +\varepsilon|\sigma|(I-\Pi_1)w_{\tau\tau} \,.
\end{equation}
Since $I-\Pi_1$ is bounded, we have
\begin{displaymath}
  \|\sigma(I-\Pi_1)w_{\tau\tau}\|_2\leq  C\varepsilon^{-{\mathfrak a}}\|w_{\tau\tau}\|_2\,.
\end{displaymath}
We can now obtain, using  \eqref{eq:7.12},  \eqref{eq:163}, and the
fact that $|\sigma|\leq2 \varepsilon^{-2 \mathfrak a}$  in $S_N$, that
\begin{equation}
\label{eq:174}
  \|w_{\tau\tau}\|_2 \leq C(\|w\|_2+\|g\|_2)\,. 
\end{equation}
Consequently,
\begin{equation}\label{eq:152a}
  \|\sigma(I-\Pi_1)w_{\tau\tau}\|_2\leq C\varepsilon^{-{\mathfrak a} }(\|w\|_2+\|g\|_2)\,.
\end{equation}
We now apply \cite[Eq. (7.16)]{AGH} to \eqref{eq:173} to obtain, with the aid
of the above inequality,
\begin{equation}
\label{eq:175}
   \|(I-\Pi_1)w\|_2\leq C(\|g\|_2 +  \varepsilon^{1-{\mathfrak a}}\|w\|_2)\,.
\end{equation}
Substituting the above into \eqref{eq:172} yields
\begin{displaymath}
  \|\Pi_1w\|_2\leq \frac{C}{(r(\varepsilon) +1-\gamma)\varepsilon}\|g\|_2 +  \frac{C\varepsilon^{1-2{\mathfrak a}}}{r(\varepsilon) +1-\gamma}\|w\|_2\,.
\end{displaymath}
Together with \eqref{eq:122} (recall that ${\mathfrak a}<(1-q) /2$,  and $r$ satisfies \eqref{condsurr}), this yields 
\begin{equation}
\label{eq:176}
  \|(\Tg_N-\lambda)^{-1}\|\leq \frac{C}{(r(\varepsilon) +1-\gamma)\varepsilon}\,.
\end{equation}
Note, for future use, that together with \eqref{eq:174} the above
inequality implies
\begin{equation}
\label{eq:154a}
  \| \partial_\tau^2 \, (\Tg_N-\lambda)^{-1}\|\leq \frac{C}{(r(\varepsilon) +1-\gamma)\varepsilon}\,.
\end{equation}
We also need an estimate for
$\partial_\sigma(\Tg_N-\lambda)^{-1}$. To this end we rewrite \eqref{eq:169}
in the following manner 
\begin{displaymath}
   \varepsilon \Big(-\partial^2_\sigma+i\theta_0|\sigma|-\frac{\lambda-\lambda_0}{\varepsilon}\Big)\Pi_1w = \Pi_1g -i\varepsilon|\sigma|\Pi_1\big((\tau-e^{-i\pi/3}\tau_m)w\big) \,.
\end{displaymath}
Taking the inner product with $\Pi_1w$ then
yields for the real part
\begin{displaymath}
  \varepsilon\|\Pi_1w_\sigma\|_2^2 \leq C \left(\varepsilon\|\Pi_1w\|_2^2 +
  \|\Pi_1w\|_2(\|\Pi_1g\|_2+\varepsilon^{1-{\mathfrak a}}
  \|(I-\Pi_1)w\|_2\right)\,.
\end{displaymath}
Using \eqref{eq:175} we then obtain
\begin{displaymath}
  \|\Pi_1w_\sigma\|_2 \leq C(\|\Pi_1w\|_2+ \varepsilon^{\frac 32 -2 {\mathfrak a}} \|w\|_2+\varepsilon^{-1}\|g\|_2)\,,
\end{displaymath}
from which we deduce,  using \eqref{eq:176} and the fact that $\mathfrak a<(1-q)/2$, 
\begin{equation}
\label{eq:177}
  \|\Pi_1w_\sigma\|_2 \leq \frac{C}{(r(\varepsilon) +1-\gamma)\varepsilon}\|g\|_2\,.
\end{equation}
We now take the real part of the inner product of \eqref{eq:173} with
$(I-\Pi_1)w$,  to obtain with the aid of  \eqref{eq:152a} and \eqref{eq:175}, 
\begin{displaymath} 
\begin{array}{l}
  \|((I-\Pi_1)w)_\tau\|_2^2+ \varepsilon\|(I-\Pi_1)w_\sigma\|_2^2\\ \qquad\qquad   \leq C \left( \varepsilon \|(I-\Pi_1)w\|_2^2 +
  \|(I-\Pi_1)w\|_2 \left( \|g\|_2+\varepsilon^{1-{\mathfrak a}}  \|w\|_2 \right)\right) \\
  \qquad \qquad  \leq \hat C \left( \|g\|_2+\varepsilon^{1-{\mathfrak a}}  \|w\|_2 \right)^2 \,.
  \end{array}
\end{displaymath}
Hence we have obtained 
\begin{displaymath}
\|(I-\Pi_1)w_\sigma\|_2\leq C(\varepsilon^{-1/2}\|g\|_2+\varepsilon^{1/2-{\mathfrak a}}\|w\|_2)\,,
\end{displaymath}
which combined with \eqref{eq:176} and \eqref{eq:177} yields
\begin{equation}
\label{eq:178}
  \|\partial_\sigma(\Tg_N-\lambda)^{-1}\|\leq \frac{C}{(r(\varepsilon) +1-\gamma)\varepsilon}\,.
\end{equation}
We may now proceed to obtain \eqref{eq:164} in precisely the same
manner as in the proof of Proposition \ref{prop:simplified-V1}.
\end{proof}

We now return to the problem appearing in Section
\ref{sec:quazimode-V2}. The operator is defined on the quarter plane
$Q=\mathbb R_+\times \mathbb R_+$ with a Dirichlet-Neumann condition.
Hence, we consider $\QQ_\varepsilon$ to be defined by \eqref{eq:151} (via
  a Lax-Milgram theorem) and whose domain can be characterized as
  \begin{equation}
    \label{eq:179} 
  D(\QQ_\varepsilon)  = \{u\in H^2(Q)\,|\, u(0,\sigma)=0 \,;\,
   \partial_\sigma  u (0,\tau) =0 \,;\,\tau u \mbox{ and } \sigma u_{\tau\tau} \in L^2(Q) \}\,.
  \end{equation}
 We can now make the following statement
\begin{proposition}
 There exist positive $C$ and
$\varepsilon_0$ such that for every $\varepsilon\in (0,\varepsilon_0]$, $\partial B(\Lambda_\gamma^2,r\varepsilon)$ belongs to $\rho(\QQ_\varepsilon)$ and 
\begin{equation}
\label{eq:180}
 \|(\QQ_\varepsilon -\lambda)^{-1}\| \leq \frac{C}{(r(\varepsilon) +1-\gamma)\varepsilon} \,,\, \forall \lambda \in \partial B(\Lambda_\gamma^2,r(\varepsilon) \varepsilon)\,.
\end{equation}
\end{proposition}
The proof follows immediately from Proposition
\ref{prop:simplified-V2} and the fact that $\Ug_\varepsilon$ is an even 
extension of $\QQ_\varepsilon$. \\
 We shall continue to obtain results for
  $\Ug_\varepsilon$. All of them, by the same token, are valid for $\QQ_\varepsilon$ as
  well.  

Another consequence Proposition \ref{prop:simplified-V2}   is  the following:
\begin{proposition}
\label{cor:second-large-sigma}
  Under the conditions of Proposition \ref{prop:simplified-V2}, for
  every $1/4<{\mathfrak a}<(1-q)/2$ there exists $C_{\mathfrak a}>0$ and $\varepsilon_0>0$
  such that, for any $\varepsilon\in (0,\varepsilon_0]\,$, 
  \begin{equation}
\label{eq:181}
  \|{\mathbf 1}_{|\sigma|\geq 2\varepsilon^{-{\mathfrak a}}}(\Ug_\varepsilon-\lambda)^{-1}\|+ \varepsilon^{1/2} \|{\mathbf 1}_{|\sigma|\geq 2\varepsilon^{-{\mathfrak a}}}\partial_\sigma(\Ug_\varepsilon-\lambda)^{-1}\|\leq
  \frac{C_{\mathfrak a}}{\varepsilon^{1-{\mathfrak a}}}\,.
  \end{equation}
\end{proposition}
\begin{proof}
 The proof is identical with the proof of Proposition~\ref{cor:large-sigma} and is therefore omitted.
\end{proof}
We continue by proving another estimate:
\begin{proposition}
  Under the conditions of Proposition \ref{prop:simplified-V2} we have
  \begin{subequations}
\label{eq:182}
\begin{equation}
  \|\partial_\tau^2(\Ug_\varepsilon-\lambda)^{-1}\|+ \|\partial_\tau (\Ug_\varepsilon-\lambda)^{-1}\|\leq  \frac{C}{ r(\varepsilon)\varepsilon} \,,
\end{equation}
and 
       \begin{equation}
  \|\partial_\sigma^2(\Ug_\varepsilon-\lambda)^{-1}\|\leq \frac{C}{\varepsilon^{2-{\mathfrak a}}}\,.
 \end{equation}
  \end{subequations}
\end{proposition}
\begin{proof}
  The proof of (\ref{eq:182}a) follows from (\ref{eq:158}b),
  (\ref{eq:158}c), and \eqref{eq:154a}.  To obtain (\ref{eq:182}b) we
  use the identity
 \begin{displaymath}
   -\Re\langle w_{\sigma\sigma},(\Ug_\varepsilon -\lambda)w\rangle = \|w_{\tau\sigma}\|_2^2 +
   \varepsilon \, \|w_{\sigma\sigma}\|_2^2 + \varepsilon\, \Im\langle w_\sigma,w_{\tau\tau}\sign\sigma\rangle -\Re\lambda \|w_\sigma\|_2^2\,.
 \end{displaymath}
 Hence,
 $$
   \varepsilon \, \|w_{\sigma\sigma}\|_2^2 \leq C \left( \| w_\sigma\|^2_2 + \varepsilon^2 \|w_{\tau\tau}\|^2_2 + \varepsilon^{-1} \|g\|^2_2\right) \,,
  $$
  which together with (\ref{eq:182}a)  and (\ref{eq:158}d) leads to 
 \begin{equation*}
   \|w_{\sigma\sigma}\|_2\leq \frac{C}{\varepsilon^{1/2}}( \varepsilon ^{\mathfrak a -\frac 32} + \frac{1}{r(\varepsilon)} +\varepsilon^{-1/2})\, \|g\|_2)\,.  
 \end{equation*}
 \end{proof}
As in the previous section we also need the following estimate:
\begin{proposition}
\label{lem:large-tau-V2}  
Let $0<a$. Then, there exist positive
$C_a$ and $\varepsilon_0$ such that for all $\varepsilon\in (0,\varepsilon_0]$,   $\partial
B(\Lambda_\gamma^2,r(\varepsilon) \varepsilon)$ belongs to $\rho(\Ug_\varepsilon)$ and for all  $\lambda\in\partial
B(\Lambda_\gamma^2,r(\varepsilon) \varepsilon)$\,, 
\begin{subequations}
\label{eq:183}
  \begin{equation}\
    \|{\mathbf 1}_{\tau\geq\varepsilon^{-a}}(\Ug_\varepsilon -\lambda)^{-1}\| + \varepsilon^{a/2}\|{\mathbf 1}_{\tau\geq\varepsilon^{-a}}\partial_\tau(\Ug_\varepsilon -\lambda)^{-1}\| +
   \varepsilon^{a/2+ 1/2}\|{\mathbf 1}_{\tau\geq\varepsilon^{-a}}\partial_\sigma(\Ug_\varepsilon -\lambda)^{-1}\|
   \leq C_a\varepsilon^a \,,\,   \end{equation}
and
\begin{equation}
   \|{\mathbf 1}_{\tau\geq\varepsilon^{-a}}\partial_\tau^2(\Ug_\varepsilon -\lambda)^{-1}\| +
   \varepsilon\|{\mathbf 1}_{\tau\geq\varepsilon^{-a}}\partial_\sigma^2(\Ug_\varepsilon -\lambda)^{-1}\| \leq
   C_a \,.
\end{equation}
\end{subequations}
\end{proposition}
\begin{proof}
     As in the proof of Proposition \ref{lem:large-tau-V1}, let $\zeta_+$ be given by \eqref{eq:110}.
    Let further
     $\zeta_\xi(\tau)=\zeta_+(2 \tau/\xi)$. Let $w\in D(\Ug_\varepsilon)$ and
 $g\in L^2(\R^2_+)$ satisfy
 \begin{displaymath}
    (\Ug_\varepsilon -\lambda)w=g \,.
 \end{displaymath}
As, with $\zeta=\zeta_\xi$, 
\begin{equation}
\label{eq:184}
\begin{array}{l}
\Re\langle (\Ug_\varepsilon -\lambda)w, \zeta^2w\rangle + \Im
\langle (\Ug_\varepsilon -\lambda)w, \zeta^2w\rangle \\
\quad  =
\|[1+\varepsilon|\sigma|]^{1/2}\partial_\tau(\zeta w)\| ^2
\\
\qquad
 -(\Im\lambda+\Re\lambda)\|\zeta w\| ^2 +
\varepsilon\|\zeta \partial_\sigma
w\| ^2 -\|w\zeta^\prime\|^2 \\
 \qquad + 2\Im\langle
w\zeta^\prime,\partial_\tau(\zeta w)\rangle +
\|\tau^{1/2}\zeta
w\| ^2\,,
\end{array}
\end{equation}
we get first 
\begin{displaymath}
  \|\tau^{1/2}\zeta_\xi w\|_2^2- (\Im\lambda+\Re\lambda)\|\zeta_\xi w\|_2^2 \leq
  2\|\zeta_\xi^\prime w\|_2^2+2\|\zeta_\xi w\|_2\|\zeta_\xi g\|_2\,.
\end{displaymath}
The proof then continues along the same lines of the proof of
Proposition \ref{lem:large-tau-V1} and is therefore omitted.
   \end{proof}

\section{Upper bound}
\label{sec:upper}

\subsection{Goals and notation}
Let $x_0\in \Sg^m$ (resp. $x_0\in \hat{\Sg}^m$), where $\Sg^m$ (resp.
$\hat{\Sg}^m$) is defined by \eqref{eq:4a} --\eqref{eq:8}
(resp.\eqref{eq:4av2}--\eqref{eq:8v2}), for type V1 (resp.
V2) potentials.  We have proven in Proposition \ref{prop3.2} (resp.
Proposition \ref{prop4.1}) the existence of an approximate eigenvalue
$\hat \Lambda^1(h,x_0)$ (resp. $\hat \Lambda^2(h,x_0)$) with a corresponding approximate
eigenstate localized (as $h\ar 0$) near $x_0$. In this section, we
prove the existence of an eigenvalue inside the disk 
$B(\hat{\Lambda}^i(h,x_0),\hat r_i (h) h^{k_i})$ where $i=1$ (resp. $i=2$)
corresponds to potentials of type V1 (resp. V2),
$k_1=4/3$ and $k_2=10/9$, $\hat \Lambda_1(h,x_0)$ (resp. $ \hat \Lambda^2(h,x_0)$) is 
defined, for $c(x_0)>0$,  in  \eqref{eqdefhatLambda}  (resp. in  \eqref{eq:44a}) by:
\begin{equation*} 
\begin{array}{ll}
\hat \Lambda^1(h,x_0) &= i V(x_0) + (c(x_0) h)^{2/3} ( \lambda_0 + \ego(h) \lambda_1) \,, \\
\hat \Lambda^2(h,x_0) &= i V(x_0) +  (c(x_0) h)^{2/3} (\lambda_0 + \varepsilon(h) \check \lambda_1 )\,,
\end{array}
\end{equation*}
and for $c(x_0) <0$ by
\begin{equation*} 
\begin{array}{ll}
\hat \Lambda^1(h,x_0) &= i V(x_0) + (-c(x_0) h)^{2/3} \overline{( \lambda_0 + \ego(h) \lambda_1) }\,,\\
\hat \Lambda^2(h,x_0) &= i V(x_0) +  (- c(x_0) h)^{2/3} \overline{(\lambda_0 + \varepsilon(h)  \check \lambda_1 )}\,.
\end{array}
\end{equation*}

We also recall from \eqref{eq:35} and \eqref{eq:49} that
$$
  \ego(h) = (2^{-1/2}\, \alpha_m^{1/2} J_m^{-\frac 56} )\,h^{2/3}\,\mbox{ and }\,\varepsilon(h) = [2^6 \hat \alpha_m^6 J_m^{-8}]^{1/9}\, h^{4/9}\,,
 $$ 
 and keep in mind Remark \ref{rem:sign-c} and that $|c(x_0)|=J_m$.\\
 The functions $\hat r_i(h)$ ($i=1,2$) are determined from $r_1(\ego)$
 and $r_2(\varepsilon)$, respectively appearing in Sections~\ref{sec:simplified}
 and \ref{sec:simplified-V2}, via the relations
\begin{equation}\label{relhatr}
J_m^{-\frac 23} h^{k_1-2/3} \hat r_1 (h) = \ego(h)  r_1(\ego(h)) \mbox{ and }   J_m^{-\frac 23} h^{k_2-2/3} \hat r_2 (h) = \varepsilon(h) r_2(\varepsilon(h))\,.
\end{equation}
We now choose
\begin{equation}
  \hat{r}_1(h) = h^{\hat q_1}  \,\mbox{ and }  \hat{r}_2(h) = h^{\hat q_2} \,,
\end{equation}
where,
\begin{equation}
\label{defhatq} 
\hat q_1= \frac 23 q \quad \mbox{ and  } \quad \hat q_2 = \frac 49 q\,.
\end{equation}
With this choice and from \eqref{eq:35} and \eqref{eq:49}, we get 
\begin{equation} \label{defr1r2}
r_1(\ego) =  \left( 2^{\frac{1+q}{2} } J_m^{\frac{ (1+5q)}{6} }\, \alpha_m^{-\frac{1+q}{2}} \right) \,\ego^q \,,\, r_2(\varepsilon) = \left( J_m^{\frac{2+8q}{9}}\, 2^{-\frac{2(1+q)}3}\, \hat \alpha_m^{-\frac{ 2(1+q)}{3} } \right) \,  \varepsilon^q\,.
\end{equation}
Recall that $q<1/6$ (in both Sections \ref{sec:simplified} and
\ref{sec:simplified-V2}) and hence
\begin{subequations}
\label{eq:condhatq}
   \begin{equation}
\hat q_1<1/9 \mbox{ and } \hat q_2<2/27\,.  \tag{\ref{eq:condhatq}a,b}
\end{equation}
\end{subequations}

 Using the preliminary estimates
established in the previous sections, we obtain in this section a
bound on the resolvent norm of $\A_h$ given by \eqref{eq:1} on a
suitable circle centered at $\hat \Lambda^i(h)$.  The method is similar to
the one used in \cite{Hen,AGH}, i.e., we obtain localized
approximations of the resolvent $(\A_h-\lambda)^{-1}$, that facilitate the
application of the various estimates obtained in Sections
\ref{sec:simplified} and \ref{sec:simplified-V2}. The combination of
these estimates with the construction of quasi-modes leads to the
proof of existence of an eigenvalue.

\subsection{Refined partition of unity}
\label{sec:refin-part-unity}
We start from the partition of unity of size $h^\varrho$ constructed in
paragraph \ref{sss2.3.1} (which we denote in this section by $(\bar
\chi_{j,h},\bar \zeta_{k,h})$ to avoid the confusion with a future
notation), with $j\in \Jg_i$, $k\in \Jg_\pa$.  Note that $\bar \chi_{j,h}$
and $\bar \zeta_{k,h}$ are respectively supported in $B(a_j,h^\varrho)$ or
$B(b_k,h^\varrho)$. Recall also the decomposition of $\Jg_\partial$ in Section 2
into three disjoint subsets $\Jg_\pa^c$, $\Jg_\pa^D$ and $\Jg_\pa^N$,
so that all corners are included in $\Jg_\pa^c$.  To simplify our
resolvent construction, we impose in addition the condition
\begin{equation}
\label{eq:164add}
 \Se \subset   \Union_{k\in  \Jg_\pa }\{b_k\}   \,. 
\end{equation}
When the potential is of type V1, we split further $\Jg_\partial^D$ by setting
\begin{displaymath}
  \Jg_\partial^s =\{ k\in\Jg_\partial^D \, | \, b_k\in\Se \}  \mbox{ and }  \Jg_\partial^r=\Jg_\partial^D\setminus\Jg_\partial^s\,.
\end{displaymath}
When the potential is of type V2, we set
\begin{displaymath}
  \Jg_\partial^s =\{ k\in\Jg_\partial^c \, | \, b_k\in\Se  \}\mbox{ and }  \Jg_\partial^r=\Jg_\partial^D\,.
\end{displaymath}
We further set
\begin{displaymath}
  \Jg_\partial^{s,0}= \{ k\in\Jg_\partial^s \, | \, V(b_k)\neq V(x_0)\} \,.
\end{displaymath}
 As in \cite{AGH} (Subsection
7.3) we need to use two different scales (or disk sizes), i.e. as
before $h^\varrho$ for  $k \in(\Jg_\pa \setminus\Jg_\partial^s)\cup \Jg_\partial^{s,0}$ or $j\in \Jg_i$  but now
$h^{\varrho_\perp}$  for $k\in \Jg_\partial^s\setminus\Jg_\partial^{s,0}$  
where 
\begin{equation}\label{condrhorhoperp}
\frac 23 > \varrho > \frac 13 > \varrho_\perp >0\,.
\end{equation}
 We  now proceed in two steps.\\
We first construct a finite (independent of
$h$) partition of unity of size $h^{\varrho_\perp}$,  $\check  \xi_{h},
\check \zeta_{k,h}$ with $k\in \mathcal \Jg^s_\partial$ such that
\begin{subequations}\label{eq:185}
   \begin{equation}
   \check \xi_h^2 + \sum_{k\in \Jg_\partial^s\setminus\Jg_\partial^{s,0}} \check \zeta_{k,h}^2 = 1 \mbox{
     in } \Omega \,,
   \end{equation}
 with 
$$
\check \zeta_{k,h}\equiv1 \mbox{ in } B(b_k,h^{\varrho_\perp}/2)\,,\,
\check \zeta_{k,h}\equiv0 \mbox{ in } \Omega\setminus B(b_k,h^{\varrho_\perp})\,,
$$
\begin{equation}
  |\nabla\check \zeta_{k,h}| + h^{\varrho_\perp}|D^2\check \zeta_{k,h}|\leq C\,h^{-\varrho_\perp} \,, \;
  \forall k \in {\Jg_\partial^s\setminus\Jg_\partial^{s,0}}\,.
\end{equation} 
and
\begin{equation}
  |\nabla \check  \xi_{h}| + h^{\varrho_\perp}|D^2 \check  \xi_{h}|\leq C\, h^{-\varrho_\perp} \,.
  \end{equation}
  \end{subequations}

Then, we set for $k\in \Jg_\pa$ 
\begin{displaymath} 
  \tilde{\zeta}_{k,h} = \bar \zeta_{k,h}\, \check  \xi_h \quad   \tilde{\chi}_{j,h} = \bar \chi_{j,h}\, \check  \xi_h \,.
\end{displaymath}
 To satisfy the Neumann boundary condition on $\partial\Omega_N$, and for later
  reference, we introduce an additional condition
  \begin{equation}
\label{eq:186}
 \frac{\pa \check \xi_{h}}{\pa \nu} \big |_{\partial\Omega}=0 \quad ; \quad \frac{\pa  \bar \zeta_{k,h}}{\pa \nu} \big |_{\partial\Omega}=0 \,. 
  \end{equation}
As $\varrho>\varrho_\perp$, we have  for  sufficiently small $h_0>0$, $
  \tilde{\zeta}_{k,h} \equiv 0$ for $k\in \Jg^s$ and $h\in (0,h_0]$. Note
  however that  we may have  for  sufficiently small $h$, 
  $j\in\Jg_i$, and $k\in \Jg_\partial^s\setminus\Jg_\partial^{s,0}$  the inclusion  $B(a_j,h^\varrho)\subset B(b_k,h^{\varrho_\perp} /2)$
  and hence $\tilde{\chi}_{j,h}\equiv0\,$. A similar observation can be made for
$ k\in (\Jg_\pa \setminus\Jg_\partial^s)\cup \Jg_\partial^{s,0}$.  We thus define
\begin{displaymath}
\begin{array}{ll} 
  \tilde{\Jg}_i= \{j\in\Jg_i \,| \, \tilde{\chi}_{j,h}\not\equiv0 \}\,,& \,  \tilde{\Jg}_\partial^N= \{k\in\Jg_\partial^N \,| \,
  \tilde{\zeta}_{k,h}\not\equiv0 \}\,,\\
    \tilde{\Jg}_\partial^D= \{k\in\Jg_\partial^D \,|  \tilde{\zeta}_{k,h}\not\equiv0 \}\,, & \, \tilde{\Jg}_\partial^r= \{k\in\Jg_\partial^r
  \,| \, \tilde{\zeta}_{k,h}\not\equiv0 \} \,.
  \end{array}
\end{displaymath}
Clearly, we have 
\begin{displaymath} 
  \sum_{k\in\Jg_\partial^s\setminus \Jg_\pa^{s,0}} \check \zeta_{k,h}^2+    \sum_{
    k\in (\tilde \Jg_\partial^c\setminus\Jg_\partial^s)\cup \Jg_\pa^{s,0}} \tilde
  \zeta_{k,h}^2+ \sum_{k\in \tilde{\Jg}_\partial^N\cup \tilde{\Jg}_\partial^r} 
  \tilde{\zeta}_{k,h}^2 +  \sum_{j\in\tilde{\Jg}_i}
  \tilde{\chi}_{j,h}^2 =1 \mbox{ in } \Omega \,.  
\end{displaymath}
Note that for type V1 potentials, $\Jg_\partial^c\setminus\Jg_\partial^s=\Jg_\partial^c$,
whereas for type V2, $\Jg_\partial^D\setminus\Jg_\partial^s=\Jg_\partial^D$. For simplicity of notation we
drop the tilde and check accents in the sequel and use
$(\chi_{j,h},\zeta_{k,h})$ instead of $(\tilde{\chi}_{j,h},\tilde{\zeta}_{k,h})$
and $(\Jg_i,\Jg_\partial^D, \Jg_\partial^N,\Jg_\partial^r)$ instead
of $(\tilde{\Jg}_i, \tilde {\Jg}_\partial^D, \tilde{\Jg}_\partial^N,\tilde{\Jg}_\partial^r)$. \\
Note further that by the previous construction we have that
\begin{subequations}
  \label{eq:187}
  \begin{equation}
    \begin{cases}
   |\nabla\chi_{j,h}|+h^{\varrho_\perp}|D^2\chi_{j,h}|\leq C\, h^{-\varrho_\perp} & \text{ in }
B(a_j,h^{\varrho}/2) \\ |\nabla\chi_{j,h}|+h^\varrho|D^2\chi_{j,h}|\leq C\, h^{-\varrho}&  \text{ in }
B(a_j,h^{\varrho})  
    \end{cases}
 \,,\,\forall j\in\Jg_i   \,.
\end{equation}
At the boundary, we have
\begin{equation}
  \begin{cases}
    |\nabla\zeta_{k,h}|+h^{\varrho_\perp}|D^2\zeta_{k,h}|\leq Ch^{-\varrho_\perp}& \text{ in }
B(b_k,h^{\varrho}/2) \\ |\nabla\zeta_{k,h}|+h^{\varrho}|D^2\zeta_{k,h}|\leq Ch^{-\varrho} &
\text{ in }
B(b_k,h^{\varrho}) 
  \end{cases}\,.
  \end{equation}
As in Section \ref{s2}, each point of
$\Omega$ belongs to at most $N_0$ disks with $N_0$ independent of $h$,
and
\begin{equation}
\sum_j |\pa^{\hat \gamma}\chi_{j,h}(x)|^2 + 
\sum_k |\pa^{\hat \gamma}\zeta_{k ,h}(x) |^2 \leq C_{\hat \gamma}  \, h^{-2 |\hat \gamma |{\varrho}}\,,\, \forall  \hat \gamma \in \mathbb N^2 \mbox{ s.t. }|\hat \gamma|\leq 2\,.
\end{equation}
\end{subequations}
As in Section \ref{s2} once again, we introduce $\eta_{k,h} = 1_{\Omega}\, \zeta_{k,h}\,. $\\
Note that, as a result of \eqref{eq:186}, we have
\begin{equation}\label{eq:188}
  \frac{\pa \zeta_{k,h}}{\pa \nu} \big |_{\pa \Omega}=0 \,,\, \forall k \in \mathcal \Jg_\pa^N\,.
  \end{equation}
To be compatible with future constraints we impose from now on the
following restrictions for potentials of type V1
\begin{subequations}
  \label{eq:189}
\begin{equation}
 \frac{13}{63} < \varrho_{\perp}  < \frac 29\, \mbox{ and }\,
\frac{\varrho_\perp}{2}+\frac{1}{3}<  \varrho < \frac 23 -\varrho_\perp\,,
\end{equation}
and for potentials of type V2 
\begin{equation}
 \frac{7}{27} < \varrho_{\perp}  < \frac 13 \,\mbox{ and }\, 
\frac{\varrho_\perp}{2}+\frac{1}{3}<  \varrho < \frac 23 -\frac{\varrho_\perp}{2} \,.
\end{equation}
\end{subequations}
These new conditions are, clearly, more restrictive than the ones
previously given in \eqref{condrhorhoperp}. \\
As in Section \ref{s2} (see \eqref{eq:22}) the approximate resolvent has the form
\begin{equation}
\label{eq:189aa}
 \mathcal{R}(h,\lambda) = \sum_{j\in \mathcal \Jg_i(h)}\chi_{j,h}(\A_{j,h}-\lambda)^{-1}\chi_{j,h} 
+  \sum_{k\in \mathcal \Jg_\partial (h)}\eta_{k,h}  R_{k,h}(\lambda) \eta_{k,h}\,,
\end{equation}
but this time we need to estimate the localized resolvents (some of
them account now for higher order terms in the Taylor expansion of $V$
near $b_k$) $R_{k,h}(\lambda)$, and the remainder
 \begin{equation}\label{defEg}
  \Eg(h,\lambda) = (\A_h -\lambda)\mathcal{R}(h,\lambda) - I\,,
\end{equation}
 for $\lambda \in \partial B (\hat \Lambda^i(h), h^{k_i} \hat r_i (h))$. 
 
\subsection{Localized resolvent estimates}
\label{sec:local-resolv-estim}
  \subsubsection{Approximation associated  with $j\in \Jg_i$ and $k\in \Jg_\pa^N$}
In this case we can directly  apply the estimates \eqref{eq:18} and
\eqref{eq:20} given in Section~\ref{s2}, as $\partial
B({\hat{\Lambda}}^i(h),\, 
  h^{k_i}\hat{r}_i(h)))$ is included in $\{ \Re \lambda \leq \omega \, h^\frac 23\}$
  for any $\omega > J_m^{\frac 23}  \frac{|\nu_1|}{2} $. 

\subsubsection{Decomposition of $\Jg_\partial^s$}
Having in mind the definition of $\Jg_\partial^{s,0}$ we further split
$\Jg_\partial^s\setminus \Jg_\partial^{s,0}$ in the following manner: 
\begin{enumerate}
\item $\Jg_\partial^{s,1}=\{k \in \Jg_\partial^s \, | \,
 V(b_k)= V(x_0)\,;\;   |\alpha(b_k)|>\alpha_m\,\}$,
\item $\Jg_\partial^{s,2}=\{k\in \Jg_\partial^s \, | \, V(b_k)= V(x_0)\,;\;    |\alpha(b_k)|=\alpha_m\,\}$,
\end{enumerate}
where $\alpha(x)$ is given by  \eqref{eq:6} for potentials of type V1 (resp.
by \eqref{defhatalpha}   for potentials of type V2, 
  where in this case $\alpha$ is replaced by $\hat \alpha$ and $\alpha_m$ by $\hat
  \alpha_m$).  \\
 Note that by Remark  \ref{rem:sign-c}, when $V(b_k)= V(x_0)$,
 we have $c(b_k) c(x_0) >0$ as $x_0$ and $b_k$ belong to the same
 connected component of $\overline{\pa \Omega^D}$.

\subsubsection{Localization associated  with $k\in \Jg^{s,0}_{\pa }$}
For potentials of type $V1$, we have for any $k\in \Jg^{s,0}_{\pa }$,
$b_k\in \pa \Omega_D$ and we can therefore use the approximate operator
$\tilde\A_{k,h}$ introduced in \eqref{eq:183a}. Upon dilation 
we then use Lemma \ref{lemma4.12} (with $\nu = \left[
  (V(b_k)- \Im \lambda )\right](J_mh)^{-\frac 23}$ and $\mu = \Re \lambda\,
(J_m h)^{-\frac 23}$) to obtain that  $\partial B({\hat{\Lambda}}^1(h),  h^{k_1}\hat{r}_1(h)) \subset \rho (\tilde{\A}_{k,h})$ and 
\begin{equation}  
\label{eq:190}
\max_{k\in \Jg_\partial^{s,0}}\sup_{\lambda\in\partial B({\hat{\Lambda}}^1(h),     \,
  h^{k_1}\hat{r}_1(h))}\|(\tilde{\A}_{k,h}-\lambda)^{-1}\| \leq
\frac{C}{h^{2/3}} \,. 
\end{equation}

For potentials of type $V2$, we have for any $k\in \Jg^{s,0}_{\pa }$,
$b_k\in \Jg_\pa^c$ and we may therefore use the approximate operator
$\tilde\A_{k,h}$ introduced in \eqref{deftildeAkcorner}.
After dilation we may apply Lemma~\ref{lemma4.12-corner} to obtain
that $\partial B({\hat{\Lambda}}^2(h), h^{k_2}\hat{r}_2(h)) \subset \rho
(\tilde{\A}_{k,h})$ and  
\begin{equation}
\label{eq:191}
  \max_{k \in \Jg_\partial^{s,0}}\sup_{\ \lambda\in\partial B(\hat{\Lambda}^2(h) ,\,
    \hat r_2(h) h^{k_2})}\|(\tilde{\A}_{k,h} -\lambda)^{-1}\| \leq \frac{C}{h^{2/3}}
  \,. 
\end{equation}

\subsubsection{$k\in \Jg_\pa^s \setminus
  \Jg_\pa^{s,0}$} 
For $k\in\Jg_\partial^s \setminus \Jg_\pa^{s,0}$,  we need the approximation of
$\A_h$, to be more refined than \eqref{eq:183a} or
\eqref{deftildeAkcorner}. We thus  introduce, for potentials of type V1, (see \eqref{eq:33})
  \begin{equation}   
\label{eq:192}
\begin{cases}
 \tilde\A_{k,h}= -h^2\Delta_{s,\rho}+h^2\kappa(b_k)\partial_\rho-2h^2\kappa(b_k)\rho\chi(h^{-2(1-\tilde{b})/3}\rho)\partial^2_s+iV_{b_k}^{(2)}\,,\\
\Dg(\tilde \A_{k,h}) = \{ u\in H^2(\R^2_+)\cap H^1_0(\R^2_+)\, | \, \rho(1+s^2)u
\in L^2(\R^2_+)  \} \,,
\end{cases}
\end{equation}
where 
\begin{equation}\label{defVbk}
  V_{b_k}^{(2)}=V(x_0) \pm 
  J_m\rho+ \frac{1}{2}\alpha(b_k) s^2\rho+ \frac{1}{2}\hat{\beta}_k \chi(h^{-2(1-b)/3}\rho)\rho^2\,.
\end{equation}
The curvilinear coordinates $(s,\rho)$, defined in Section
\ref{sec:local-coordinates}, are centered at $b_k$ (see Remark
\ref{remext}). The curvature $\kappa$ is approximated in
(\ref{eq:192}) by its value at $b_k$ and $\hat{\beta}_k =\hat \beta (b_k) $ is given by \eqref{defalphac}.\\
The cutoff function $\chi$ is the restriction to $\R_+$ of
\eqref{eq:59}, the positive parameters $b$ and $\tilde{b}$
satisfy the limitations set in Sections \ref{sec:v1-potentials:-1d} and
\ref{sec:simplified}, i.e.,
\begin{equation}\label{eq:244}
  \frac{1}{2}<b<\frac{3}{4} \quad ; \quad 0<\tilde{b}<\frac{1}{2}-q \,.
\end{equation}
Further restrictions for  $b$ and $\tilde{b}$ will be imposed at a later stage. 

For potentials of type V2 we use the coordinates introduced in
Paragraph \ref{sss2.3.4}, centered at the corner $b_k$, and consider the
approximate operator
\begin{equation}  \label{deftildeAkhV2}
\begin{cases}
 \tilde\A_{k,h}= -h^2[(1+ \tilde{\alpha}_{b_k}s)\partial^2_\rho+\partial^2_s]+i\big(V(b_k) \pm
  J_m \rho\big)\,,\\
\Dg(\tilde \A_{k,h}) = \{ u\in H^2(Q) \, | \,  u_{\partial Q_\|}=0 \,;\,
\partial_\nu u_{\partial Q_\perp}=0 \,;\, \rho \, u \mbox{ and } s\,\partial^2_\rho u \in L^2(Q)  \} \,,
\end{cases}
\end{equation}
where the coordinates $(s,\rho)$ are given by \eqref{eq:24}, and $\hat \alpha_{b_k} $ is
the same as in \eqref{eq:27} with $\mathfrak c=b_k\,$.

For potentials of type $V1$ (with Remark \ref{remext} in mind
  once again) we apply to (\ref{eq:192}) the dilation
\begin{equation}     
\label{eq:193}
 s=\Big[\frac{J_mh^4} {8 |\alpha(b_k)|^3}\Big]^{1/12}\sigma
\quad ; \quad \rho=\Big[\frac{h^2}{J_m}\Big]^{1/3}\tau 
 \,,
\end{equation}
to  obtain the unitary equivalent operator (for $c(x_0) >0$ which
is equivalent to $c(b_k) >0$)
\begin{equation}
  \label{eq:194}
i V (x_0) +  [hJ_m]^{\frac 23}\, \tilde{\B}_{\ego_k} \,.
\end{equation}
In the above, $\tilde{\B}_\ego $ is defined in
\eqref{eq:92}-\eqref{eq:134}-\eqref{eq:148}, and (see also
\eqref{defomega})
$$
\begin{array}{ll}\displaystyle
\beta_k= \frac{\hat{\beta}_k}{[8|\alpha(b_k)|J_m]^{1/2}} \,, & 
 \omega =  \kappa (b_k) \Big[2 \frac{J_m}{ \alpha(b_k) }\Big]^{1/2} \,, \\ & \\
  \theta=2^\frac 32\, \displaystyle \frac{J_m^{1/2}\kappa(b_k)}{|\alpha(b_k)|^{1/2}} \,,&  \ego_k(h)  = \, J_m^{-5/6}2^{-\frac 12}\, |\alpha(b_k)|^{1/2} \,h^{2/3}\,.
  \end{array}
$$
When $c(x_0) <0$ (and hence $c(b_k) < 0$) we obtain
$\overline{\tilde{\B}_\ego}$ instead of $\tilde{\B}_\ego$ in \eqref{eq:194}. Note that since  $\Se$ is finite, there exists $C>1$
such that, $\forall k \in \Jg_\pa^s \setminus
  \Jg_\pa^{s,0} $, 
\begin{equation}
  \label{eq:242}
1\leq \ego_k(h)/ \ego(h) = \alpha_m^{-\frac {1}{2}} |\alpha(b_k)|^{\frac 12} 
\leq C\,.
\end{equation}
Note further, that since $\lambda \in \pa
B ({\hat{\Lambda}}^1(h), \hat r_1(h) h^{k_1})$  we obtain, in view of
\eqref{eq:194}, 
$$\check \lambda:=[hJ_m ]^{-\frac 23} (\lambda - i V(x_0)) \in\partial
B(\Lambda_{\gamma_k}^1(\ego_k(h)),r_{1} (\ego_k(h);b_k)\ego_k(h) )
$$
where, for $k\in\Jg_\partial^{s,2}$,
$$
r_{1}( \ego;b_k) =  r_1(\ego;x_0) =r_1(\ego)\,,\, \gamma_k =1\,,\, \ego_k(h)=\ego(h)\,,
$$
and for $k\in\Jg_\partial^{s,1}$,
$$
 r_{1} (\ego;b_k) =   \left( 2^{\frac 12  (1+q)} J_m^{\frac {1+5q}{6 } } |\alpha(b_k)| ^{-\frac{1+q}{2}} \right) \ego^{q} \,,\, \gamma_k=[\alpha_m/|\alpha(b_k)|]^{1/2}\,.
$$ 
Note that  by  \eqref{defr1r2}  there exists $c>0$ such that, $\forall k \in \Jg_\pa^s \setminus
  \Jg_\pa^{s,0} $,  
$$
c \leq r_{1} (\ego;b_k) /r_1(\ego) = \gamma_k^{1+q}\leq 1\,.
$$
By
(\ref{eq:149}a) we then have  that $\partial B({\hat{\Lambda}}^1(h),
  h^{k_1}\hat{r}_1(h)) \subset \rho (\tilde{\A}_{k,h})$ and 
\begin{equation}     
\label{eq:195} 
\begin{array}{ll}
\sup_{k\in \Jg_\partial^{s,1} \cup \Jg_\partial ^{s,2}}\sup_{\lambda\in\partial
  B(\hat{\Lambda}^1(h),\hat r_1 (h) h^{k_1})}\|{\mathbf 1}_{B(b_k,h^{\varrho_\perp})}(\tilde{\A}_{k,h} -\lambda)^{-1}\tilde{\eta}_{k,h}\|& \leq
\frac{C}{h^{k_1+ \hat q_1}}\,,
\end{array}
\end{equation}
where the cutoff function
$\tilde{\eta}_{k,h}$ is given by the boundary operator defined in paragraph~\ref{sec:conclusion} (see also \cite{AGH,Hen})
\begin{displaymath}
  \tilde{\eta}_{k,h} =T_{\mathcal{F}_{b_k}}(\eta_{k,h}) \,.
\end{displaymath}

To show that (\ref{eq:149}a) can be applied we first observe that, for
all $x=(s,\rho) \in B(b_k,h^{\varrho_\perp})$ we have, by \eqref{eq:193} and
\eqref{eq:35}, that
$$
\tau\leq  J_m^{\frac 13} h^{\varrho^\perp-\frac 23} =J_m^{1/3} (J_m^{5/6} 2^\frac
12 |\alpha (b_k)|^{-\frac 12})^{3\varrho_\perp/2-1}  \, \ego_k(h)^{3\varrho_\perp/2-1} \leq C  \, \ego_k(h)^{3\varrho_\perp/2-1}\,,
$$ 
where $C$ is independent of $k$. Hence, for $1-3\varrho_\perp/2 < a' < 1$, there exists $h_0>0$, such that for all
$h\in (0,h_0]$ and $k \in \Jg_{\pa}^s\setminus \Jg_{\pa}^{s,0}$ we have
$\tau\leq\ego_k(h) ^{-a^\prime}$, which is precisely what we need to apply (\ref{eq:149}a).

Similarly, for potentials of type V2 we apply (with Remark
\ref{remext} and \eqref{eq:49-1} in mind) the transformation
\begin{equation}     
\label{eq:196}
\rho=\Big[\frac{h^2}{J_m}\Big]^{1/3}\tau \quad ; \quad
s=\Big[\frac{J_mh^4}{2 |\hat \alpha(b_k)| ^3}\Big]^{1/9}\sigma \,,
\end{equation}
and set (see \eqref{eq:49}), 
$$ 
\gamma_k=[\hat \alpha_m/|\hat \alpha(b_k)|]^{2/3} \quad \mbox{  and } \quad 
\varepsilon_k(h)= \gamma_k^{-1}\,\ \varepsilon (h)\,.
$$
As in the case of potentials of type V1, we have
$$
r_2(\varepsilon;b_k) = (2^{-\frac{2 (1+q)}{3}}\, |\hat \alpha (b_k)|^{-\frac{2(1+q)}{3}}\, J_m ^{\frac{2(q+4)}{9}})\, \varepsilon^q 
= \gamma_k^{1+q}\,r_2(\varepsilon) \,.
$$
 We then obtain via \eqref{eq:164}  that there exists $h_0>0$ such
 that, for $h\in (0,h_0]$,  $\pa B (\hat{\Lambda}^2(h),\hat 
r_2(h) h^{k_2}) \subset \rho(\tilde{\A}_{k,h})$ and  
\begin{equation}     
\label{eq:197}
\sup_{k \in \Jg_\partial^s\setminus \Jg_\pa^{s,0}}\sup_{\lambda\in\partial
  B(\hat{\Lambda}^2(h),\hat r_2(h) h^{k_2})}\|(\tilde{\A}_{k,h} -\lambda)^{-1}\tilde{\eta}_{k,h}\|\leq
\frac{C}{h^{k_2+\hat q_2}}\,.
\end{equation}

\subsubsection{$k\in \Jg_\partial^c\setminus\Jg_\partial^s$.}
In this case we use as our approximate operator
\begin{equation}  
\label{eq:198}
\begin{cases}
 \tilde\A_{k,h} = -h^2\Delta_{s,\rho}+i\big(V(b_k)\pm 
  \jg_k\rho\big) \\
\Dg(\tilde \A_{k,h}) = \{ u\in H^2(Q) \, | \, u|_{\partial Q_\perp}=0\,,\; \partial_\nu u_{\partial Q_\|}=0
\,;\; \rho u\in L^2(Q)  \} \,.
\end{cases}
\end{equation}
Applying the dilation 
\begin{equation}
  \label{eq:241}
(\rho,s)=h^{2/3}(x_1,x_2)\,,
\end{equation}  
$\tilde\A_{k,h}$ is transformed into \eqref{eq:11}. By
\eqref{eq:12}, inverse dilation, and the fact that $|\jg_k| > J_m$  we have
\begin{equation}    
\label{eq:199} 
\sup_{\lambda\in\partial B(\hat{\Lambda}^i (h),\, h^{k_i}\hat{r}_i (h))}\ \|(\tilde{\A}_{k,h}-\lambda)^{-1}\|\leq
\frac{C}{h^{2/3}}\,,
\end{equation}
where $i=1$ for potentials of type V1 and $i=2$ for potentials of type
V2.

\subsubsection{$k\in \Jg_\partial^r$.}
In this case, we use  \eqref{eq:183a} as the approximate resolvent. Let
\begin{equation*}
 \delta_k(h)= \jg_k -J_m= |\nabla V (b_k(h))| -J_m\,.
\end{equation*}
Upon the dilation \eqref{eq:241} we use Lemma
\ref{lem:preliminary-lemmas-half-Dirichlet} to obtain that, for
$i=1,2$, \break $\pa B
(\hat{\Lambda}^i(h),\hat r_i(h) h^{k_i}) \subset \rho(\tilde{\A}_{k,h})$\,, and
that for some $C>0$ 
\begin{equation}    
\label{eq:200} 
\sup_{\lambda\in\partial B(\hat{\Lambda}^i,\,h^{k_i } \hat r_i(h)}\|(\tilde{\A}_{k,h} -\lambda)^{-1}\| 
\leq \frac{C}{\delta_k(h) \, h^{2/3}} \,.
\end{equation}
Note that, in contrast with \cite{AGH},  ${\rm card}\, \Jg_\partial^r(h)$ is not bounded
as $h\to0$. Consequently, $\delta_k(h)$ depends on $h$ through the distance between $b_k(h)$
and  $\Se$. Depending on the potential type there
exist positive constants $h_0, \,c$, and $C$ such that 
\begin{equation}
\label{eq:201}
  \delta_k(h)  \geq C\, d(b_k(h),\Se)^{p_i}\geq c\, h^{p_i \varrho_\perp}\,,\, \forall h\in (0,h_0], \forall k \in \Jg_\partial^r(h)\,,
\end{equation}
where $p_1=2$ for potentials of type $V1$ and $p_2=1$ for potentials of type
$V2$.

\subsubsection{Approximate resolvent norm}
We have shown that there exists $h_0 >0$ such that the approximate
resolvent \eqref{eq:189aa} is well defined when $\lambda \in \pa B (\hat
\Lambda_i(h), h^{k_i} \hat r_i(h))$ for all $h\in (0,h_0]$. To summarize, we
state the following 
\begin{lemma}\label{lemma8.1}
For $i=1,2$, there exists $C$ and $h_0$ such that, for any $h\in (0,h_0]$ and any $\lambda \in 
\pa B (\hat \Lambda_i(h), h^{k_i+ \hat q_i})$,  the approximate resolvent satisfies
\begin{equation}
\| \RR(h,\lambda)\| \leq C \, h^{-k_i -\hat q_i }\,.
\end{equation}
\end{lemma}
The proof follows immediately from \eqref{eq:18}, \eqref{eq:20},
  \eqref{eq:190}, \eqref{eq:191}, \eqref{eq:195}, \eqref{eq:197},
  \eqref{eq:199}, and \eqref{eq:200}.

\subsection{Approximate resolvent error}
In this subsection, we show that $\RR(h,\lambda)$ is a good
approximation of $(\A_h-\lambda)^{-1}$. \\
We proceed as in Paragraph \ref{sss2.3.6}, albeit with the refined
partition of unity defined in Section \ref{sec:refin-part-unity} .
From \eqref{eq:202}) we recall that
$$
 \Eg(h,\lambda)= \sum_{j\in \mathcal J_i(h)}\B_j (h,\lambda)\, \chi_{j,h}  
+  \sum_{k\in \mathcal J_\partial (h)}\B_k(h,\lambda) \, \eta_{k,h} \,,
$$
and keep the same definition for $\B_j(h,\lambda)$ and $\B_k(h,\lambda)$ as in
\eqref{eq:203}.

For the present partition of unity, we set, as in Section \ref{s2}, 
\begin{itemize}
\item
$\Supp \widehat \chi_{j,h}\subset B(a_j(h),2 h^{\varrho})$ for $j\in \mathcal J_i(h)\,$,
\item
$\Supp \widehat \eta_{k,h}\subset B(b_k(h),2h^{\varrho_\perp})$ for $k\in \Jg_\partial^\perp\,$, 
\item
$\Supp \widehat \eta_{k,h}\subset B(b_k (h),2h^{\varrho})$ for $k\in \Jg_\partial^N\,$, 
\item  
$\widehat \chi_{j,h} \chi_{j,h}  =\chi_{j,h}$ and $\widehat \eta_{k,h}
  \eta_{k,h}=\eta_{k,h}\,$,
\end{itemize}
and
\begin{displaymath}
  \widehat {\A}_{k,h}=T_{\mathcal{F}_{b_k}} \A_h T_{\mathcal{F}_{b_k}}^{-1}\,.
\end{displaymath}

In the sequel we  prove the following generalization of  \cite[Lemma
7.6]{AGH} (see also paragraph \ref{sss2.3.6}).
\begin{proposition}\label{lemma8.2}
Let $\hat q_i$ be defined by  \eqref{defhatq} with  $0< q < \frac 16$, and 
 $(\varrho_\perp\,, \,\varrho)$ satisfy \eqref{eq:189}. Let further $\tilde{b}$ and $b$
  respectively satisfy  (\ref{eq:244}) and
\begin{equation}
\label{eq:210}
  1-\frac{3\varrho_\perp}{2} <b<3/4\,.
\end{equation}
Then, under the assumptions of either
  Theorem \ref{thm:interior} (V1 potentials) or Theorem \ref{thm:nonsmooth}
  (V2 potentials),we have
 \begin{equation}
\label{eq:204} 
    \lim_{h\to0} \sup_{\lambda\in\partial B(\hat{\Lambda}^i(h) ,h^{k_i+\hat q_i})} \| \Eg(h,\lambda)\| =0\,.
\end{equation}
\end{proposition}
Note that  by \eqref{eq:189}, we have $\varrho_\perp >\frac 16$ which implies
that the interval  $( 1-\frac{3\varrho_\perp}{2} ,3/4) $ is not empty, though
\eqref{eq:210} is certainly more restrictive than
(\ref{eq:244}).

Keeping \eqref{eq:205} in mind, \eqref{eq:204} follows, under the
assumptions of Proposition \ref{lemma8.2}, from the following lemma
\begin{lemma} 
Under the assumptions of Proposition \ref{lemma8.2}, 
there exist $C >0$, $h_0 >0$ and $d_0 >0$ such that, for  all $h\in
(0,h_0]$ and   $\lambda\in\partial
B(\hat{\Lambda}^i(h) ,h^{k_i+\hat q_i})$, we have  
\begin{equation}\label{estlemma}
\sup_{j\in \Jg_i(h)}\| \mathcal B_j(h,\lambda) \| + \sup_{k\in \Jg_\partial(h)}\| \| \mathcal B_k(h,\lambda) \| \leq C h^{d_0}\,.
\end{equation}
\end{lemma}

\begin{proof} \strut \\
  We split the proof into several steps, estimating the $\| \mathcal
B_j(h,\lambda) \| $ or $\| \mathcal B_k(h,\lambda) \|$ in the various cases
listed in Subsection \ref{sec:local-resolv-estim}. An explicit formula
of $d_0$ is provided at the end of the proof. Throughout the proof
of the lemma, all the constants $C$ and $h_0$ appearing at each step
can be chosen independently of $h\in (0,h_0]$, $j\in \Jg_i(h)$, $k\in
\Jg_\pa (h)$ and $\lambda \in \pa B (\hat \Lambda_i(h), h^{ k_i + \hat q_i})$.

{\bf {\em Step 1:} Estimate $\| \mathcal B_j\|$ for $j\in\Jg_i$ and
  $\|\mathcal B_k \|$ for $k\in\Jg_\partial^{s,0}\cup (\Jg_\partial^c\setminus\Jg_\partial^s)$.
  \\[1ex]} We refer the reader to \cite[(7.50) -(7.51)]{AGH} where it
is shown that there exist $C>0\,$ and $h_0 >0$, such that, for all
$h\in (0,h_0]$, $j\in \Jg_i(h)$, $k\in \Jg_\pa^N(h) \cup \Jg_\partial^{s,0} $,
and $\lambda\in\partial B(\hat{\Lambda}^i(h) ,h^{k_i+\hat q_i})$,
\begin{equation}   
\label{eq:206}
 \|\mathcal B_j(h,\lambda)\|+ \|\mathcal B_k(h,\lambda) \|   \leq C \, (h^{2/3-\varrho}+ 
h^{2(\varrho-1/3)})\,.
\end{equation}

For $k\in \Jg_\pa^c \setminus \Jg_\pa^s$ we use \eqref{eq:243} (which remains
valid under the assumptions set above on $\lambda$) to obtain
\begin{equation}
\label{eq:208}
  \|\mathcal B_k(h,\lambda)\|  \leq C \, (h^{2/3-\varrho}+
h^{\varrho}) \,.
\end{equation}
Hence \eqref{estlemma} is satisfied for $k\in \Jg_\pa^c \setminus \Jg_\pa^s$ if $d_0$ satisfies
\begin{equation}\label{eq:conditiond01}
0 < d_0 \leq d_1:= \inf (\varrho\,, 2/3 -\varrho\,, 2(\varrho -1/3))\,.
\end{equation}
Note that by \eqref{eq:189} $d_1$  is positive.
\\[3ex]

{\bf {\em Step 2:} Estimate $ \big\|\eta_{k,h}T_{\mathcal{F}_{b_k}}^{-1}
  (\widehat{\A}_{k,h}-\tilde{\A}_{k,h}) (\tilde{\A}_{k,h}-\lambda)^{-1}
  T_{\mathcal{F}_{b_k}}\hat{\eta}_{k,h}\big\|$ for
  $k\in\Jg_\partial^s\setminus\Jg_\partial^{s,0}$ and $V1$ potentials. \\[2ex] } Note that
the above term is the first one on the right-hand-side of
(\ref{eq:203}b), and we must provide an effective bound for it in
order to obtain a proper estimate of $\|\B_k(h,\lambda)\|$.  We recall that
$\widehat \A_{k,h}$ has been introduced in \eqref{defwidehatA} and
that, for type V1 potentials, $\tilde{\A}_{k,h}$ is introduced in
\eqref{eq:192}.  Let $(a_0,\tilde b)$ satisfy
\begin{equation}\label{eq:conda0}
0 < a_0<2\tilde{b}/3\,.
\end{equation}
Decomposing $(\widehat{\A}_{k,h}-\tilde{\A}_{k,h}) $,
we write (cf. also \cite[(7.52)-(7.55)]{AGH}) 
\begin{equation}    
\label{eq:209}
\begin{array}{l}
  \big\|\eta_{k,h}T_{\mathcal{F}_{b_k}}^{-1}(\widehat{\A}_{k,h}-\tilde{\A}_{k,h})  
  (\tilde{\A}_{k,h}-\lambda)^{-1}T_{\mathcal{F}_{b_k}}\hat{\eta}_{k,h}\big\| \\\qquad 
  \leq C \Big(\|\tilde{\eta}_{k,h}\delta_1{\mathbf 1}_{\rho\leq h^{2/3-a_0}}h^2\partial^2_s(\tilde\A_{k,h}-\lambda)^{-1} \breve{\eta}_{k,h}\| 
\\
\qquad \quad \qquad  +\|\tilde{\eta}_{k,h}\delta_1{\mathbf
   1}_{\rho>h^{2/3-a_0}}h^2\partial^2_s(\tilde\A_{k,h}-\lambda)^{-1}\breve{\eta}_{k,h}\|  
\\ \qquad \quad \qquad \quad 
 +  \|\tilde{\eta}_{k,h}\delta_2h^2\partial_\rho(\tilde\A_{k,h}-\lambda(h))^{-1}\breve{\eta}_{k,h} \| \\\quad  \quad \qquad\qquad \quad + \|\tilde{\eta}_{k,h}\delta_3h^2\partial_s(\tilde\A_{k,h}-\lambda(h))^{-1}\breve{\eta}_{k,h}\| 
 \\ \qquad \quad \qquad  \quad \quad  \quad  +
 \| {\mathbf 1}_{\rho\leq h^{2/3-a_0}}(V\circ \mathcal F_{b_k}
 -V_{b_k}^{(2)}\,)\tilde{\eta}_{k,h}
 (\tilde\A_{k,h}-\lambda)^{-1}\breve{\eta}_{k,h}\| \\
\quad  \qquad \quad \quad \qquad \qquad  + \| {\mathbf 1}_{\rho>h^{2/3-a_0}}(V\circ \mathcal F_{b_k}
 -V_{b_k}^{(2)}\,)\tilde{\eta}_{k,h}
 (\tilde\A_{k,h}-\lambda)^{-1}\breve{\eta}_{k,h}\| \Big) \,, 
 \end{array}
 \end{equation}
where 
\begin{displaymath}
  \delta_1=\tilde g^{-2}-1-2\kappa(b_k)\rho \chi(h^{-2(1-\tilde{b})/3}\rho)\,,\,
\delta_2=\kappa(s)\tilde g^{-1} -\kappa(b_k) \,,\, \delta_3= \rho\kappa^\prime(s) \tilde g^{-3}\,,
\end{displaymath}
\begin{displaymath} 
  \breve{\eta}_{k,h} = T_{\mathcal{F}_{b_k}}\hat{\eta}_{k,h}\,,\quad \tilde \eta_{k,h} = T_{\mathcal{F}_{b_k}} \eta_k\,,
\end{displaymath}
and where  $V_{b_k}^{(2)}$ is introduced in \eqref{defVbk}, 

For the first term on the right-hand-side of \eqref{eq:209},
observing that  $$ \delta_1 =\OO(\rho^2) + \OO(\rho s) + \OO (\rho) (1-\chi ( h^{-2(1-\tilde{b})/3}\rho ))\,,
$$
we obtain with the aid  of \eqref{eq:conda0},
\begin{displaymath}
   \|\delta_1\tilde{\eta}_{k,h}{\mathbf 1}_{\rho\leq h^{2/3-a_0}}\|_\infty\leq
   C\, h^{2/3+\varrho_\perp-a_0}\,.
\end{displaymath}
Using the dilation \eqref{eq:193}  together
with (\ref{eq:149}a) (with $3a_0/2<a'<1$) we obtain
  \begin{equation}  
\label{eq:211} 
h^2\|\tilde{\eta}_{k,h}\delta_1{\mathbf 1}_{\rho\leq
  h^{2/3-a_0}}\partial^2_s(\tilde\A_{k,h}-\lambda)^{-1} \breve{\eta}_{k,h} \| \leq
C\, h^{\varrho_\perp-a_0-\hat q_1} \,.
\end{equation}
To confirm the effectiveness of the above bound we need to restrict
$a_0$ further by imposing
$$ 0<a_0<\varrho_\perp-\hat q_1\,, $$
which is possible by
\eqref{eq:condhatq} and (\ref{eq:189}a).\\
 As 
\begin{displaymath}
    \|\delta_1\tilde{\eta}_{k,h}\|_\infty \leq h^{\varrho_\perp} \,,
\end{displaymath}
the second term can be estimated using the above dilation together
with (\ref{eq:149}b) for $ a\leq\frac 32 a_0$ and $a'\geq1 - \frac 32
\varrho_\perp$, (thus satisfying $0 < a < a'<1$ as is required in Proposition
\ref{prop:effect-2}) to obtain
  \begin{equation}  
\label{eq:212}
h^2\|\tilde{\eta}_{k,h}\delta_1{\mathbf
  1}_{\rho>h^{2/3-a_0}}\partial^2_s(\tilde\A_{k,h}-\lambda)^{-1} \breve{\eta}_{k,h} \|
\leq Ch^{\varrho_\perp} \,.
\end{equation}
Since
\begin{displaymath}
  \| \tilde{\eta}_{k,h}\delta_2\|_\infty+  \| \tilde{\eta}_{k,h}\delta_3\|_\infty \leq Ch^{\varrho_\perp}\,,
\end{displaymath}
the third term and the fourth terms on the right-hand-side of
\eqref{eq:209}, can be bounded via  the same dilation and
(\ref{eq:149}a) (with $a'\geq1 - \frac 32 \varrho_\perp$),
yielding
\begin{equation}  
\label{eq:213} 
h^2[\|\tilde{\eta}_{k,h}\delta_2\partial_\rho(\tilde\A_{k,h}-\lambda(h))^{-1}\breve{\eta}_{k,h}
\|+\|\tilde{\eta}_{k,h}\delta_3\partial_s(\tilde\A_{k,h}-\lambda(h))^{-1}
  \breve{\eta}_{k,h}\|]\leq Ch^{\varrho_\perp-\hat{q}_1}\,. 
\end{equation}
Note that $\varrho_\perp-\hat{q}_1>0$ by (\ref{eq:189}a) and \eqref{defhatq}.
To control the fifth term on the right-hand-side of \eqref{eq:209} we
use \eqref{esttildeV} (with $x_0$ replaced by $b_k$) to obtain
\begin{displaymath}
  \| {\mathbf 1}_{\rho\leq h^{2/3-a_0}}(V\circ \mathcal F_{b_k}
  -V_{b_k}^{(2)})\,\tilde{\eta}_{k,h}\|_\infty \leq Ch^{4/3+\varrho_\perp -2a_0}\,,
\end{displaymath}
and, with the aid of the same dilation and (\ref{eq:149}a),
\begin{equation}
  \label{eq:214}
\| {\mathbf 1}_{\rho\leq h^{2/3-a_0}}(V\circ \mathcal F_{b_k}
 -V_{b_k}^{(2)}\,)\tilde{\eta}_{k,h}
 (\tilde\A_{k,h}-\lambda)^{-1}\breve{\eta}_{k,h}\| \leq Ch^{\varrho_\perp-\hat{q}_1-2a_0}\,.
\end{equation}
For \eqref{eq:214} to be an effective bound we must further restrict
$a_0$ so that $$0<a_0<(\varrho_\perp-\hat{q}_1)/2\,.$$
Finally, to bound the last term on the right-hand-side of
\eqref{eq:209} we have, by (\ref{eq:149}b) and the same dilation
\begin{equation}   
\label{eq:215}
   \| {\mathbf 1}_{\rho>h^{2/3-a_0}}(V\circ \mathcal F_{b_k}
 -V_{b_k}^{(2)}\,)\tilde{\eta}_{k,h}
 (\tilde\A_{k,h}-\lambda)^{-1}\breve{\eta}_{k,h}\|\leq Ch^{ 3\varrho_\perp-2/3+a_0}\,,
\end{equation}
which leads to a further restriction on $a_0$
$$
3\varrho_\perp-2/3+a_0 >0\,.
$$
Altogether $a_0$ should satisfy:
\begin{equation}\label{eq:8.48}
-3 \Big(\varrho_\perp  -\frac 29\Big)  < a_0 < \frac{\varrho_\perp-\hat{q}_1}{2}\,.
\end{equation}
Such $a_0$ exists if and only if  
\begin{equation*}
 - 3 \Big(\varrho_\perp  -\frac 29\Big)  < \frac{\varrho_\perp-\hat{q}_1}{2}\,,
\end{equation*}
or, equivalently, when 
\begin{equation*}
\frac 43 + \hat q_1< 7 \varrho_\perp\,.
\end{equation*}
Note that the above condition is satisfied by the upper bound of $\hat
q_1$ in \eqref{eq:condhatq} and the lower bound of $\varrho_\perp$ in
(\ref{eq:189}a).  

In view of \eqref{eq:8.48} we set 
\begin{equation}\label{eq:216} 
  a_0=\frac{1}{3}-\frac{5\varrho_\perp}{4}-\frac{\hat{q}_1}{4}>\frac{1}{36}\,,
\end{equation}
(wherein the lower bound follows from (\ref{eq:189}a) and \eqref{eq:condhatq}).\\
By \eqref{eq:conda0}, $\tilde b\in(0,\frac12 -q )$ must satisfy 
$$
\frac{1}{2}-\frac{15\varrho_\perp}{8}-\frac{3\hat{q}_1}{8} < \tilde{b}<\frac 12 - \frac{3}{2}\hat q_1\,.
$$
For the above inequality to make any sense we must have
$
\hat q_1 < \frac 53 \varrho_\perp\,,
$
which clearly holds by \eqref{eq:condhatq} and  (\ref{eq:189}a). 

By   \eqref{eq:211}-\eqref{eq:216}, there exists $C>0$ such that for
sfficiently small $h$
\begin{displaymath}
  \big\|\eta_{k,h}T_{\mathcal{F}_{b_k}}^{-1}
  (\widehat{\A}_{k,h}-\tilde{\A}_{k,h}) (\tilde{\A}_{k,h}-\lambda)^{-1}
  T_{\mathcal{F}_{b_k}}\hat{\eta}_{k,h}\big\|\leq Ch^{d_0}\,,
\end{displaymath} 
for any 
\begin{equation}
\label{eq:conditiond02}
0 < d_0 \leq  d_2:= \inf ( \varrho_\perp-\hat{q}_1-2a_0,   3\varrho_\perp-2/3+a_0 )= 7\varrho_\perp/4-1/3-\hat{q}_1/4\,.
\end{equation}
Note that the positivity of $d_2$ has been verified above.\\

{\bf {\em Step 3:} Estimate $  \big\|\eta_{k,h}T_{\mathcal{F}_{b_k}}^{-1}
  (\widehat{\A}_{k,h}-\tilde{\A}_{k,h}) (\tilde{\A}_{k,h}-\lambda)^{-1}
  T_{\mathcal{F}_{b_k}}\hat{\eta}_{k,h}\big\|$  for $k\in\Jg_\partial^s\setminus\Jg_\partial^{s,0}$ and $V2$
  potentials. \\[2ex] }
For potentials of type V2, we recall that $\tilde{\A}_{k,h}$ is
  introduced in \eqref{deftildeAkhV2} and write 
\begin{equation}    
\label{eq:218}
\begin{array}{l}
  \big\|\eta_{k,h}T_{\mathcal{F}_{b_k}}^{-1}
  (\widehat{\A}_{k,h}-\tilde{\A}_{k,h}) (\tilde{\A}_{k,h}-\lambda)^{-1}
  T_{\mathcal{F}_{b_k}}\hat{\eta}_{k,h}\big\|\\\qquad 
  \leq C h^2\big( \| {\mathbf 1}_{\rho>h^{2/3-a_0}}\tilde{\eta}_{k,h}(\tilde g_c-1)
  \partial^2_s(\tilde\A_{k,h}-\lambda)^{-1}\breve{\eta}_{k,h}\| \\\qquad \quad  \quad \qquad +  \|{\mathbf
    1}_{\rho>h^{2/3-a_0}}\tilde{\eta}_{k,h}(\tilde g_c-1-\alpha s)
  \partial^2_\rho(\tilde\A_{k,h}-\lambda)^{-1}\breve{\eta}_{k,h}\|\\
 \qquad \quad \qquad \qquad   +\| {\mathbf
    1}_{\rho\leq h^{2/3-a_0}}\tilde{\eta}_{k,h}(\tilde g_c-1)
  \partial^2_s(\tilde\A_{k,h}-\lambda)^{-1}\breve{\eta}_{k,h}\| \\ \qquad \qquad \quad \qquad  \quad +  \|{\mathbf
    1}_{\rho\leq h^{2/3-a_0}}\tilde{\eta}_{k,h}(\tilde g_c-1-\alpha s)
  \partial^2_\rho(\tilde\A_{k,h}-\lambda)^{-1}\breve{\eta}_{k,h}\| \big)\,.
  \end{array}
\end{equation}
For the first two terms we use the dilation \eqref{eq:196} together with (\ref{eq:183}b) (for $a= \frac 94 a_0$)  to
obtain 
\begin{multline}
  \label{eq:219}
 h^2\big(\| {\mathbf 1}_{\rho>h^{2/3-a_0}}\tilde{\eta}_{k,h}(\tilde g_c-1)
  \partial^2_s (\tilde\A_{k,h}-\lambda)^{-1}\breve{\eta}_{k,h}\| \\ + \|{\mathbf
    1}_{\rho>h^{2/3-a_0}}\tilde{\eta}_{k,h}(\tilde g_c-1-\alpha s)
  \partial^2_\rho(\tilde\A_{k,h}-\lambda)^{-1}\breve{\eta}_{k,h}\|\big) \leq Ch^{\varrho_\perp}\,.
\end{multline}
For the third term we use the same dilation and (\ref{eq:182}b) to
obtain for  $\mathfrak a$ satisfying  $1/4<\mathfrak a=5/12<(1-q)/2$ that
\begin{equation}
  \label{eq:220} 
  h^2\| {\mathbf 1}_{\rho\leq h^{2/3-a_0}}\tilde{\eta}_{k,h}(\tilde g_c-1)
  \partial^2_s (\tilde\A_{k,h}-\lambda)^{-1}\breve{\eta}_{k,h}\| \leq Ch^{\varrho_\perp-7/27}\,.
\end{equation}
Recall that $\varrho_\perp > \frac{7}{27}$ by (\ref{eq:189}b).

Finally, for the last term on the right-hand-side of \eqref{eq:218} we
use the fact that
\begin{displaymath}
  \|{\mathbf  1}_{\rho\leq h^{2/3-a_0}}\tilde{\eta}_{k,h}(\tilde g_c-1-\alpha
  s)\|_\infty\leq C(h^{2\varrho_\perp}+h^{2/3-a_0})\,,
\end{displaymath}
to obtain by the same dilation and (\ref{eq:182}a) 
\begin{equation}
  \label{eq:221}
h^2 \|{\mathbf
    1}_{\rho\leq h^{2/3-a_0}}\tilde{\eta}_{k,h}(\tilde g_c-1-\alpha s)
  \partial^2_\rho(\tilde\A_{k,h}-\lambda)^{-1}\breve{\eta}_{k,h}\| \leq
  C\, (h^{2/9-a_0-\hat{q}_2}+h^{2\varrho_\perp-4/9-\hat{q}_2})\,.
\end{equation}
As $2\varrho_\perp-4/9-\hat q_2 >0$ by \eqref{eq:condhatq} and (\ref{eq:189}b)
we need to select  $a_0 < \frac 29 - \hat q_2$ (yielding $a_0  < \frac {4}{27}$).
 Setting 
\begin{equation}
\label{eq:222} 
  a_0=2(1/3-\varrho_\perp)\,,
\end{equation}
yields by \eqref{eq:218}-\eqref{eq:221}
that \eqref{estlemma} is satisfied in the present context if
\begin{equation}\label{eq:conditiond03}
0 < d_0 \leq d_3:= 2 \varrho_\perp-4 /9-\hat{q}_2 \,,
\end{equation}
where the positivity of $d_3$ follows from \eqref{eq:condhatq} and (\ref{eq:189}b).
\\[3ex]

{\bf {\em Step 4:} Estimate $\| [\A_h,\eta_{k,h}]R_{k,h}\check \eta_{k,h}
  \|$  for $k\in   \Jg_\partial^s\setminus\Jg_\partial^{s,0} $ and V1 potentials.} \\[2ex]
Note that the above term is the second term on the right-hand side of
(\ref{eq:203}b), and an effective bound for it in is necessary order
to obtain a proper estimate of $\|\B_k(h,\lambda)\|$.  . Having in mind
\eqref{eq:formcomm}, we decompose $[\A_h,\eta_{k,h}]$ in the form
  \begin{equation}
\label{eq:223}
    [\A_h,\eta_{k,h}] = - h^2 (\Delta \eta_{k,h}) - 2 h^2  1_{\rho>h^{2/3-a_0}}
    \nabla \eta_{k,h} \cdot \nabla - 2 h^2   1_{\rho < h^{2/3-a_0}} \nabla \eta_{k,h} \cdot \nabla
  \,,
  \end{equation}
where $a_0$ is given by \eqref{eq:216}.  Note by (\ref{eq:185}b) that
  $(\Delta \eta_{k,h})$ and $\nabla \eta_{k,h}$ are supported in
  $B(b_k,h^{\varrho_\perp})\setminus B(b_k,h^{\varrho_\perp}/2)$.  Hence, whenever
  $\rho<h^{2/3-a_0}$ we have, for sufficiently small $h$,
\begin{displaymath}
|s|\geq \frac 12 h^{\varrho_\perp} - h^{2/3-a_0} > \frac13 h^{\rho_\perp}\,,
\end{displaymath}
since by \eqref{eq:216} we have 
$
\varrho_\perp < \frac 23 - a_0\,.
$
\\
Consequently, we may represent \eqref{eq:223} in the form
 $$
    [\A_h,\eta_{k,h}] = - h^2 (\Delta \eta_{k,h}) - 2 h^2  1_{\rho>h^{2/3-a_0}}
    \nabla \eta_{k,h} \cdot \nabla - 2 h^2   1_{\rho < h^{2/3-a_0}}
    1_{|s|>\frac{h^{\rho_\perp}}{3}} \nabla \eta_{k,h} \cdot \nabla 
  \,.
 $$
Recall from (\ref{eq:185}b) that
  for all $k\in   \Jg_\partial^s\setminus\Jg_\partial^{s,0} \,,$  
$$
|\nabla \eta_{k,h}| + h^{\varrho_\perp}|D^2\eta_{k,h}|\leq C\, h^{-\varrho_\perp}\,.
$$
We finally note that by \eqref{eq:186} and \eqref{eq:188}, we have,
for $\rho<h^{2/3-a_0}$, that
$$
 |\partial \tilde \eta_{k,h}/\partial\rho|\leq C \, h^{2/3-a_0-2\varrho_\perp}\,.
$$ 
 With these remarks in mind  we obtain the existence of $C$ and $h_0$
 such that for $h\in (0,h_0]$ and  $\lambda\in \rho ( \tilde A_{k,h})$
\begin{equation}
\label{eq:217}
\begin{array}{l}
  \big\|[\A_h,\eta_{k,h}]R_{k,h} \check\eta_{k,h}  \big\| \\[1.2ex]
  \quad \quad  \leq C \, h^{2(1-\varrho_\perp)} \|\breve \eta_{k,h} (\tilde\A_{k,h}-\lambda)^{-1} \breve \eta_{k,h}\|
    \\  \quad \quad \quad  + C \, h^{2-\varrho_\perp} \|\breve \eta_{k,h} {\mathbf
     1_{\rho>h^{2/3-a_0}}}\pa_\rho (\tilde\A_{k,h}-\lambda)^{-1}\breve \eta_{k,h} \|
     \\   \quad \quad \quad  + C \, h^{2-\varrho_\perp} \|\breve \eta_{k,h} {\mathbf
     1_{\rho>h^{2/3-a_0}}}\pa_s (\tilde\A_{k,h}-\lambda)^{-1}\breve \eta_{k,h} \|\\
       \quad \quad \quad + C \,  h^{2-\varrho_\perp} \,\|\breve \eta_{k,h} {\mathbf
     1_{s>\frac{h^{\rho_\perp}}{3}}}\partial_s(\tilde\A_{k,h}-\lambda)^{-1}\breve \eta_{k,h} \| \\ \qquad \quad  + C \,  h^{8/3-a_0-2\varrho_\perp}\|\breve \eta_{k,h} \partial_\rho(\tilde\A_{k,h}-\lambda)^{-1}\breve \eta_{k,h} \| \,,
     \end{array}
     \end{equation}
where $\tilde\A_{k,h}$ is given by \eqref{eq:192}.\\ 
We now estimate term by term the 
right hand side of \eqref{eq:217} using  the dilation
\eqref{eq:193}. To this end we use \eqref{eq:149} below with $a= \frac 32 a_0$ and $a' \geq
 1-\frac32 \varrho_\perp$.  \\
For the first term, we get from   (\ref{eq:149}a)
$$
 h^{2(1-\varrho_\perp)} \|\breve \eta_{k,h} (\tilde\A_{k,h}-\lambda)^{-1} \breve \eta_{k,h}  \| \leq  C h^{\frac23 -2\varrho_\perp - \hat q_1 }\,.
 $$
 We note that by  \eqref{eq:condhatq} and (\ref{eq:189}a)  it follows
 that $\frac23 -2\varrho_\perp - \hat q_1 >0$.\\ 
 For the second term, we use  (\ref{eq:149}b) to obtain
 $$
  h^{2- \varrho_\perp} \| \breve \eta_{k,h}   {\mathbf
     1_{\rho>h^{2/3-a_0}}\pa_\rho}(\tilde\A_{k,h}-\lambda)^{-1}\| \leq C
   h^{\frac23 -  \varrho_\perp + \frac{a_0}{2}} \,.
 $$
 For the third term, using  (\ref{eq:149}b) yields
 $$
   h^{2- \varrho_\perp}  \|\breve \eta_{k,h}   {\mathbf
     1_{\rho>h^{2/3-a_0}}\pa_s}(\tilde\A_{k,h}-\lambda)^{-1}\| \leq C  h^{\frac
     23  -  \varrho_\perp } \,.
 $$
 For the forth term, we use (\ref{eq:149}c) whose assumptions hold  when $a' \geq 1-\frac32
 \varrho_\perp$ and $\mathfrak a$ satisfies
     \begin{equation}
\label{eq:226}
   \frac 16 < \mathfrak a < \frac14 \mbox{ and }   \frac 32 \varrho_\perp - \frac12 + \mathfrak a < 0\,.       
     \end{equation}
     Such a choice for $\mathfrak a$  is possible since $\varrho_\perp < \frac 29 $ is satisfied by (\ref{eq:189}a).
We obtain
$$
 \|\breve \eta_{k,h} {\mathbf
     1_{s>\frac{h^{\rho_\perp}}{3}}}\partial_s(\tilde\A_{k,h}-\lambda)^{-1}\breve \eta_{k,h} \| \leq C  h^{\frac{4}{3}{\mathfrak a}-\varrho_\perp}\,.
$$
 Note by (\ref{eq:189}a) that  $$4{\mathfrak a}/3 -\varrho_\perp >2/9 -\varrho_\perp > 0\,.$$
Finally, for the fifth term, we get from (\ref{eq:149}a)
     $$
      h^{8/3-a_0-2\varrho_\perp}\|\breve \eta_{k,h} \partial_\rho(\tilde\A_{k,h}-\lambda)^{-1}\breve \eta_{k,h} \| \leq C\,  h^{\frac 23 - a_0 - 2 \varrho_\perp - \hat q_1}\,.
     $$
     Using \eqref{eq:216} together with \eqref{eq:condhatq} and
     (\ref{eq:189}a) then yields
     $$
     \frac 23 - a_0 - 2 \varrho_\perp - \hat q_1= \frac 13 - \frac 34 (\hat q_1 + \varrho_\perp) >0 \,.
     $$
In conclusion, we have established, for potentials of type V1 that
$ [\A_h,\eta_{k,h}]R_{k,h}\check \eta_{k,h} $ satisfies \eqref{estlemma} when $d_0$ satisfies 
\begin{equation}\label{eq:conditiond041}
  0<d_0 \leq d_4^1:= \min \Big\{ \frac23 -2\varrho_\perp - \hat q_1\,
  ,\frac{4}{3}{\mathfrak a} -\varrho_\perp\, , \frac 13 - \frac 34 (\hat q_1 +
  \varrho_\perp)\Big\}\,,
\end{equation}
for  ${\mathfrak a}$ satisfying \eqref{eq:226}.\\
The positivity of $d_4^1$ is established by the foregoing
discussion. \\

{\bf {\em Step 5:} Estimate $\| [\A_h,\eta_{k,h}]R_{k,h}\check \eta_{k,h}
  \|$  for $k\in   \Jg_\partial^s\setminus\Jg_\partial^{s,0} $ for V2 potentials. \\[2ex]  }
We begin the estimate, as in the V1 case, starting from \eqref{eq:217} but with
$\tilde A_{k,h}$ introduced in \eqref{deftildeAkhV2}, $a_0$
given by \eqref{eq:222}. We use here the dilation
\eqref{eq:196}.  For the first term on the right-hand-side of
\eqref{eq:217}, we use \eqref{eq:180} to obtain
\begin{displaymath}
  h^{2(1-\varrho_\perp)} \|\breve \eta_{k,h} (\tilde\A_{k,h}-\lambda)^{-1} \breve
  \eta_{k,h}\|\leq Ch^{8/9-2\varrho_\perp-\hat{q}_2}\,.
\end{displaymath}
Note that by (\ref{eq:189}b) and \eqref{eq:condhatq} we have
$8/9-2\varrho_\perp-\hat{q}_2>4/27$.

Upon dilation we use (\ref{eq:183}a) with $a=9a_0/4\,$,
for the second
and third terms  on the right-hand-side of \eqref{eq:217} to obtain 
\begin{displaymath}
  h^{2-\varrho_\perp} \big(\|\breve \eta_{k,h} {\mathbf
     1_{\rho>h^{2/3-a_0}}}\pa_\rho (\tilde\A_{k,h}-\lambda)^{-1}\breve \eta_{k,h} \|
     + \|\breve \eta_{k,h} {\mathbf
     1_{\rho>h^{2/3-a_0}}}\pa_s (\tilde\A_{k,h}-\lambda)^{-1}\breve \eta_{k,h}
   \|\big)\leq h^{2/3-\varrho_\perp+a_0/2}\,.
\end{displaymath}
The fourth term on the right-hand-side of \eqref{eq:217} can be
estimated, upon dilation, with the aid of \eqref{eq:181}, where
$1/4 < {\mathfrak a}< (1-q)/2$. We get
\begin{displaymath}
  h^{2-\varrho_\perp} \,\|\breve \eta_{k,h} {\mathbf
     1_{s>\frac{h^{\rho_\perp}}{3}}}\partial_s(\tilde\A_{k,h}-\lambda)^{-1}\breve
   \eta_{k,h} \| \leq h^{2/9-\varrho_\perp+{\frac49 {\mathfrak a}}}\,.
\end{displaymath}
To have $s>h^{\rho_\perp}/3\Rightarrow\sigma>\varepsilon^{-{\mathfrak a}}$ we further require
\begin{equation}\label{eq:condaV2}
\frac 94 \varrho_\perp-1 + \mathfrak a <0\,,
\end{equation}
which can be  satisfied  since $\varrho_\perp < \frac 13$ by (\ref{eq:189}b)
that implies, in addition, 
\begin{displaymath}
  2/9-\varrho_\perp+\frac 49 {\mathfrak a} >0\,.
\end{displaymath}
Finally, by (\ref{eq:182}a), applied upon \eqref{eq:196}, we have for
the last term on the right-hand-side of \eqref{eq:217},
\begin{displaymath}
  h^{8/3-a_0-2\varrho_\perp}\|\breve \eta_{k,h}
  \partial_\rho(\tilde\A_{k,h}-\lambda)^{-1}\breve \eta_{k,h} \|\leq
  C\, h^{8/9-2\varrho_\perp-a_0- \hat q_2}\,,
\end{displaymath}
By (\ref{eq:189}b) and  \eqref{eq:222} we have
$$
\frac89-2\varrho_\perp-a_0 - \hat q_2 >\frac29 - \frac {2}{27} >0\,.
$$

In conclusion, we have obtained for potentials of type V2 that 
$[\A_h,\eta_{k,h}]R_{k,h}\check \eta_{k,h} $ satisfies \eqref{estlemma} for 
\begin{equation}\label{eq:conditiond042} 
0 < d_0 \leq d_4^2:=\inf \{ \frac 89-2\varrho_\perp- a_0 -\hat{q}_2,  2/3-\varrho_\perp+a_0/2, \frac 89 -\varrho_\perp+\frac 49 {\mathfrak a}  \}\,,
\end{equation}
with $\mathfrak a \in (1/6, (1-q)/2)$  satisfying \eqref{eq:condaV2}.\\
The positivity of $d_4^2$ has already been established in the
foregoing
discussion.\\

{\bf {\em Step 6:} Estimate of  $\|\eta_{k,h} T_{\mathcal{F}_{b_k}}^{-1}
  (\widehat {\A}_{k,h}-\tilde{\A}_{k,h}) (\tilde{\A}_{k,h} -\lambda)^{-1}
  \breve{\eta}_{k,h}\big\| $ for $k\in\Jg_\partial^r$. \\[2ex]  }
We recall that $\tilde \A_{k,h}$ is
introduced in \eqref{eq:183a} and we first observe that
\begin{equation} \label{eq:Taylor2}
\begin{array}{ll}
  \|(V\circ \mathcal F_{b_k} -V_{b_k}^{(1)})
  \breve{\eta}_{k,h}\|_\infty & \leq C\,h^{2\varrho}\Big( \Big\|\frac{\partial^2V}{\partial
    s^2}\Big\|_{L^\infty(B(b_k,h^\varrho))}\\ &\qquad  \qquad \quad + \Big\|\frac{\partial^2V}{\partial
    s\partial\rho}\Big\|_{L^\infty(B(b_k,h^\varrho))}\\ & 
    \qquad \qquad \qquad +\Big\|\frac{\partial^2V}{
    \partial\rho^2}\Big\|_{L^\infty(B(b_k,h^\varrho))}\Big)\\
   &  \leq \tilde C \, h^{2\varrho}\,,
   \end{array}
\end{equation}
where $$V_{b_k}^{(1)} (\rho,s) := V(b_k)\pm
  \jg_k\rho $$
   as in  (\ref{eq:183a}).\\
\paragraph{V1 Potentials} For potentials of type V1 we now observe that
\begin{displaymath}
  \Big\|\frac{\partial^2V}{\partial s\partial\rho}\Big\|_{L^\infty(B(b_k,h^\varrho))}\leq C\,\left(  d(b_k,\Se) + h^\varrho\right)  \,,
\end{displaymath}
and that
$b_k\not\in\Union_{n\in\Jg_\partial^s\setminus\Jg_\partial^{s,0}}B(b_n,h^{\varrho_\perp})$ for all
$k\in\Jg_\partial^r$ we have
\begin{equation}\label{eq:minoraa}
 d(b_k,\Se) \geq h^{\varrho_\perp}\,.
 \end{equation}
Clearly,
 \begin{displaymath}
  \Big\|\frac{\partial^2V}{\partial s\partial\rho}\Big\|_{L^\infty(B(b_k,h^\varrho))}\leq C\, d(b_k,\Se)  \,.
\end{displaymath}
Furthermore, since $V_{ss}(0,s)=0$,  we have 
\begin{displaymath}
  \Big\|\frac{\partial^2V}{\partial s^2}\Big\|_{L^\infty(B(b_k,h^\varrho))}\leq Ch^\rho\,,
\end{displaymath}
and, for any $a_1>0$,
\begin{displaymath}
  \Big\|{\mathbf 1}_{\rho\leq h^{2/3-a_1}}\frac{\partial^2V}{\partial
    s^2}\Big\|_{L^\infty(B(b_k,h^\varrho))}\leq Ch^{2/3-a_1} \,.
\end{displaymath}
Consequently, using a Taylor expansion of order 2, we obtain from the preceding inequalities
\begin{equation*}
    \|{\mathbf 1}_{\rho\leq h^{2/3-a_1}}(V\circ \mathcal F_{b_k} -V_{b_k}^{(1)})
  \breve{\eta}_{k,h}\|_\infty \leq C\left( h^{2\varrho}d(b_k, \Sg) + h^{\frac 23 -a_1+ 2\varrho} + h^{\frac43 -2 a_1}\right) \,. 
\end{equation*} 
Set $ a_1=1/3-\varrho/2$. As $\varrho >
\frac 29$ and $\varrho+\varrho_\perp<2/3$  by (\ref{eq:189}a) The above inequality
together with \eqref{eq:minoraa},  lead to
\begin{equation}\label{eq:quadraticerror}
    \|{\mathbf 1}_{\rho\leq h^{2/3-a_1}}(V\circ \mathcal F_{b_k} -V_{b_k}^{(1)})
  \breve{\eta}_{k,h}\|_\infty \leq C\, h^{2\varrho}d(b_k, \Sg)  \,. 
\end{equation} 
We now write 
\begin{equation}  
\label{eq:228}
\begin{array}{l}
  \big\|\eta_{k,h} T_{\mathcal{F}_{b_k}}^{-1}
  (\widehat {\A}_{k,h}-\tilde{\A}_{k,h} ) (\tilde{\A}_{k,h} -\lambda)^{-1}
  \breve{\eta}_{k,h}\big\|\\
\qquad \qquad  \leq C\, h^\varrho\|h^2\Delta_{(s,\rho)}(\tilde\A_{k,h}-\lambda)^{-1} \breve{\eta}_{k,h} \| 
\\ \quad \qquad \qquad  + C 
  \|h^2{\nabla_{(s,\rho)}}(\tilde\A_{k,h}-\lambda(h))^{-1}\breve{\eta}_{k,h} \| \\
  \quad \qquad \qquad  + C 
h^{2\varrho}\|{\mathbf
  1}_{\rho\geq h^{2/3-a_1}}\breve{\eta}_{k,h}(\tilde\A_{k,h}-\lambda)^{-1}\breve{\eta}_{k,h}\| \\
   \quad \qquad \qquad  + C  h^{2\varrho} d(b_k,\Sg) \|{\mathbf
  1}_{\rho\leq h^{2/3-a_1}}\breve{\eta}_{k,h}(\tilde\A_{k,h}-\lambda)^{-1}\breve{\eta}_{k,h}\| \,.  
\end{array}
\end{equation}
By \eqref{eq:200} and Lemma
\ref{lem:preliminary-lemmas-half-Dirichlet} we have 
\begin{equation}  \label{eq:230} 
h^\varrho\|h^2 {\Delta_{(s,\rho)}}(\tilde\A_{k,h} -\lambda)^{-1}  \breve{\eta}_{k,h}
\| \leq C\, h^{\varrho-2\varrho_\perp}\,.
\end{equation}
Note that by (\ref{eq:189}a)
$$ \varrho - 2 \varrho_\perp > \frac 13 - \frac 32 \varrho_\perp >0\,.
$$
Similarly, 
\begin{equation}
\label{eq:225}
  \|h^2{\nabla_{(s,\rho)}}(\tilde\A_{k,h} -\lambda)^{-1}\breve \eta_{k,h} \| \leq C\,
  h^{2/3-2\varrho_\perp}\,.
\end{equation}
 Upon the dilation
\eqref{eq:241} we may use \eqref{eq:150} with $a=\frac{a_1}{3 \varrho_\perp}$ to obtain for
the third term 
\begin{equation}
\label{eq:227}
  h^{2\varrho}\|{\mathbf
  1}_{\rho\geq h^{2/3-a_1}}\breve{\eta}_{k,h}(\tilde\A_{k,h}-\lambda)^{-1}\breve{\eta}_{k,h}\| \leq Ch^{2(\varrho-1/3)}\,.
\end{equation}
Note that the cutoff $\rho > h^{2/3-a_1}$ leads after dilation to $\tau >
h^{ -a_1}$ which should be compared with $\tau \geq \delta^{-a}$ for the
application of \eqref{eq:150}.  Hence, we must have
$h^{-a_1} >> \delta^{-a}$ which leads, in view of \eqref{eq:minoraa},  to
$ 2 a \varrho_\perp  -  a_1 <0$.  For the application of \eqref{eq:150} we
must have in addition $d(b_k,\Sg) <
\delta_1$ for some $\delta_1 >0$. If
$d(b_k,\Sg) \geq \delta_1$ \eqref{eq:227} 
follows immediately from Lemma
\ref{lem:preliminary-lemmas-half-Dirichlet}.   

For the last term, we use \eqref{eq:200} to obtain
\begin{equation}  \label{eq:229}
h^{2\varrho} d(b_k,\Sg) \|{\mathbf
  1}_{\rho\leq h^{2/3-a_1}}\breve{\eta}_{k,h}(\tilde\A_{k,h}-\lambda)^{-1}\breve{\eta}_{k,h}\|\leq C d(b_k,\Sg)^{-1} h^{2\varrho-\frac23}
  \leq 
\tilde C\, h^{2\varrho-\varrho_\perp-2/3}\,.
\end{equation}
Note that by (\ref{eq:189}a) $2\varrho-\varrho_\perp-2/3>0$.\\
Substituting the above
into \eqref{eq:228} yields, that $
 \eta_{k,h} T_{\mathcal{F}_{b_k}}^{-1}
  (\widehat {\A}_{k,h}-\tilde{\A}_{k,h}) (\tilde{\A}_{k,h} -\lambda)^{-1}
  \breve{\eta}_{k,h}$ satisfies \eqref{estlemma} with
  \begin{equation}\label{eq:d51}
 0 <  d_0  \leq d_5^1:= \inf\{2\varrho-\varrho_\perp -2/3,  \varrho- 2 \varrho_\perp, 2(\varrho-1/3), 2/3 - 2 \varrho_\perp\}\,, 
\end{equation}
where positivity of $d_5^1$ has been established above.\\ 

\paragraph{ Type V2 potentials} 
For potentials of type $V_2$, we  write, using \eqref{eq:Taylor2} directly, 
\begin{equation}  
\label{eq:232}
\begin{array}{l}
  \big\|\eta_{k,h} T_{\mathcal{F}_{b_k}}^{-1}
  (\widehat {\A}_{k,h}-\tilde{\A}_{k,h} ) (\tilde{\A}_{k,h} -\lambda)^{-1}
  \breve{\eta}_{k,h}\big\|\\
  \qquad 
  \leq C\, \Big(h^\varrho\|h^2\Delta_{(s,\rho)}(\tilde\A_{k,h}-\lambda)^{-1} \breve{\eta}_{k,h} \| \\
  \quad \qquad \qquad \qquad  +
 \|h^2{\nabla_{(s,\rho)}}(\tilde\A_{k,h}-\lambda(h))^{-1}\breve{\eta}_{k,h} \| \\
  \quad \qquad  \qquad \qquad \qquad + 
h^{2\varrho}\|\tilde{\eta}_{k,h}(\tilde\A_{k,h}-\lambda)^{-1}\breve{\eta}_{k,h}\|\Big) \,.
\end{array}
\end{equation}
Proceeding as in the V1 case, we obtain, similarly to the proof of
\eqref{eq:230} and \eqref{eq:225} but with \eqref{eq:201} in mind,
$$
h^\varrho\|h^2\Delta_{(s,\rho)}(\tilde\A_{k,h}-\lambda)^{-1} \breve{\eta}_{k,h} \| \leq C h^{\varrho  -\varrho_\perp}
$$
and
$$
\|h^2\nabla_{(s,\rho)}(\tilde\A_{k,h}-\lambda)^{-1} \breve{\eta}_{k,h} \| \leq C h^{\frac 23 -\varrho_\perp}\,.
$$
For the third term, we may now use \eqref{eq:200} and \eqref{eq:201}, with
$p_2=1$, to obtain that 
\begin{equation*}  
h^{2\varrho}\|\tilde{\eta}_{k,h}(\tilde\A_{k,h}-\lambda)^{-1}\breve{\eta}_{k,h}\|\leq
C\, h^{2\varrho-\varrho_\perp-2/3}\,.
\end{equation*}
Hence, we obtain that $\eta_{k,h} T_{\mathcal{F}_{b_k}}^{-1}
  (\widehat {\A}_{k,h}-\tilde{\A}_{k,h}) (\tilde{\A}_{k,h} -\lambda)^{-1}
  \breve{\eta}_{k,h}$ satisfies \eqref{estlemma} with
  \begin{equation} \label{eq:235}
0<d_0  \leq d_5^2:= \inf\{ 2\varrho-\varrho_\perp-2/3, \varrho  -\varrho_\perp\}\,, 
\end{equation}
where positivity of $d_6^2$ follows from (\ref{eq:189}b).

{\bf {\em Step 7:} Estimate of $[\A_h,\eta_{k,h}]R_{k,h}\check \eta_{k,h}$  for $k\in\Jg_\partial^r$. \\[2ex]  }
We now estimate the rest of the contribution of $\Jg_\partial^r$,   as in
Steps 4 and 5, by writing (dropping the cut-off in the $s$ variable)
\begin{equation}
\label{eq:236}
\begin{array}{l}
 \Big\|[\A_h,\eta_{k,h}]R_{k,h}\check \eta_{k,h}  \Big\| \\
  \quad  \leq C  \,\Big(h^{2(1-\varrho)} \|\breve \eta_{k,h} (\tilde\A_{k,h}-\lambda)^{-1}\breve \eta_{k,h}  \|
    \\   \quad \quad \qquad \quad   +h^{2-\varrho} \|\breve \eta_{k,h} {\mathbf
     1_{\rho>h^{2/3-a_1}}} \partial_\rho(\tilde\A_{k,h}-\lambda)^{-1}\breve \eta_{k,h}  \| \\
      \\  \quad \quad \qquad \quad \qquad \qquad 
     \qquad
      + h^{2-\varrho}  \|\breve \eta_{k,h} \partial_s(\tilde\A_{k,h}-\lambda)^{-1}\breve
      \eta_{k,h}  \| \\ \qquad \qquad 
     \qquad \qquad \quad \qquad \qquad \quad  +
     h^{8/3-a_1-2\varrho}\|\breve \eta_{k,h}
     \partial_\rho(\tilde\A_{k,h}-\lambda)^{-1}\breve \eta_{k,h} \| 
   \Big) \,.  
\end{array}
\end{equation}

\paragraph{V1 potentials}
For potentials of type V1,  we use Lemma
\ref{lem:preliminary-lemmas-half-Dirichlet}  upon the dilation
\eqref{eq:241} to 
obtain for the first term
\begin{equation}
h^{2(1-\varrho)} \|\breve \eta_{k,h} (\tilde\A_{k,h}-\lambda)^{-1}\breve \eta_{k,h}  \| \leq C h^{\frac 43 -2 \varrho -2  \varrho_\perp }\,,
\end{equation}
and observe that $\frac 43 -2 \varrho -2  \varrho_\perp $ is positive by
(\ref{eq:189}a). \\
 Upon dilation we then use (\ref{eq:150}a), with
$a=\frac{a_1}{3 \varrho_\perp}$,  to obtain
for the second term
\begin{equation}
\label{eq:237}
  h^{2-\varrho}\|\breve \eta_{k,h} {\mathbf
     1_{\rho>h^{2/3-a_1}}}\pa_\rho (\tilde\A_{k,h}-\lambda)^{-1}\breve \eta_{k,h}\| \leq Ch^{2/3-\varrho}\,.
\end{equation}
Then we write, using (\ref{eq:150}b) upon dilation, with $\delta$
satisfying \eqref{eq:201},
\begin{equation}
\label{eq:237b}
  h^{2-\varrho}\|\breve \eta_{k,h}  \partial_s(\tilde\A_{k,h}-\lambda)^{-1}\breve
   \eta_{k,h}  \| \leq Ch^{2/3-\varrho-\varrho_\perp}\,.
\end{equation}

Finally, we obtain with the aid of \eqref{eq:200} and Lemma
\ref{lem:preliminary-lemmas-half-Dirichlet}, setting $a_1=2/3-\varrho
-\varrho_\perp$,
\begin{displaymath}
  h^{8/3-a_1-2\varrho}\| \breve \eta_{k,h}  \partial_\rho(\tilde\A_{k,h}-\lambda)^{-1}\breve \eta_{k,h}  \|\leq Ch^{2/3-\varrho-\varrho_\perp}\,.
\end{displaymath}
Hence \eqref{estlemma} holds for $\breve \eta_{k,h} [\A_h,\eta_{k,h}]R_{k,h}\check \eta_{k,h}$   if  $d_0$ satisfies
\begin{equation}\label{eq:conditiond06V1}
0 < d_0 \leq d_6^1:=\frac{2}{3}-\varrho -\varrho_\perp\,,
\end{equation}
where the positivity of $d_6^1\,$ follows from (\ref{eq:189}a).\\

\paragraph{V2 potentials} 
In a similar manner to the V1 case,  we begin by employing
Lemma~\ref{lem:preliminary-lemmas-half-Dirichlet} upon dilation to
obtain 
\begin{equation*}
h^{2(1-\varrho)} \|\breve \eta_{k,h} (\tilde\A_{k,h}-\lambda)^{-1}\breve \eta_{k,h}  \| \leq C \, h^{\frac 43 -2 \varrho - \varrho_\perp }\,.
\end{equation*}
Then by (\ref{eq:150}a), with $a:= \frac{a_1}{2 \varrho_\perp}$, we have
 \begin{displaymath}
 h^{2-\varrho} \, \|\breve \eta_{k,h}  {\mathbf
     1_{\rho>h^{2/3-a_1}}}\pa_\rho (\tilde\A_{k,h}-\lambda)^{-1} \breve \eta_{k,h}\| \leq C\,h^{2/3-\varrho} \,,
 \end{displaymath}
 and  by  (\ref{eq:150}b), 
 \begin{displaymath}
   h^{2-\varrho}\, \|\breve \eta_{k,h}\partial_s(\tilde\A_{k,h}-\lambda)^{-1}\breve \eta_{k,h}\| \leq C\,h^{2/3-\varrho-\varrho_\perp/2}\,.
 \end{displaymath}
 Finally, with the aid of \eqref{eq:200} and Lemma
\ref{lem:preliminary-lemmas-half-Dirichlet}, we obtain by setting
first $a_1=2/3-\varrho -\frac 12 \varrho_\perp$, 
\begin{displaymath} 
  h^{8/3-a_1-2\varrho}\| \breve \eta_{k,h} \partial_\rho(\tilde\A_{k,h}-\lambda)^{-1} \breve
  \eta_{k,h} \|\leq Ch^{2/3-\varrho -  \varrho_\perp/2}\,. 
\end{displaymath}
Hence \eqref{estlemma} holds in the V2  case for the term  $[\A_h,\eta_{k,h}]R_{k,h}\check \eta_{k,h}$    if  $d_0$ satisfies
\begin{equation}\label{eq:conditiond06V2}
0 < d_0 \leq d_6^2 := \frac{2}{3}-\varrho-\frac{\varrho_\perp}{2}\,,
\end{equation}
where the positivity of $d_6^2 $ results of (\ref{eq:189}b).\\

{\bf {\em Conclusion:} Computation of $d_0$. \\[2ex]}
Combining \eqref{eq:conditiond01}, \eqref{eq:conditiond02},
\eqref{eq:conditiond03}, \eqref{eq:conditiond041}, \eqref{eq:conditiond042}, \eqref{eq:d51}, \eqref{eq:235},
\eqref{eq:conditiond06V1},   and \eqref{eq:conditiond06V2},  we have
established \eqref{estlemma} with 
$$
d_0= \inf\{d_1, d_2, d_3, d_4^1, d_4^2,  d_5^1, d_5^2, d_6^1,d_6^2\}\,.
$$
This completes the proof of the lemma.
\end{proof}

\subsection{Eigenvalue existence}
From \eqref{eq:204} it follows that $(I+\Eg(h,\lambda))^{-1}$ is uniformly
bounded as $h\to0$.  Since by Lemma \ref{lemma8.1} we have  $\|
\RR(h,\lambda)\| \leq C\, h^{-(k_i+\hat{q}_i)}$, we get the existence of  positive $h_0$ and $C$ , such that for $h\in
(0,h_0]$ the circle $\partial B( \hat \Lambda^i(h), h^{k_i+\hat q_i})$
is in $\rho(\mathcal A_h)$ and the resolvent, which is given by 
$$
 (\A_h -\lambda)^{-1} = \RR(h,\lambda) (I+\Eg(h,\lambda))^{-1}\,,
$$
 consequently  satisfies there
\begin{displaymath}
  \|(\A_h -\lambda)^{-1}\|\leq C\, h^{-(k_i +\hat{q}_i)}\,.
\end{displaymath}
Hence, we obtain the
following result:
\begin{proposition} Let $q\in (0,\frac 16)$ and for $i=1,2\,$,  $\hat q_i = (\frac 23)^i  q\,$, $k_i =\frac 23 + ( \frac 23) ^{i}$. 
Under the assumptions of Theorem \ref{thm:interior}, for potentials of
type V1 (where $i=1$) and the assumptions of Theorem~\ref{thm:nonsmooth}, for potentials of
type V2 (where $i=2$) 
there exist positive constants $C$ and $h_0$ such that, for all $h\in
(0,h_0]$, 
\begin{equation}
    \label{eq:240}
  \sup_{\lambda\in\partial B(\hat \Lambda^i(h) ,h^{k_i+\hat{q}_i})}\|(\A_h -\lambda)^{-1}\|\leq C\, h^{-(k_i+\hat{q}_i)}\,.
\end{equation}
\end{proposition}
 We can now prove the upper bound for the spectrum. 
\begin{proposition}
Let $i\in \{1,2\}$ and suppose that $V$ is of type $Vi$.
There exist $h_0 >0$ and, for $h\in (0,h_0]$  an eigenvalue $\lambda
\in \sigma(\A_h)$ satisfying 
\begin{equation}  
\label{eq:5bis} 
\lambda- \hat  \Lambda^i (h)   = o(h^{k_i}) \quad \text{as }h\to0\,.
\end{equation}
\end{proposition}
\begin{proof}
Let $U^1$ be given by \eqref{eq:45} and $U^2$ by \eqref{eq:55}
and let $f_i=(\A_h-\hat{\Lambda}^i(h))U^i$. Clearly,
\begin{displaymath}
  (\A_h-\lambda)U^i= f_i + (\hat{\Lambda}^i(h) -\lambda)U^i\,.
\end{displaymath}
Hence, for $\lambda \in \pa B (\hat \Lambda^i (h), h^{k_i +\hat q_i})$,  we can write
\begin{displaymath} 
  \langle U^i, (\A_h -\lambda)^{-1}U^i\rangle =-\frac{1}{\lambda-\hat{\Lambda}^i(h)}[\langle U^i,U^i\rangle -\langle U^i, (\A_h-\lambda)^{-1}f_i\rangle] \,.
\end{displaymath}
By   \eqref{eq:240} and either \eqref{eq:47} for $i=1$  or \eqref{eq:57} for $i=2$, we then
obtain 
\begin{displaymath}
  \|(\A_h-\lambda)^{-1}f_i\|_2\leq \frac{C}{h^{k_i+\hat{q}_i}}\, \|f_i\|_2\leq C\, h^{m_i-\hat{q}_i}\| U^i\|\,,
\end{displaymath}
where $m_1=1/3$ and $m_2=2/9\,$. \\
Consequently, observing that $\hat{q}_i<m_i$ ($i=1,2$),
\begin{displaymath}
  \Big|\frac{1}{2\pi
    i}\oint_{\partial B(\Lambda^i(h) ,h^{k_i+\hat{q}_i}) }\langle U^i,
  (\A_h -\lambda)^{-1}f_i\rangle \,  d\lambda+\| U^i\|^2 \Big|\leq C\, h^{m_i-\hat{q}_i} \, \| U^i\|^2 \,.
\end{displaymath}
Hence there exists $h_0 >0$ such that, for $h\in (0,h_0]$, $(\A_h
-\lambda)^{-1}$ is not holomorphic in $B( \Lambda^i(h)  ,h^{k_i+\hat q_i})$ and the
proposition follows.
\end{proof}

The existence of an eigenvalue satisfying \eqref{eq:5bis} provides an
effective upper bound for $\inf \Re\sigma(\A_h)$. Together with the lower
bound \eqref{eq:10}, this completes the proof  of Theorems
\ref{thm:interior} and \ref{thm:nonsmooth}. 
\appendix

\section{Examples of potentials satisfying \eqref{eq:2}}
\label{app:example-potent}
We now derive some simple examples of potentials satisfying
\eqref{eq:2}, demonstrating that both types of potentials may exist.
The basic idea is again that the problem introduced in \eqref{eq:2}
exhibits some invariance to conformal mapping (see \cite{PS}).  Hence
starting form a problem defined on the square $\Box\,$, where the
solution of \eqref{eq:2} is a linear function, we can get from family
of conformal maps a corresponding family of potentials satisfying
\eqref{eq:2} in various domains, together with  \eqref{ass1},
\eqref{AssR1}, and \eqref{AssR2}.

 Let $\Box=(0,1)\times(0,1)\subset\C$ and $\Omega=f(\Box)$ where, for $w=u+iv$, $f$ is the conformal map
\begin{displaymath}
  f(w)= w +\delta\Big(\frac{1}{2}w^2+ \frac{\gamma}{3}w^3\Big)\,,
\end{displaymath}
in which $\delta>0$ and $\gamma\in\R$. 
Let further $f(w)=z=x+iy\in\Omega$, and set
$g=f^{-1}:\Omega\to \Box \,$, which clearly exists for sufficiently small
$\delta$. We may now set
\begin{displaymath}
  \partial\Omega_D^1 = \{ f(u)\,,\, u\in (0,1)\} \quad ; \quad  \partial\Omega_D^2=\{ f(u+i)\,,\, u\in(0,1)\}\,,
\end{displaymath}
and 
\begin{displaymath}
   \partial\Omega_N^1 = \{ f(iv) \,,\, v\in (0,1)\} \quad ; \quad  \partial\Omega_N^2=\{ f(1+iv)\,,\,v\in(0,1)\}\,.
\end{displaymath}
Let $g=U+iV$ (or $U=\Re g$ and $V=\Im g$).  Clearly, $V\equiv0$ on $\partial\Omega_D^1$
and $V\equiv1$ on $\partial\Omega_D^2$.  Furthermore, $V$ is harmonic and since $f$
is conformal, we must have $\partial V/\partial\vartheta=0$ on $\partial\Omega_N$. It follows that
$V$ is a solution of \eqref{eq:2} with $C_1=0$ and $C_1=1$. Since $V$
is constant on each connected component of $\partial\Omega_D$ we have there,
using the Cauchy-Riemann equations satisfied by $g$,
 \begin{equation}
 |\partial V/\partial\vartheta |=|\nabla V|=|g^\prime| \,.
 \end{equation}
The same argument shows that for fixed $\gamma \in\R$, and for $\delta$ small
enough, Assumption~\ref{ass1} is also satisfied in $\overline{\Omega}$.
Consequently, we need to identify the location of
\begin{displaymath}
  \inf_{z\in\partial\Omega_D}|g^\prime(z)|=\inf_{0<u<1} \min \Big(\frac{
    1}{|f^\prime(u)|}\,,\, \frac{
    1}{|f^\prime(u+i)|}\Big)\,.
\end{displaymath}
It can be easily verified that whenever $-2<\gamma<0$ we have
\begin{displaymath}
  \sup_{0<u<1}|f^\prime(u)|=\Big|f^\prime\Big(-\frac{1}{2\gamma}\Big)\Big|= 1+
  \frac{\delta}{4\gamma}\,,
\end{displaymath}
 and for sufficiently small $\delta$ we have 
\begin{displaymath}
  \sup_{0<u<1}|f^\prime(u+i)|=\Big|f^\prime\Big(-\frac{1}{2\gamma}+i\Big)\Big|+\OO(\delta^2)= 1-
  \frac{\delta}{4\gamma}-\delta\gamma+\OO(\delta^2)\,.
\end{displaymath}
It follows that whenever $\gamma<-1/2$ the minimum of $|\partial V/\partial\vartheta|$ over $\partial\Omega_D$  is
obtained, for sufficiently small $\delta$ at an interior point (close to
$i-\frac{1}{2\gamma}$),  whereas for $0 >\gamma>-1/2$ the minimum is attained,
for sufficiently small $\delta$ at one of the corners.\\
 Note that
\begin{displaymath}
  f^{(3)}(z)= \delta \gamma\,,
\end{displaymath}
and hence, for $\gamma<-1/2$, the maximum of $|f^\prime|$ (or the minimum of
$|g^\prime|$) is
non-degenerate. \\
Hence, depending on the value of $\gamma$, we can either find $\delta$
  and a pair $(V,\Omega)$ for which (V1) is satisfied or find $\delta$ and a
  pair $(V,\Omega)$ for which (V2) is satisfied.\\

{\bf Acknowledgments:}\\
Y. Almog was partially supported by NSF Grant DMS-1613471.


\begin{thebibliography}{1}
\bibitem{AS}   {\sc  M. Abramowitz and I. A. Stegun.}
{\em Handbook of Mathematical Functions.}
Dover Publisher, New York, 1965.

\bibitem{ag82}
{\sc S.~Agmon}, {\em {Lectures on exponential decay of solutions of
  second-order elliptic equations: bounds on eigenfunctions of {$N$}-body
  {S}chr\"odinger operators}}, vol.~29 of {Mathematical Notes}, Princeton
  University Press, Princeton, NJ, 1982.

\bibitem{al08}
{\sc Y.~Almog.} {\em {The stability of the normal state of superconductors in
  the presence of electric currents}}. SIAM Journal on Mathematical Analysis,
  40 (2008), pp.~824--850.

\bibitem{aletal13}
{\sc Y.~Almog, B.~Helffer, and X.-B. Pan.} {\em {Superconductivity near the
  normal state in a half-plane under the action of a perpendicular electric
  current and an induced magnetic field}}. Trans. Amer. Math. Soc., 365 (2013),
  pp.~1183--1217.
  
\bibitem{aletal17}
{\sc Y.~Almog, B.~Helffer, and X.-B. Pan.} {\em Mixed
  normal-superconducting states in the presence of strong electric currents}.
  Arch. Ration. Mech. Anal., 223 (2017), pp.~419--462.

      \bibitem{AGH} {\sc Y.~Almog, D.~Grebenkov, and B.~Helffer.}  {\em
          On a Schr\"odinger operator with a purely imaginary potential
          in the semiclassical limit}. 
        ArXiv:1703.07733, (2017).
  
 
 \bibitem{AH} 		{\sc Y. Almog and B. Helffer.}
	\newblock {\em On the spectrum of non-selfadjoint Schr\"odinger
          operators with compact resolvent.} 
	\newblock Comm. in PDE  40 (8) (2015), pp. 1441--1466.  

\bibitem{alhe14}
{\sc  Y.~Almog and B.~Helffer.} {\em {Global stability of the normal state
  of superconductors in the presence of a strong electric current}}.
  Comm. Math. Phys., 330 (2014), pp.~1021--1094.
  
  \bibitem{Ahen}		{\sc Y. Almog and R. Henry.}
	\newblock {\em Spectral analysis of a complex Schr\"odinger operator in the semiclassical limit.}
	\newblock  SIAM J. Math. Anal. 44 (2016), pp. 2962--2993.

\bibitem{AsDa} {\sc A. Aslanyan and E.B. Davies.}
	\newblock Spectral instability for some Schr\"odinger operators.
	\newblock  Numer. Math. 85 (2000), no. 4, 525--552.
	

\bibitem{BHHR} 	{\sc 	K. Beauchard, B. Helffer, R. Henry, and L. Robbiano.}
	\newblock {\em Degenerate parabolic operators of Kolmogorov
          type with a geometric control condition.} 
	\newblock ESAIM: COCV 21 (2015), pp. 487--512.

\bibitem{da07}		{\sc E. B. Davies.  }
	\newblock {\em {Linear Operators and their Spectra}}, vol.~106 of
	{Cambridge Studies in Advanced Mathematics}, 
	\newblock Cambridge University Press, Cambridge, 2007.

\bibitem{Grebenkov07}{\sc 	D. S. Grebenkov. }
	\newblock {\em NMR Survey of Reflected Brownian Motion}. 
	\newblock Rev. Mod. Phys. 79 (2007), pp. 1077--1137.

\bibitem{GH} 	{\sc 	D. S. Grebenkov, B. Helffer. }
	\newblock {\em On spectral properties of the Bloch-Torrey operator in two dimensions.}
	\newblock  http://arxiv.org/abs/1608.01925.   SIAM J. Math. Anal 50, 622--676 (2018). 

\bibitem{Hel2} 		{\sc B. Helffer.}
	\newblock {\em On pseudo-spectral problems related to a time
          dependent model in superconductivity with electric current.}
	\newblock Confluentes Math. 3 (2) (2011), pp.~237-251.
	
	\bibitem{Helbook} {\sc B. Helffer.}
	\newblock Spectral theory and its applications.
	\newblock Cambridge University Press (2013).
	
	\bibitem{HeSj} {\sc B. Helffer and J. Sj\"ostrand.}
	\newblock {\em From resolvent bounds to semigroup bounds.}
	\newblock ArXiv:1001.4171v1 (23 Jan 2010).

\bibitem{Hen} 	{\sc 	R. Henry.}
	\newblock {\em On the semi-classical analysis of Schr\"odinger operators with purely imaginary electric potentials in a bounded domain.}
	\newblock ArXiv:1405.6183 (2014).  

\bibitem{ivko84}
{\sc B.~I. Ivlev and N.~B. Kopnin.} {\em Electric currents and resistive states
  in thin superconductors}, Advances in Physics, 33 (1984), pp.~47--114.

\bibitem{PS} {\sc G. Polya and G. Szeg\"o.}
\newblock Isoperimetric Inequalities in Mathematical Physics.
\newblock Princeton University Press, Princeton, New Jersey (1951).

\bibitem{sh03}		{\sc A. A. Shkalikov.}
	\newblock {\em Spectral portraits of the Orr-Sommerfeld operator at large Reynolds numbers}, 
	\newblock Sovrem. Mat. Fundam. Napravl. 3 (2003), pp.~89--112.

\bibitem{Sj} 		{\sc J. Sj\"ostrand.}
	\newblock {\em Resolvent estimates for non-selfadjoint operators via semigroups.}
	\newblock {\em Around the research of Vladimir Maz'ya. III.}
	\newblock Int. Math. Ser.  13, Springer, New York (2010), pp. 359-384.

\bibitem{Stoller91}	{\sc S. D. Stoller, W. Happer, and F. J. Dyson.}
	\newblock {\em Transverse spin relaxation in inhomogeneous magnetic fields.}
	\newblock Phys. Rev. A 44 (1991), pp. 7459-7477.

\bibitem{deSwiet94}	{\sc T. M. de Swiet and P. N. Sen.}
	\newblock {\em Decay of nuclear magnetization by bounded
          diffusion in a constant field gradient.} 
	\newblock J. Chem. Phys. 100 (1994), pp. 5597-5604.

 	
	
\end{thebibliography}
\def\cprime{$'$}

\end{document}